\documentclass[reprint,amsfonts, amssymb, amsmath,  showkeys, nofootinbib,pra, superscriptaddress, twocolumn,longbibliography]{revtex4-2}

\usepackage{float}
\makeatletter
\let\newfloat\newfloat@ltx
\makeatother
\usepackage[english]{babel}
\usepackage[utf8]{inputenc}
\usepackage{graphics}
\usepackage{selinput}
\usepackage{enumerate}

\usepackage{braket}
\usepackage{amsthm}
\usepackage{mathtools}
\usepackage{physics}
\usepackage[dvipsnames]{xcolor}
\usepackage{graphicx}
\usepackage[left=16mm,right=16mm,top=35mm,columnsep=15pt]{geometry} 
\usepackage{adjustbox}
\usepackage{placeins}
\usepackage[T1]{fontenc}
\usepackage{lipsum}
\usepackage{csquotes}
\usepackage{mathbbol} 
\usepackage[linesnumbered,ruled,vlined]{algorithm2e}
\SetKwInput{kwInit}{Init}

\def\P{\mathbb{P}}

\def\HC{\mathcal{H}}

\def\LC{\mathcal{L}}

\usepackage[makeroom]{cancel}
\usepackage[colorlinks=true,citecolor=blue,linkcolor=magenta]{hyperref}

\usepackage{tikz}
\tikzset{every picture/.style=remember picture}

\newcommand{\poly}{\operatorname{poly}}

\newcommand{\AC}{\mathcal{A}}

\newcommand{\DC}{\mathcal{D}}

\newcommand{\IC}{\mathcal{I}}
\newcommand{\MC}{\mathcal{M}}

\newcommand{\OC}{\mathcal{O}}
\newcommand{\PC}{\mathcal{P}}

\newcommand{\SC}{\mathcal{S}}

\newcommand{\Var}{{\rm Var}}

\renewcommand{\geq}{\geqslant}
\renewcommand{\leq}{\leqslant}

\renewcommand{\Re}{\text{Re}}

\newcommand{\rhot}{\widetilde{\rho}}

\renewcommand{\vec}[1]{\boldsymbol{#1}}  
\newcommand{\ot}{\otimes}

\newcommand{\bs}{\textsf{BS}}

\newcommand{\thv}{\vec{\theta}}
 
\newcommand{\alv}{\vec{\alpha} }

\def\be{\begin{equation}}
\def\ee{\end{equation}}
\def\bs{\begin{split}}
\def\e{\end{split}}
\def\ba{\begin{eqnarray}}
\def\bea{\begin{eqnarray}}

\def\tea{\end{eqnarray}}
\def\ea{\end{eqnarray}}
\def\eea{\end{eqnarray}}

\def\P{\mathbb{P}}

\def\P{\mathbb{P}}

\newtheorem{theorem}{Theorem}
\newtheorem{lemma}{Lemma}
\newtheorem{corollary}{Corollary}

\newtheorem{proposition}{Proposition}
\newtheorem*{proposition*}{Proposition}

\newtheorem{definition}{Definition}

\usepackage{amssymb}
\usepackage{dsfont}

\usepackage[normalem]{ulem}
\usepackage{hyperref}

\newcommand{\rw}{\boldsymbol{\Delta}_N} 
\newcommand{\povm}{\MC} 
\newcommand{\pov}{M} 
\newcommand{\povv}{e} 
\renewcommand{\oc}{\lambda} 
\renewcommand{\rhot}{\rho(\alv)} 

\newcommand{\rand}{\widehat{\ell}_{\rm fixed}}

\usepackage{ulem}

\usepackage{minitoc}
\usepackage{tocloft}

\makeatletter
\renewcommand{\part}[1]{%
  \refstepcounter{part}
  \addcontentsline{toc}{part}{#1}
  \mtcaddpart[part]{#1}
}
\makeatother

\makeatletter
\renewcommand\onecolumngrid{
\do@columngrid{one}{\@ne}
\def\set@footnotewidth{\onecolumngrid}
\def\footnoterule{\kern-6pt\hrule width 1.5in\kern6pt}
}
\renewcommand\twocolumngrid{
        \def\footnoterule{
        \dimen@\skip\footins\divide\dimen@\thr@@
        \kern-\dimen@\hrule width.5in\kern\dimen@}
        \do@columngrid{mlt}{\tw@}
}
\makeatother

\begin{document}
\doparttoc 
\faketableofcontents 
\part{}

\title{Pitfalls when tackling the exponential concentration of parameterized quantum models}

\author{Reyhaneh Aghaei Saem}
\affiliation{Institute of Physics, École Polytechnique Fédérale de Lausanne (EPFL), Lausanne, Switzerland}
\affiliation{Centre for Quantum Science and Engineering, Ecole Polytechnique F\'{e}d\'{e}rale de Lausanne (EPFL), Lausanne, Switzerland}

\author{Behrang Tafreshi}
\affiliation{Institute of Physics, École Polytechnique Fédérale de Lausanne (EPFL), Lausanne, Switzerland}

\author{{Zo\"{e}} Holmes}
\affiliation{Institute of Physics, École Polytechnique Fédérale de Lausanne (EPFL), Lausanne, Switzerland}
\affiliation{Centre for Quantum Science and Engineering, Ecole Polytechnique F\'{e}d\'{e}rale de Lausanne (EPFL), Lausanne, Switzerland}

\author{Supanut Thanasilp} 
\email[Correspondence to: ]{supanut.thanasilp@gmail.com}
\affiliation{Institute of Physics, École Polytechnique Fédérale de Lausanne (EPFL), Lausanne, Switzerland}
\affiliation{Centre for Quantum Science and Engineering, Ecole Polytechnique F\'{e}d\'{e}rale de Lausanne (EPFL), Lausanne, Switzerland}
\affiliation{Chula Intelligent and Complex Systems Center of Excellence, Department of Physics, Faculty of Science, Chulalongkorn University, Bangkok, Thailand}

\date{\today}

\begin{abstract}
Identifying scalable circuit architectures remains a central challenge in variational quantum computing and quantum machine learning. Many approaches have been proposed to mitigate or avoid the barren plateau phenomenon or, more broadly, exponential concentration. However, due to the intricate interplay between quantum measurements and classical post-processing, we argue these techniques often fail to circumvent concentration effects in practice. Here, by analyzing concentration at the level of measurement outcome probabilities and leveraging tools from hypothesis testing, we develop a practical framework for diagnosing whether a parameterized quantum model is inhibited by exponential concentration. Applying this framework, we argue that several widely used methods (including quantum natural gradient, sample-based optimization, and certain neural-network-inspired initializations) do not overcome exponential concentration with finite measurement budgets, though they may still aid training in other ways.
\end{abstract}

\maketitle

\section{Introduction}

If you are reading this sentence you are no doubt already well aware that Variational Quantum Algorithms (VQAs) and Quantum Machine Learning (QML)~\cite{cerezo2020variationalreview,bharti2021noisy,abbas2023quantum,mcardle2020quantum} have attracted much attention in recent years.
Core to these models are parameterized quantum circuits, which take the form of trainable circuits in the case of quantum neural networks~\cite{mcclean2016theory,perez2020data,mitarai2018quantum} and also appear when encoding input data in the case of quantum kernel-based models~\cite{schuld2021supervised, havlivcek2019supervised, thanasilp2022exponential, kubler2021inductive,  huang2021power, gentinetta2022complexity,  gan2023unified}. On the one hand, these algorithms are highly versatile with the potential to solve a broad range of scientific problems. On the other hand, they are largely heuristic and, due to the unavailability of high-quality large-scale quantum hardware, their scalability is subject to debate~\cite{zimboras2025myths}. Key scalability challenges include poor local minima~\cite{anschuetz2022quantum,larocca2021theory}, classical simulability~\cite{schreiber2022classical, sweke2023potential,sahebi2025dequantization,rudolph2025pauli,angrisani2024classically,shin2024dequantising,goh2023lie,kerenidis2021classical,schuster2024polynomial,fontana2023classical} and Barren Plateaus (BPs)~\cite{larocca2024review,mcclean2018barren,fontana2023theadjoint,ragone2023unified,cerezo2023does} (or more generally exponential concentration~\cite{huang2021power, thanasilp2022exponential, kubler2021inductive, suzuki2023effect, xiong2023fundamental, xiong2025role,shaydulin2021importance}).

In the presence of BPs, the vast majority of the loss landscape becomes exponentially flat (in the number of qubits, $n$) and concentrates around a certain value~\cite{arrasmith2021equivalence}. Gaining reliable information about the loss values on these flat regions demands exponential resources~\cite{arrasmith2020effect}. Since in practice only polynomial measurement shots can be used (at least in the asymptotic limit), this renders the landscape effectively untrainable at most points in the landscape.

A central objective in the pursuit of scalable quantum models is to design circuit architectures and optimization strategies that are resilient to exponential concentration.
An increasing number of proposals aim to circumvent exponential concentration~\cite{larocca2022group,schatzki2022theoretical,cerezo2020cost,larocca2021diagnosing, letcher2023tight,basheer2022alternating, napp2022quantifying, zhang2023absence, pesah2020absence, monbroussou2023trainability,cherrat2023quantum, diaz2023showcasing,mele2024noise,deshpande2024dynamic,srimahajariyapong2025connecting, mhiri2025unifying,mele2022avoiding,puig2024variational, chang2024latent, wang2023trainability, park2023hamiltonian,park2024hardware, zhang2022escaping, tangpanitanon2020expressibility, srimahajariyapong2025connecting, shi2024avoiding,cao2024exploiting, grant2019initialization, friedrich2022avoiding,miao2024neural,rad2022surviving, Fa_lde_2023, kashif2023resqnets}. These include methods that explicitly claim to avoid or mitigate barren plateaus, such as specialized circuit architectures~\cite{larocca2022group,schatzki2022theoretical,cerezo2020cost,larocca2021diagnosing, letcher2023tight,basheer2022alternating, napp2022quantifying, zhang2023absence, pesah2020absence, monbroussou2023trainability,cherrat2023quantum, diaz2023showcasing,mele2024noise,deshpande2024dynamic,srimahajariyapong2025connecting}, alternative initialization schemes~\cite{mhiri2025unifying,mele2022avoiding,puig2024variational, chang2024latent, wang2023trainability, park2023hamiltonian,park2024hardware, zhang2022escaping, tangpanitanon2020expressibility, srimahajariyapong2025connecting, shi2024avoiding,cao2024exploiting, grant2019initialization}, and modified training strategies~\cite{friedrich2022avoiding,miao2024neural,rad2022surviving, Fa_lde_2023, kashif2023resqnets}. Other approaches, such as sample-based optimization~\cite{barkoutsos2020improving} or quantum natural gradient descent~\cite{stokes2020quantum}, have also been informally discussed as potential remedies~\cite{larocca2024review}.

Here we argue that, given the subtle interplay between quantum measurements and classical processing strategies, it is crucial to carefully evaluate whether proposed approaches truly mitigate barren plateaus in practice. A common diagnostic involves analyzing the scaling of the loss variance or, equivalently, the variance of the loss gradients. However, relying solely on variance scaling can be misleading. For instance, one might appear to suppress barren plateaus simply by multiplying the loss function by an exponentially large prefactor. Yet this superficial remedy clearly offers only the illusion of improvement. Consequently, while examining loss variances is currently the standard practice~\cite{mcclean2018barren, larocca2024review,marrero2020entanglement,sharma2020trainability,patti2020entanglement,wang2020noise,arrasmith2021equivalence,larocca2021diagnosing,holmes2021connecting, cerezo2020cost,khatri2019quantum,rudolph2023trainability,kieferova2021quantum,thanaslip2021subtleties,holmes2021barren,martin2022barren, fontana2023theadjoint,ragone2023unified, thanasilp2022exponential, letcher2023tight, anschuetz2024unified, xiong2023fundamental, crognaletti2024estimates, mao2023barren, xiong2025role, suzuki2023effect, cerezo2023does, arrasmith2020effect} to analyze exponential concentration, it can be insufficient to validate strategies for circumventing the issue.

\begin{figure}
\centerline{\includegraphics[width=\linewidth]{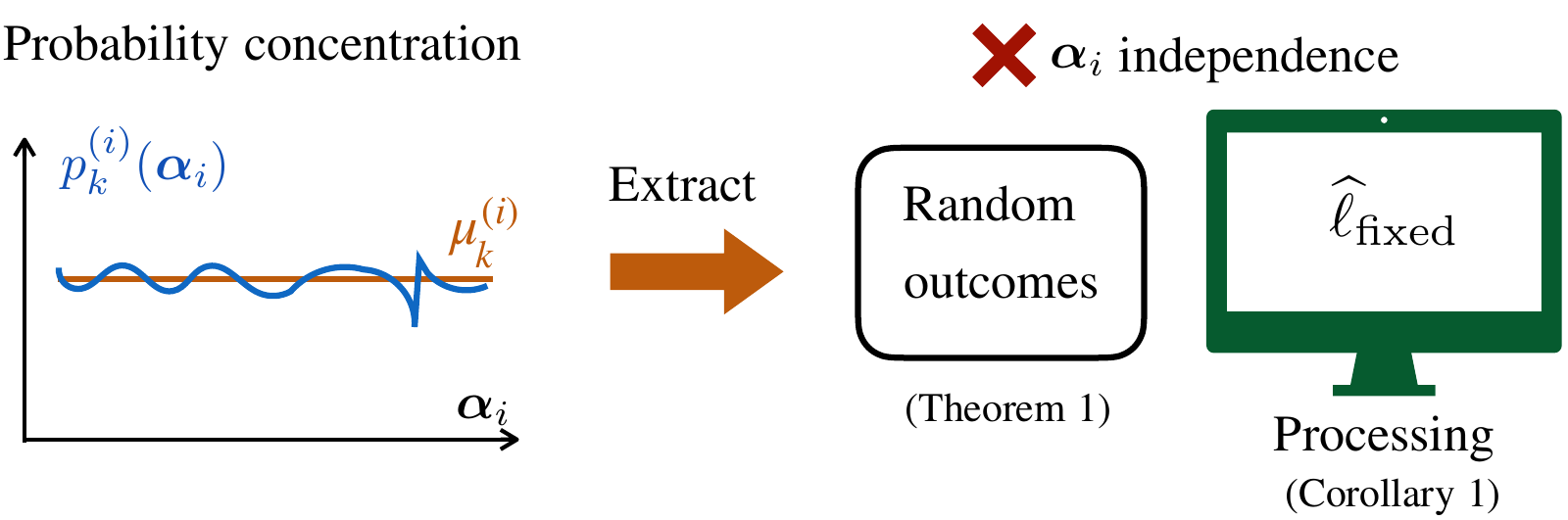}}
\caption{\textbf{Schematic representation of results.}  If measurement outcomes (at the level of probabilities) are exponentially concentrated then, with high probability, they contain no information about the trainable parameters and/or data inputs in the sense that they are indistinguishable from samples drawn from a variable-independent probability distribution (Theorem~\ref{thm:main-indistinguishable}). It follows that further post-processing these measurement outcomes results in parameter-independent and/or data-independent random quantities (Corollary~\ref{coro:no-post-process}).}
\label{fig:consequence}
\end{figure}

\begin{figure*}
\centerline{\includegraphics[width=0.95\linewidth]{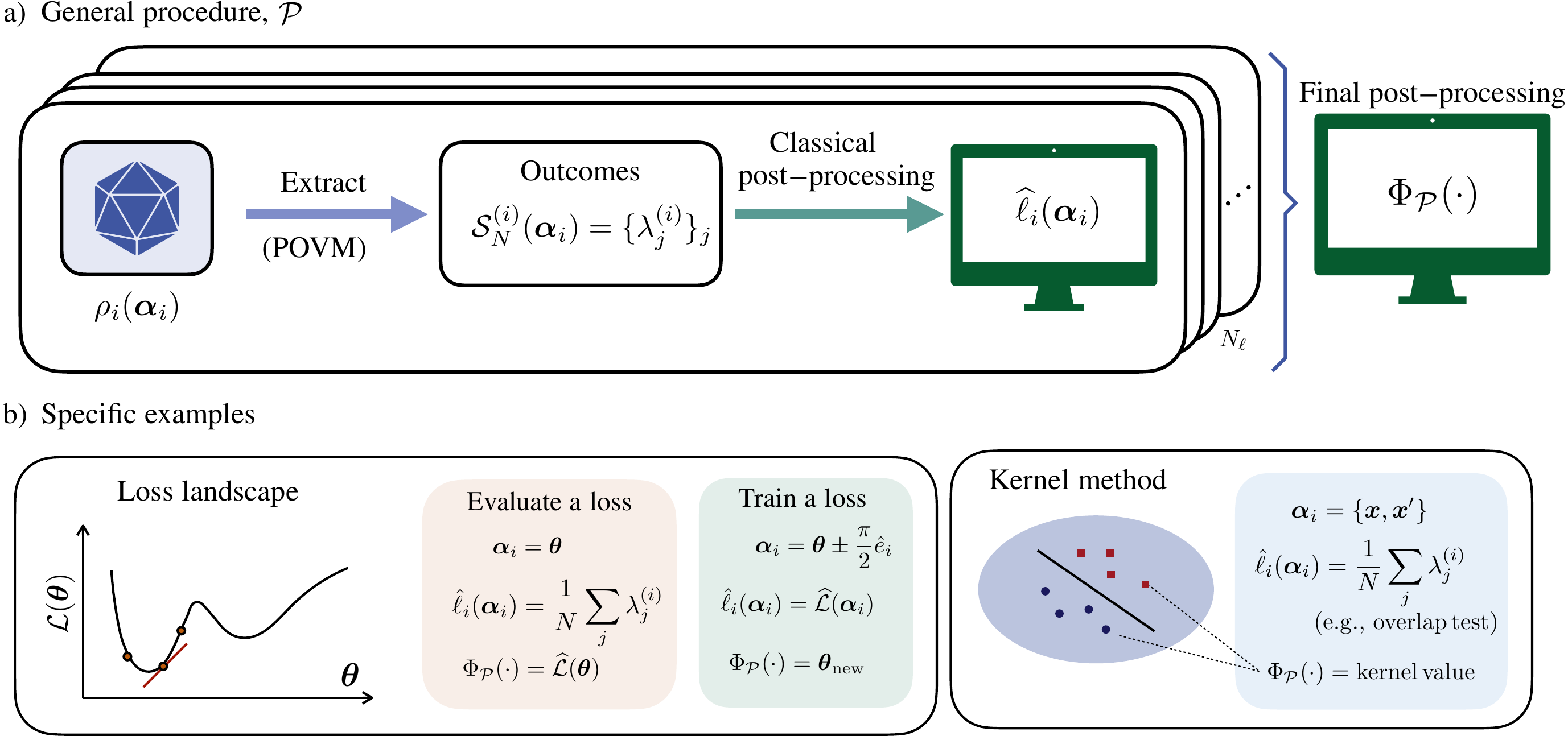}}
\caption{\textbf{Framework}. Panel a) illustrates the description of the general procedure (as described in Section~\ref{sec:procedure}). The procedure involves extracting information from variable-dependent quantum states through some quantum measurements. These measurement outcomes are then processed to estimate some variable-dependent quantities, which can in turn be all processed together. Panel~b) illustrates specific examples that fall within the scope of our general procedure, enabled by appropriately specifying the relevant quantities— i.e., the variables, measured observables, and post-processing functions. In particular, a loss evaluation is realized by interpreting a physical quantity $\widehat{\ell}(\alv_i)$ as an individual term in the overall loss, evaluated at a specific parameter setting $\alv_i = \thv$, while training the loss builds upon this by letting each physical quantity $\widehat{\ell}(\alv_i)$ represent the estimated full loss at some shifted parameter values, $\alv_i = \thv \pm \frac{\pi}{2} \hat{e}_i$, that allow one to compute the gradient. On the other hand, for non-variational models such as quantum kernel methods, the physical quantity of interest could be a fidelity kernel, where the variable is a data pair $\alv_i = \{\vec{x}, \vec{x}'\}$, which can be estimated using, for example, an overlap test.}
\label{fig:summary}
\end{figure*}

In this work, we provide systematic guidelines to address the following question:
\textit{Given a procedure that claims to circumvent exponential concentration, how can we determine whether it actually works in practice?}
To answer this, we begin with a simple yet important observation: many procedures used in variational quantum computing and QML involve processing a set of parameter-dependent quantities $\{ \ell_i(\boldsymbol{\alpha}_i) \}_i$ obtained from a quantum device. Each $ \ell_i(\boldsymbol{\alpha}_i)$ is estimated by performing measurements on a parameterized quantum state $\rho_i(\boldsymbol{\alpha}_i)$, followed by classical post-processing of the measurement outcomes. From this observation, one can see that it is more appropriate to study concentration at the level of \textit{outcome probabilities}, rather than at the level of expectation values (as is typically done in the literature~\cite{mcclean2018barren, larocca2024review,marrero2020entanglement,sharma2020trainability,patti2020entanglement,wang2020noise,arrasmith2021equivalence,larocca2021diagnosing,holmes2021connecting, cerezo2020cost,khatri2019quantum,rudolph2023trainability,kieferova2021quantum,thanaslip2021subtleties,holmes2021barren,martin2022barren, fontana2023theadjoint,ragone2023unified, thanasilp2022exponential, letcher2023tight, anschuetz2024unified, xiong2023fundamental, crognaletti2024estimates, mao2023barren, xiong2025role, suzuki2023effect, cerezo2023does, arrasmith2020effect}). This shift in perspective allows us to directly account for the effect of finite measurement shots in our analysis.

Specifically, we argue that standard parameterized quantum models implicitly use quantum measurements defined by a Positive Operator-Valued Measure (POVM) with at most polynomially many elements. We then show that if the outcome probabilities of such POVMs are exponentially concentrated, then (with high probability) the outcomes are statistically indistinguishable from samples drawn from a fixed, variable-independent distribution (Theorem~\ref{thm:main-indistinguishable}), as illustrated in Fig.~\ref{fig:consequence}. It follows that the resulting measurements carry no meaningful information about the underlying variables and so cannot be used for meaningful learning/training. Crucially, this statistical indistinguishability cannot be overcome through post-processing of the measurement outcomes (Corollary~\ref{coro:no-post-process}). This observation debunks the earlier foolish (if strawman) example of attempting to avoid BPs by multiplying expectation values by an exponentially large prefactor.

Based on these theoretical results, we provide a practical step-by-step guideline for identifying whether a given procedure can circumvent exponential concentration. This guideline is then used to argue that a number of parameterized quantum models (natural gradient descents~\cite{stokes2020quantum}, sample-based CVaR optimization~\cite{barkoutsos2020improving}, agnostic classical neural network-assisted initialization~\cite{friedrich2022avoiding}, a rescaled gradient approach~\cite{Teo_2023}, and others such as~\cite{rad2022surviving, Fa_lde_2023, kashif2023resqnets}) cannot circumvent these barriers (though may provide alternative training benefits). We further emphasize that the guidelines we provide apply not only to optimization strategies for losses, but also extend to various non-variational QML models such as quantum kernel methods~\cite{huang2021power, thanasilp2022exponential, kubler2021inductive, suzuki2023effect, xiong2023fundamental, xiong2025role,shaydulin2021importance} and quantum reservoir-based models~\cite{xiong2023fundamental, xiong2025role}.

Lastly, as a by-product of our main theoretical results, we contribute to the fundamental understanding of quantum landscape theory by proving that directly training on BPs with vanilla gradient descent, and practical measurement shot budgets, results in a random walk on the landscape (Corollary~\ref{coro:random-walk}). While this has been mentioned in passing in the literature (for example, see Ref.~\cite{arrasmith2020effect}), it had not been proven as far we are aware. Here, by taking the post-processing in Corollary~\ref{coro:no-post-process} to be a gradient calculation, we show that the estimated loss gradients at each training step are, with high probability, statistically indistinguishable from random variables that carry no information about the landscape. Consequently, the entire training trajectory resembles a random walk.

\section{Framework}\label{sec:procedure}

Any parameterized quantum model, whether part of a variational or non-variational algorithm, consists of two essential components:
\begin{enumerate}
 \item Extracting information from a quantum device via measurements, and
\item Classical post-processing of the resulting data.
\end{enumerate}
To formalize this structure, we introduce a general procedure $\mathcal{P}$ (sketched in Fig.~\ref{fig:summary}) that underlies a wide range of parameterized quantum models.
For variational quantum algorithms, which are adaptive in the sense that later measurements depend on earlier outputs, our procedure $\mathcal{P}$ captures a single iteration of the algorithm. In contrast, for non-adaptive models such as quantum kernel methods~\cite{schuld2021supervised, havlivcek2019supervised, thanasilp2022exponential, kubler2021inductive,  huang2021power, gentinetta2022complexity,  gan2023unified} or quantum reservoir computing~\cite{xiong2025role,xiong2023fundamental,fujii2017harnessing,nakajima2019boosting,mujal2021opportunities,ghosh2020reconstructing,hu2023tackling}, $\mathcal{P}$ captures the entire algorithm.

Specifically, we consider a procedure $\mathcal{P}$ that computes a set of variable-dependent physical quantities $\{ \ell_i(\boldsymbol{\alpha}_i) \}_{i=1}^{N_\ell}$ where $\ell_i(\boldsymbol{\alpha}_i)$ is real valued and depends on some variables $\alv_i$. We suppose that the number of relevant quantities $N_\ell$ scales at most polynomially with the system size i.e., $N_\ell \in \OC(\poly(n))$. We note here that each individual $\alv_i$ is kept general for now. As will be seen in the examples below, it could represent either some variable evaluated at some specific values (e.g., trainable parameters evaluated at some specific point on the loss landscape), or a new \textit{type} of variable (e.g., classical input data).

Each $\ell_i(\boldsymbol{\alpha}_i)$ is then estimated by performing a POVM measurement $\mathcal{M}^{(i)} = \{ M^{(i)}_k \}_k$ on a quantum state $\rho_i(\boldsymbol{\alpha}_i)$, where the $M^{(i)}_k$ are positive operators satisfying $\sum_k M^{(i)}_k = \mathbb{1}$.
Thus, each $\ell_i(\boldsymbol{\alpha}_i)$ is associated with a specific POVM $\mathcal{M}^{(i)}$. It follows that the total procedure is associated with a set of $i = 1, \dots , N_\ell$ different POVMs (i.e., a set of sets of POVM operators). 

After performing the POVM measurement $N$ times on the quantum state $\rho_i(\alv_i)$ (by repeatedly preparing the state for each measurement), we obtain a set of outcomes $\SC_N^{(i)}(\alv_i) = \{\oc_{j}^{(i)} \}_{j=1}^N$ where the $j^{\rm th}$ measurement outcome $\lambda^{(i)}_j$ takes the label $m_k^{(i)}$ (chosen from a set of labels $\{ m^{(i)}_k\}_{k=1}^{|\povm^{(i)}|} $
associated with POVM elements $\{ \pov^{(i)}_k \}_{k=1}^{|\povm^{(i)}|}$) with the probability 
\begin{align}
    p^{(i)}_k(\alv_i) = \Tr[ \rho_i(\alv_i) \pov^{(i)}_k] \;\;.
\end{align}
Then, by applying some post-processing map $\Phi_i(\cdot)$ on the outcomes, we obtain the statistical estimate of $\ell_i(\alv_i)$ as
\begin{align}
    \widehat{\ell}_i (\alv_i) = \Phi_i\big( \SC_N^{(i)}(\alv_i)\big) \;.
\end{align}
Eventually, with all $N_\ell$ estimates obtained, the procedure $\PC$ processes these with another map $\Phi_{\PC}\big( \{ \widehat{\ell}_i(\alv_i)\}_{i=1}^{N_\ell}\big)$.

By specifying appropriate components of the procedure $\PC$, this framework covers a wide range of schemes used in VQAs and QML. We provide a non-exhaustive list of examples:

\medskip

\textbf{Evaluating a standard linear loss.} Consider a loss of the form $\LC(\thv) = \Tr[H \rho(\thv)]$ with a parametrized state $\rho(\thv) = U(\thv)\rho_0 U^{\dagger}(\thv)$, where $U(\thv)$ is a unitary operator parametrized by the parameters $\thv$ and $\rho_0$ is an initial state, and some observable $H = \sum_i c_i O_i$ where $\{c_i\}_i$ is a set of coefficients and $\{O_i\}_i$ is a set of operators (commonly, some Pauli strings or projectors)~\cite{cerezo2020variationalreview,bharti2021noisy,abbas2023quantum,mcardle2020quantum}. In our framework, each quantity $\ell_i(\alv_i)$ corresponds to $\Tr[O_i\rho(\thv)]$, with $\rho_i(\alv_i) = \rho(\thv)$ (and $\alv_i = \thv$) $; \;\forall i$. The associated POVM measurement is simply an eigenbasis measurement of $O_i$.
    
    Given a set of samples $\SC^{(i)}_N(\thv)$ where each outcome label $\oc^{(i)}_j $ corresponds to one of the eigenvalues of $O_i$, a statistical estimate $\widehat{\ell}_i(\thv)$ for $\ell_i(\thv)$ can be obtained from the empirical mean of $\SC^{(i)}_N(\thv)$. That is,
    \begin{equation}\label{eq:loss-individual-term}
        \widehat{\ell}_i(\thv) = \Phi_{\rm mean} \left( \SC^{(i)}_N(\thv) \right) = \frac{1}{N}\sum_{j = 1}^N \oc^{(i)}_j 
    \end{equation}
    The final post-processing map $\Phi_{\PC}(\cdot)$ simply takes the linear combination of these estimates with the relevant $c_i$ coefficient. That is, the procedure outputs
    \begin{equation}\label{eq:final-map-loss-evaluation}
        \Phi_{\PC}\big( \{ \widehat{\ell}_i(\alv_i)\}_{i=1}^{{N_\ell}}\big) = \sum_{i=1}^{N_\ell} c_i \widehat{\ell}_i(\thv) \, , 
    \end{equation}
    which is an unbiased estimate of $ \Tr[H \rho(\thv)]$.

\medskip

\textbf{Evaluating other loss functions.}
    The evaluation of non-linear loss functions can be captured in this framework by modifying the final map in Eq.~\eqref{eq:final-map-loss-evaluation} to incorporate a non-linear function and also different initial states. For example, consider a supervised learning task with a training dataset $\{ \rho_i, y_i\}_{i=1}^{N_{\ell}}$ such that an input quantum state $\rho_i$ is associated with a label $y_i$. Further consider a parametrized circuit $U(\thv)$ and a model of the form $\Tr[U(\thv)\rho_i U^{\dagger}(\thv) h]$ where $h$ is some Pauli operator. Again, in our framework we have that each $\ell_i(\alv_i)$ corresponds to the model prediction with $\alv_i = \thv \;;\;\forall i$. The Mean Square Error (MSE) can be estimated using a post-processing map $\Phi_{\PC}(\cdot)$ that implements 
    \begin{align}\label{eq:square-loss-evulation}
         \Phi_{\PC}\big( \{ \widehat{\ell}_i(\alv_i)\}_{i=1}^{N_\ell}\big) = \frac{1}{N_\ell}\sum_{i=1}^{N_\ell} (y_i - \widehat{\ell}_i (\thv))^2 \;.
    \end{align}

\medskip
    
\textbf{Gradient-based and non-gradient based training.} 
    The general procedure can cover any training strategy which requires loss values at different points on the landscape. This generally includes gradient-based methods~\cite{mitarai2018quantum,schuld2019evaluating} and gradient-free methods~\cite{singh2023benchmarking,arrasmith2020effect}. The main idea here is that each loss value can be estimated with the approach described above, and the procedure map $\Phi_{\PC}(\cdot)$ corresponds to processing these estimated losses and outputting the updated parameters according to the optimization method. For example, in the case of vanilla gradient descent we have that the final output $ \Phi_{\PC}\big( \{ \widehat{\ell}_i(\alv_i)\}_{i=1}^{N_\ell}\big)$ is the vector which represents the difference between the updated parameters for the next training iteration and the current ones. In particular, for circuits obeying the parameter shift rule~\cite{mitarai2018quantum,schuld2019evaluating}, the $k^{\rm th}$ component of the output vector $[\Phi_{\PC}\big( \{ \widehat{\ell}_i(\alv_i)\}_{i=1}^{N_\ell}\big)]_k$ with current parameter values $\thv$ (with the $k^{\rm th}$ component $\theta_k$) is expressed as
    \begin{align}\label{eq:loss-gradient-evaluation}
       \left[\Phi_{\PC}\big( \{ \widehat{\ell}_i(\alv_i)\}_{i=1}^{N_\ell}\big)\right]_k = \theta_k - \frac{\eta}{2}  & \left[\widehat{\LC}\left(\thv + \frac{\pi}{2}\hat{e}_k\right)\right. \nonumber \\ & \left. \;\;\;\; -  \widehat{\LC}\left(\thv  - \frac{\pi}{2}\hat{e}_k\right)\right] \;,
    \end{align}
    where $\eta$ is the learning rate. The term $\widehat{\LC}\left(\thv \pm \frac{\pi}{2}\hat{e}_k\right)$ denotes the estimated loss in Eq.~\eqref{eq:final-map-loss-evaluation} evaluated at the shifted parameter values $\thv \pm \frac{\pi}{2}\hat{e}_k$ where $\hat{e}_k$ is the unit vector in the direction of the $k^{\rm th}$ parameter component. That is, in the simple case where the Hamiltonian is a single Pauli operator $O$ and there are $N_{\rm p}$ parameters (i.e., $\thv = (\theta_1,\theta_2,...,\theta_{N_p})$), 
    we have that $\ell_i(\alv_i) = \Tr[\rho(\alv_i) O]$ with $\alv_i = \thv + \frac{\pi}{2}\hat{e}_i$ and $\alv_{i+N_{\rm p}} = \thv - \frac{\pi}{2}\hat{e}_i$ for all $i = 1, \dots, N_p$  such that $N_\ell = 2N_{\rm p}$.
    
\medskip
    
\textbf{Quantum natural gradient.} This optimization method explicitly takes into account the local geometric structure of the parametrized state space encoded in Quantum Geometric Tensor~(QGT)~\cite{stokes2020quantum}. In practice, the QGT can be approximated as a block-diagonal matrix, with each block computable on quantum hardware. In particular, if the parametrized circuit uses gates composed of Pauli gate generator, the relevant components of this approximation of the QGT can be obtained by measuring and processing expectation values of some relevant Pauli operators. Hence, in addition to the loss values for estimating the typical gradients (as in Eq.~\eqref{eq:loss-gradient-evaluation}), these extra quantities correspond to additional POVM measurements in the procedure (see Appendix~\ref{app:numerical-details} for more details). 

\medskip
    
\textbf{Sample-based training strategy.} We can also tackle an optimization strategy which does not explicitly construct an expectation value but relies directly on samples obtained from measurements. Essentially, in the optimization process, one has to process the samples in some sense (-- even if this process does not correspond to expectation/loss values), this is equivalent to choosing an appropriate processing map in our setting. To be more concrete, consider CVaR optimization proposed in Ref.~\cite{barkoutsos2020improving} to solve a binary optimization problem. For a given set of samples $\SC_N(\thv)$ (sorted in non-decreasing order), the strategy involves only $\lceil{\gamma N}\rceil$ samples with some hyperparameter $\gamma \in [0,1]$ to construct a loss component of the form 
    \begin{equation}
        \widehat{\ell}_i(\thv) = \Phi_{\rm CVaR}^{(\gamma)} \left( \SC^{(i)}_N(\thv) \right) = \frac{1}{\lceil{\gamma N}\rceil} \sum_{i=0}^{\lceil{\gamma N}\rceil} \oc^{(i)}_j \, .
    \end{equation}

\medskip

\textbf{Classical neural network-assisted initialization strategy.} This approach relies on initializing \emph{circuit} parameters with a classical neural network~\cite{friedrich2022avoiding,miao2024neural}. Upon training, the weights and biases of the classical neural network are directly adjusted (instead of circuit parameters). Through the lens of our general procedure, we can cover this strategy by simply identifying the trainable weights and biases as $\alv$. 

\medskip

\textbf{Non-variational QML models.} The general procedure also encompasses different families of QML model. The core thread here is that any QML algorithm, simply by definition, requires the model to interact with quantum computers one way or the other. One can then appropriately specify the components in the procedure to fit the QML model of interest. 
Consider quantum kernel methods as an example. Let $\vec{x}$ denote a classical input, $U(\vec{x})$ a data-embedding unitary and $\rho_0$ an initial state chosen to be the all-zero state $\rho_0 = |0\rangle\langle0|^{\otimes n}$. Then we have $\ell_i(\alv_i) = \Tr[|0\rangle\langle0|^{\otimes n} U^\dagger({\vec{x}'})U(\vec{x}) \rho_0 U^\dagger(\vec{x})U(\vec{x}')]$ for all $i$ such that $\alv_i = \alv = \{\vec{x},\vec{x'}\}$ and $\rho(\alv)= U^\dagger({\vec{x}'})U(\vec{x}) \rho_0 U^\dagger(\vec{x})U(\vec{x}')$. Using the overlap test which computes the fidelity by measuring the expectation value of the all-zero state for $\rho(\alv)$, the estimated fidelity kernel between $\vec{x}$ and $\vec{x'}$ can be expressed as in Eq.~\eqref{eq:loss-individual-term} with the POVM $\MC = \{ |0\rangle\langle 0|^{\otimes n},\mathbb{1}-|0\rangle\langle 0|^{\otimes n}\}$.

\medskip

\paragraph*{Polynomial POVMs in disguise.} A particularly observant reader might have noticed that in our specification of the procedure $\mathcal{P}$ we consider only ‘polynomial’ POVMs. Namely, we require that each of the POVMs $\povm^{(i)}$ contain at most polynomially (in $n$) many elements, $|\povm^{(i)}|\in \OC(\poly(n))$. It is then natural to wonder if standard procedures in the literature really do satisfy this constraint. After all, a computational basis measurement has exponentially many ($2^n$) different outcomes. Here we briefly address this potential point of confusion and explain how, while it is possible to consider parameterized quantum models that use POVMs with exponentially many outcomes (more on this in the Discussion), the standard ones currently used by the community in fact involve polynomial POVMs in disguise.

To illustrate this, consider the procedure of computing the standard VQA loss function $\LC(\thv) = \Tr[H \rho(\thv)]$, where $H$ is a diagonal Hamiltonian of the form $H = \sum_{i} Z_i Z_{i+1}$, with $Z_i$ a single-qubit Pauli-Z operator on the $i^{\rm th}$ qubit. From an initial brief look at the Hamiltonian, one might have thought that the relevant POVM elements are all the $n$-qubit computational basis states, hence scaling exponentially. However, this example does fall under the umbrella of our general procedure, as should be clear from looking back at the section on ‘evaluating a standard linear loss’ above. Namely, here the set of quantities to be estimated is $\{\ell_{i}(\thv) = \Tr[\rho(\thv) Z_iZ_{i+1}] \}_{i=1}^n$. Then, we can identify the associated POVM for the $i^{\rm th}$ term $\ell_{i}(\thv)$ as $\MC^{(i)} = \{|00\rangle\langle00|_{i,i+1} + |11\rangle\langle11|_{i,i+1},|01\rangle\langle01|_{i,i+1} + |10\rangle\langle10|_{i,i+1}\}$ where these are projectors onto $+1$ and $-1$ eigenvalues subspaces. This POVM crucially contains only two elements and can be achieved by measuring in the two-qubit computational basis.

A further observation to make from the previous example is that while each of the $\povm^{(i)}$ are different, they can all be simultaneously measured by measuring all the qubits in the computational basis. This is true more generally. Namely, the different $\ell_i(\alv_i)$ generally require different POVMs, i.e., $\povm^{(i)} \neq \povm^{(j)}$ for $i \neq j$. However, some of the $\povm^{(i)}$ may be measurable simultaneously.

As another example of a polynomial POVM in disguise, consider the same procedure with the global Pauli-$Z$ operator as the Hamiltonian, $H = \bigotimes_{i=1}^n Z_i$. At first glance, one might think the corresponding POVM consists of all computational basis projectors, suggesting exponentially many POVM elements. However, to estimate the expectation of $H$, it is not necessary to distinguish individual bit strings. Instead, what matters is determining the eigenvalue sector to which a given bit string belongs. In the language of our framework, the corresponding POVM is therefore $\MC = \{ \Pi_{+}, \Pi_{-} \}$, where $\Pi_{+}$ projects onto the $+1$ eigenspace (the span of computational basis states with even parity), and $\Pi_{-}$ projects onto the $-1$ eigenspace (the span of those with odd parity).

\begin{figure*}
\centerline{\includegraphics[width=\linewidth]{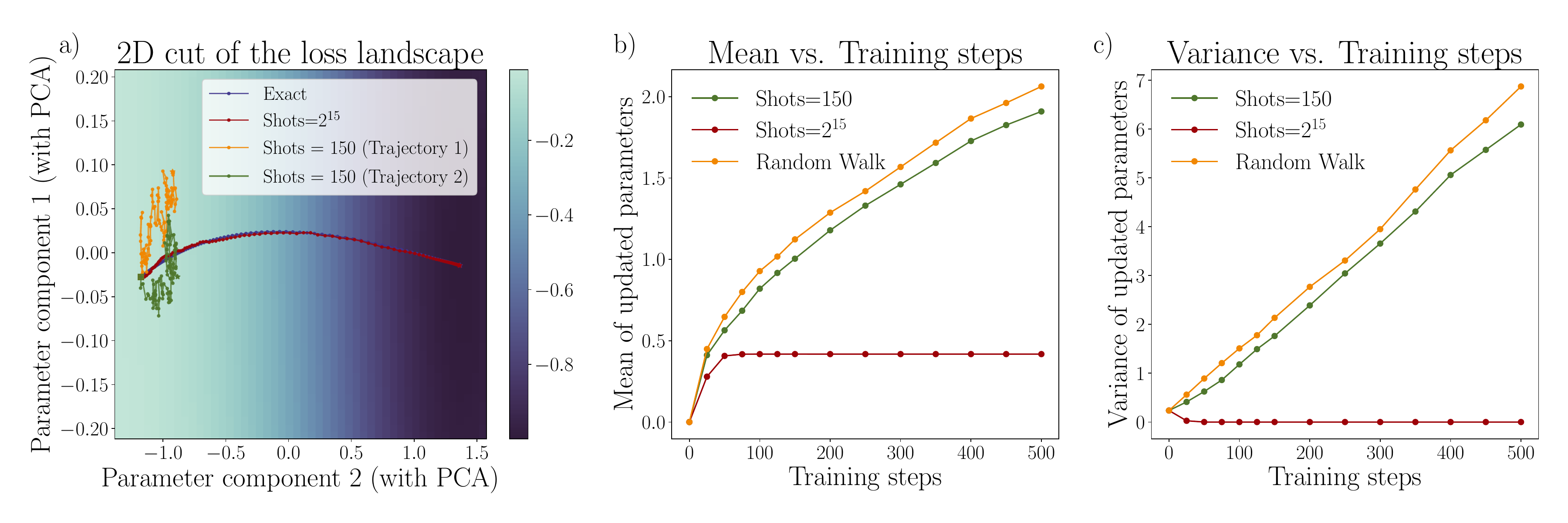}}
\caption{\textbf{Training on the BP landscape with $15$ qubits}. In panel a), training trajectories with different measurement shots are present on a 2D cut of the landscape (chosen via PCA analysis~\cite{rudolph2021orqviz}), showing random trajectories with polynomial measurement shots. In panels b) and c), the mean and variance of the updated parameters over different training trajectories are shown to align with the mean and variance of random walks. More specifically, we consider the quantity $\frac{1}{N_p} \|\thv^{(t)} - \thv^{(0)}\|_1$ where $\thv^{(t)}$ denotes the parameter vector at training iteration $t$ and $\boldsymbol{\theta}^{(0)}$ is the initial parameter vector. Panels b) and c) respectively display the average and variance of this quantity, computed over different parameter initializations, as functions of the training iteration. Here we have $n=15$, a layer of single X qubit rotations as an ansatz and measure the global Pauli-$Z$ observable.}
\label{fig:randomwalk}
\end{figure*}

\begin{figure*}
\centerline{\includegraphics[width=1.0\linewidth]{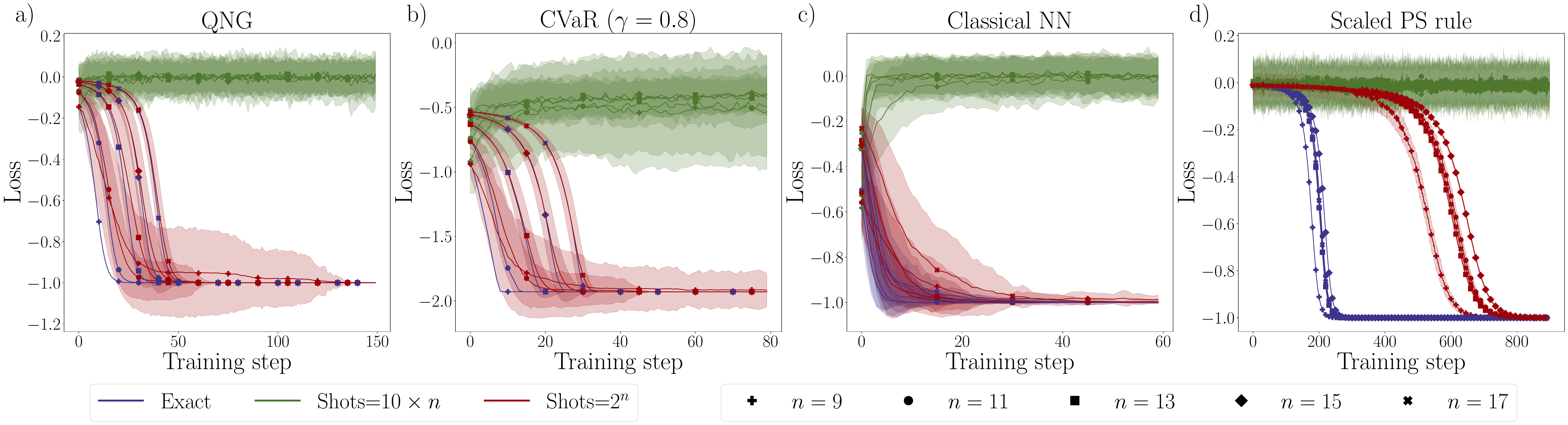}}
\caption{\textbf{Training curves.} We plot the loss as a function of training steps for different shot budgets, various optimization methods and system sizes: Panel a) shows quantum natural gradient descent;
Panel b) shows sample-based CVaR optimization;
Panel c) shows classical neural network–assisted initialization;
Panel d) shows the re-scaled parameter-shift rule. For each system size $n$, the training is performed under three shot regimes: infinite shots, $2^{n}$ shots, and $10 \times n$ shots.
The ansatz consists of a single layer of $X$-rotations on each qubit, and the observable is a linear combination of global $Z$ operators.}
\label{fig:post-process}
\end{figure*}

\section{Exponential concentration}

Having distilled the information processing strategies used by standard parameterized quantum models into the general procedure outlined above, it becomes clear that the scalability of $\mathcal{P}$ hinges on whether the measurement outcome probabilities, and thus the outcomes themselves, carry information about the underlying variables. This motivates a shift in focus: rather than analyzing exponential concentration at the level of the overall loss function (as is commonly done~\cite{mcclean2018barren, larocca2024review,marrero2020entanglement,sharma2020trainability,patti2020entanglement,wang2020noise,arrasmith2021equivalence,larocca2021diagnosing,holmes2021connecting, cerezo2020cost,khatri2019quantum,rudolph2023trainability,kieferova2021quantum,thanaslip2021subtleties,holmes2021barren,martin2022barren, fontana2023theadjoint,ragone2023unified, thanasilp2022exponential, letcher2023tight, anschuetz2024unified, xiong2023fundamental, crognaletti2024estimates, mao2023barren, xiong2025role, suzuki2023effect, cerezo2023does, arrasmith2020effect}), we instead consider concentration at the level of the POVM outcome probabilities associated with individual quantities.

For a given POVM $ \MC $, chosen from the set of POVMs $ \{ \MC^{(i)} \}_{i=1}^{N_\ell}$ associated with a procedure, together with the state $\rho(\alv)$, this alternative notion of exponential concentration can be formally defined as follows.
\begin{definition}[Outcome probability concentration]\phantomsection\label{def:probability-concentration-main}
Consider a parametrized $n$-qubit state $\rhot$ and a POVM $\povm = \{ M_k \}$. The POVM outcome probabilities concentrate with respect to the variable $\alv$ (which are drawn from some distribution $\DC$ over a certain domain) if for all $\pov_k \in  \povm$ we have
\begin{align}
    \Pr_{\alv\sim \DC}\left( |p_k(\alv) - \mu_k| \geq \delta \right) \leq \frac{\beta}{\delta^2} \;\; ;\;\; \beta \in \OC\left( \exp(-n)\right)
\end{align}
where $p_k(\alv) = \Tr[\rhot  \pov_k]$ is the probability of measuring an outcome associated with the element $\pov_k \in  \povm$ and $\mu_k$ is some concentration point independent of $\alv$. 
\end{definition}
First, we remark that the definition of concentration depends on the distribution $\DC$ and the domain of the variable $\alv$. The same holds for the concentration of expectation values commonly studied in the literature. When $\alv$ represents some form of input data, the distribution is determined by the data-generation process. On the other hand, when $\alv$ represents trainable parameters, most studies consider the uniform distribution over either state space or parameter space. More recently, progress has been made on studying alternative parameter initializations restricted to small subregions of the loss landscape~\cite{mhiri2025unifying, mele2022avoiding, puig2024variational, grant2019initialization}.

Crucially, we note that the underlying mechanisms that result in the concentration of the outcome probabilities are identical to those leading to the concentration of expectation values~\cite{larocca2024review}. This simply follows from the fact that the outcome probability can be expressed as the expectation of a POVM operator. 

Nevertheless, Definition~\ref{def:probability-concentration-main} allows us to pin down the practical consequence of the concentration and later provide the guidelines to determine the scalability of the given protocol (as presented in Section~\ref{sec:procedure}). 
That is, \textit{for any} state $\rho(\alv)$ the outcome probability distribution associated with the POVM $\mathcal{M}$, i.e., $\P_{\alv} = \big(p_1(\alv), p_2(\alv), ..., p_{|\povm|}(\alv)\big)$, is close to the fixed $\alv$-independent distribution $\P_{\rm fixed} = \big(\mu_1, \mu_2, ..., \mu_{|\MC|}\big)$. By using tools from hypothesis testing, we can rigorously show that if the number of POVM elements of these distributions scales at most polynomially in $n$, then with high probability these two probability distributions cannot be distinguished using only polynomial outcomes. Thus the following theorem holds.

\begin{theorem}[Indistinguishability from probability concentration, informal] \label{thm:main-indistinguishable}
Assume the exponential concentration of outcome probabilities as in Definition~\ref{def:probability-concentration-main} on a POVM set $\povm$ with  $|\povm|\in \OC(\poly(n))$. 
After polynomial measurement shots $N \in \OC(\poly(n))$, the obtained measurement samples $\SC_N(\alv)$ are with high probability at least $1-\delta$, such that $\delta \in \OC(\exp(-n))$, statistically indistinguishable from samples $\SC_{ N,{\rm fixed}}$ drawn from the fixed $\alv$-independent distribution $\P_{\rm fixed}$.
\end{theorem}

The direct consequence of Theorem~\ref{thm:main-indistinguishable} is stated below - namely, no classical post-processing removes this indistinguishability. 
\begin{corollary}[No post-processing, informal]\label{coro:no-post-process} 
Post-processing $\SC_N^{(i)}(\alv)$ with any arbitrary map $\Phi'(\cdot)$, with high probability at least $1-\delta'$, such that $\delta' \in \OC(\exp(-n))$, leads to an estimate that is statistically indistinguishable from an $\alv$-independent random variable 
\begin{align}
    \rand = \Phi'(\SC_{ N,{\rm fixed}}) \;\;,
\end{align}
where each outcome in $\SC_{ N,{\rm fixed}}$ is drawn from $\mathbb{P}_{\rm fixed}$.
\end{corollary}
\noindent We refer the readers to Appendix~\ref{app:theorem1-proof} for the formal statements, the proofs and the additional details.

Corollary~\ref{coro:no-post-process} leaves the choice of post-processing arbitrary and hence applies generally to any procedure $\PC$. To illustrate its practicality, we apply it to the training of a standard VQA loss with the BP landscape using a traditional gradient-based method.
\begin{corollary}[Random walk via gradient descent, informal]\label{coro:random-walk}
Consider a parametrized quantum state $\rho(\thv)$ that depends on some trainable parameters $\thv$ and a loss function of the form $\LC(\thv) = \sum_{i=1}^{N_\ell} c_i \ell_i(\thv)$ where each $\ell_i(\thv) = \Tr[\rho(\thv)O_i]$ with some parametrized state $\rho(\thv)$, some Pauli operator $O_i$ and $N_\ell \in\OC(\poly(n))$. Further, consider a state $\rho(\thv)$ that is generated from some parametrized circuit and suppose that the parameter shift rule is applied. Training the loss with vanilla gradient descent algorithm with a random initialization for polynomial training iterations using overall polynomial measurement shots is statistically indistinguishable from a random walk with high probability at least $1 - c$ for some $c \in \OC(\exp(-n))$. That is, for a given iteration, the updated parameters $\thv^{(\rm new)}$ for the next iteration are statistically indistinguishable from the update rule given as 
\begin{align}
    \thv^{\rm (new)} = \thv^{\rm (current)} + \rw \;\;,
\end{align}
where $\thv^{(\rm current)}$ are parameter values for the current iteration, and $\rw$ is a vector where each component is an instance of some parameter-independent random variable. 
\end{corollary}

This is further supported by the numerical results shown in Fig.~\ref{fig:randomwalk}, which illustrate training on a BP landscape under different shot budgets for a $15$-qubit system. The loss function is defined as the expectation value of a global Pauli-$Z$ operator, and the ansatz consists of a single layer of single-qubit $X$ rotations. This setup is known to suffer from the globality-induced barren plateau phenomenon that is observed for unstructured quantum circuits~\cite{cerezo2020cost}. Panel~a) presents a two-dimensional projection of training trajectories obtained under various shot conditions ($150$ shots, $2^{15}$ shots, and infinite shots). The projection is computed using Principal Component Analysis (PCA), as implemented in the \texttt{ORQVIZ} package~\cite{rudolph2021orqviz}. We observe that trajectories with $2^{15}$ shots and infinite shots converge toward a solution, while those with $150$ shots exhibit a random wandering behavior over the landscape. Panels~b) and c) show the scaling of the mean and variance of the parameter updates (averaged over multiple trajectories) as functions of training steps. 
Notably, both the mean and variance of the updates under the $150$-shot regime closely resemble those of a random walk.

We note that a similar conclusion is expected to hold for quantum natural gradient descent, albeit with some technical subtleties. While the estimated loss gradients remain indistinguishable, the situation is more subtle for the estimated QGT. In particular, certain elements of the QGT that capture the curvature of the quantum state generated by an \textit{early} part of the circuit may still be estimated efficiently. However, extracting curvature information from the deeper parts of the circuit becomes increasingly difficult. As these are combined with the information-less estimated gradients, we do not expect the overall \textit{estimated} quantum natural gradient to provide any meaningful direction on a featureless landscape. Nonetheless, the proof would be more complex. 

\section{Practical step-by-step guidelines}

We now outline a set of criteria to determine whether a given training or encoding procedure is vulnerable to scalability limitations due to exponential concentration.

\begin{enumerate}
    \item Given a procedure $\PC$, identify the quantities $\{\ell_i(\alv_i)\}$ which require information to be extracted from quantum computers. 
    \item For each quantity $\{\ell_i(\alv_i)\}$, identify the corresponding $\mathcal{M}^{(i)}$ and check whether $|\mathcal{M}^{(i)}|$, the number of POVM elements in $\mathcal{M}^{(i)}$, scales at most polynomially with system size $n$.
    
    Note that all barren plateau mitigation strategies we are aware of involve such polynomial-sized POVMs, even if some may initially appear to have exponentially many elements (see Section~\ref{sec:procedure} above on `Polynomial POVMs in disguise').
    \item Determine whether the outcome probabilities $p^{(i)}_k(\alv_i) = \Tr[ \rho_i(\alv_i) \pov^{(i)}_k]$ exponentially concentrate with respect to $\alv_i$ (as per Definition~\ref{def:probability-concentration-main}). 

    If this is the case, the procedure $\PC$ suffers from the concentration in the sense that the measurement outcomes, with probability exponentially close to 1, contain no information about the variables $\alv_i$.
\end{enumerate}

Based on the above guidelines, it follows that some proposals that were hoped to mitigate, or entirely avoid, the effects of exponential concentration may in fact still suffer from them. Namely, all of the examples discussed in Section~\ref{sec:procedure} potentially fall under this category. These include certain forms of natural gradient descent~\cite{stokes2020quantum}, sample-based CVaR optimization~\cite{barkoutsos2020improving}, agnostic classical neural network-assisted initialization~\cite{friedrich2022avoiding}, and a rescaled gradient approach~\cite{Teo_2023}. In particular, if a given parameterized quantum circuit exhibits exponential concentration under standard gradient-based training, switching to a more sophisticated cost evaluation strategy, such as natural gradient descent or neural network-assisted training, will not resolve the problem, because it does not address the root cause: exponential concentration at the level of outcome probabilities.

In Fig.~\ref{fig:post-process}, we further provide numerical simulations of the actual training on BP landscapes with these optimization methods using different shot budgets (polynomial shots of $10\times n$, exponential shots of $2^n$ and infinite shot limit) on various system sizes ($n=9, 11, 13, 15, 17$). Similar to Fig.~\ref{fig:randomwalk}, the numerical set-up is engineered such that the global-induced BP emerges (see further details in Appendix~\ref{app:numerical-details}). Here, one can clearly see the interplay with shot noise in the training process. While these methods indeed lead to successful optimization in the limit of $2^{n}$ measurement shots, they fail to move in any meaningful direction with relatively small $10 \times n$ shots.

It is important to note that while our results highlight that some methods hoped to avoid exponential concentration still suffer the effects of exponential concentration, this does not mean that they cannot be used in any way to boost scalability. For example, quantum natural gradients methods are naturally employed in conjunction with identity initialization or warm starting strategies~\cite{wierichs2020avoiding,Haug_2021}. Such deployments can indeed avoid exponential concentration. However, this is in virtue of using a different initialization strategy that changes the concentration properties of the circuit, not in virtue of switching from vanilla gradients to natural gradients. More generally, these optimization strategies could have their own other strengths that aid training or scalability. Indeed, some of them were originally proposed independently of the scalability issue. By considering again the quantum natural gradient approach as an example, the method takes into account the local curvature of the state space, resulting in generally faster training convergence. More heuristically, neural network initialization strategies have been seen to lead to faster training~\cite{friedrich2022avoiding, miao2024neural}.

Lastly, although our primary focus in this section has been on variational quantum algorithms, we stress that these guidelines also apply to non-variational QML models. In particular, scalability issues in quantum kernel methods~\cite{huang2021power, thanasilp2022exponential, kubler2021inductive, suzuki2023effect,shaydulin2021importance} and quantum reservoir approaches~\cite{xiong2023fundamental,xiong2025role}, arising from concentration over the input data, can also be diagnosed using our core guidelines. 

\medskip

\paragraph*{Subtlety regarding measure-first-estimate-later approaches.} 

The POVMs we consider in our guideline are essentially determined by the target quantities one aims to estimate. A natural question from attentive readers might be: what about approaches such as classical shadow techniques~\cite{huang2020predicting, elben2022randomized}, in which measurement data are collected first to build a classical representation of the quantum state—prior to knowing the specific observable of interest?

Our key argument is that the subsequent post-processing of the classical representation to evaluate a chosen observable is mathematically equivalent to implementing the corresponding POVM of that observable. Consequently, our concentration results apply at the level of the POVM itself, regardless of how it is realized in practice.

To build intuition, consider a simplified `measure first' scheme for estimating expectation values of Pauli-$Z$ operators. Here, measurements are performed solely in the computational basis (i.e., Pauli-$Z$ basis). 
After $N$ measurement outcomes, the classical representation takes the form
\begin{align}
    \hat{\rho}_{Z} = \frac{1}{N} \sum_{i=1}^{N} |b_i\rangle\langle b_i| \; ,
\end{align}
where $|b_i\rangle$ denotes the computational basis state corresponding to the $i^{\rm th}$ outcome bitstring. Suppose we are now interested in estimating the observable $Z_1 Z_2$. Using the classical representation $\hat{\rho}_Z$, the estimate is obtained by tracing out all irrelevant qubits, yielding $\Tr[Z_1 Z_2 \hat{\rho}_Z] = \frac{1}{N} \sum_{i=1}^{N} \langle b_i^{(12)}| Z_1Z_2 |b_i^{(12)}\rangle$ where $|b_i^{(12)}\rangle$ is the reduced two-qubit state on qubits $1$ and $2$. In other words, the combined process of i) constructing $\hat{\rho}_Z$ and ii) post-processing it once the observable is specified, is equivalent to implementing the POVM $\MC = \{ |00\rangle\langle00|_{12} + | 11\rangle\langle 11|_{12},  |01\rangle\langle01|_{12} + | 10\rangle\langle 10|_{12}\}$. 

The same reasoning extends to more general classical shadow techniques, such as those for estimating local Pauli observables. In these methods, the procedure can be understood as a two-step realization of POVMs: first, randomized local Pauli measurements are performed to construct a classical representation of the state; second, this representation is post-processed once the target observable is specified. Crucially, the randomness in the measurement step does not alter the underlying POVM associated with the observable. Together this approach provides an alternative physical implementation of the POVMs, equivalent to directly measuring term by term.

\medskip

\begin{figure}
\centerline{\includegraphics[width=\linewidth]{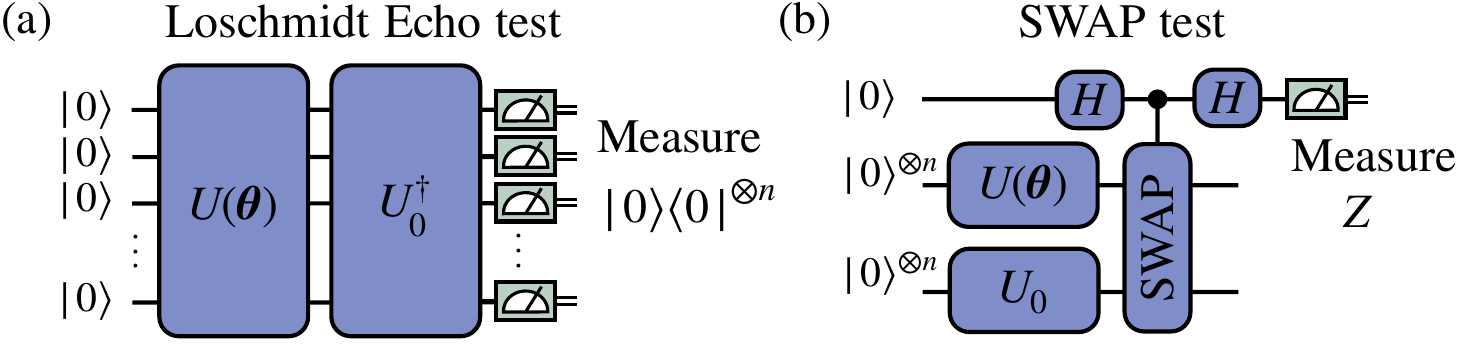}}
\caption{\textbf{Two different POVMs for estimating fidelity.} Panel~(a) illustrates the Loschmidt echo test, while panel~(b) schematically depicts the SWAP test.}
\label{fig:swap-overlap}
\end{figure}

\paragraph*{Subtlety regarding the choice of POVM.} 

While Theorem~\ref{thm:main-indistinguishable} is tied to a specific choice of POVM, independently of how that POVM is implemented, the same quantity may be estimated by different POVMs with different convergence rates.
Changing the measurement scheme does not affect the variance over the variable space and so we do not expect such a change in POVM to alleviate the effect of exponential concentration. Rather, the distinct statistical behaviors arising from employing different estimators affect the degree to which the problem suffers from shot noise.

To illustrate this point more explicitly, consider the task of learning a target state $|\phi\rangle = U_0 |0\rangle$ for some unitary $U_0$, with the loss function given by the infidelity $\ell(\thv) = 1 - |\langle \psi(\thv) | \phi\rangle|^2$. We examine two different POVM-based approaches to compute the fidelity, as illustrated in Fig.~\ref{fig:swap-overlap}. Suppose further that the variational state $|\psi(\thv)\rangle = U(\thv)|0\rangle$ forms a unitary 2-design. We will see that while both approaches suffer from concentration, the way statistical indistinguishability manifests differs between the two POVM choices.

The first approach is to employ the Loschmidt Echo test, which corresponds to estimating the expectation value of the global zero projector $|0\rangle\langle0|^{\otimes n}$ with the associated POVM $\MC = \{ |0\rangle\langle0|^{\otimes n} , \mathbb{1} - |0\rangle\langle0|^{\otimes n} \}$. By invoking the formal Theorem~\ref{thm:main-indistinguishable-formal} in Appendix~\ref{app:theorem1-proof}, the fixed distribution takes the form $\P_{\rm fixed} = (0,1)$. In other words, with high probability exponentially close to $1$, we always obtain non-all-zero bitstrings and the estimated fidelity is zero. 

Alternatively, one may employ the SWAP test. In this case, the fidelity is estimated from the expectation value of a single Pauli-$Z$ operator on the ancilla qubit, corresponding to the POVM $\widetilde{\MC} = \{ |0\rangle\langle 0 |, |1\rangle \langle 1 |\}$. Although this POVM also suffers from concentration effects, it yields a different fixed distribution, namely $\widetilde{\P}_{\rm fixed} = (1/2, 1/2)$. That is, with probability exponentially close to one, the measurement outcomes are equally likely to be $+1$ or $-1$.

For quantities corresponding to the expectation value of an observable (as in the case of fidelity), an alternative useful diagnostic, which can be used as a general rule of thumb, is to consider the ratio between the estimator variance and the variance of the quantity over the variable space.
Specifically, after performing $N$ measurement shots with $\rho(\alv)$, we define
\begin{align}
    \epsilon_N (\alv) = \frac{\Var_{\rho(\alv)}[\widehat{\ell}(\alv)]}{ N\,\Var_{\alv}[\ell(\alv)]} \;.
\end{align}
Intuitively, this ratio quantifies the relative precision arising from the interplay between the estimator’s shot noise and the intrinsic fluctuation of the quantity across the parameter space.
That is, it measures the ability of the estimator to resolve between different points on the landscape.
To have sufficient resolution, we require $ \epsilon_N (\alv) \lesssim 1 $. 
In the presence of exponential concentration, achieving such precision necessitates an exponentially large number of measurement shots, $N \in \Omega(\exp(n))$. 

We now apply this to the example of estimating the fidelity, $F_{\thv} = | \langle \psi(\thv) | \phi \rangle |^2$. In this case, the variance over the landscape remains exponentially small, independent of the chosen measurement scheme. However, the estimator variances differ. For the SWAP test, we have $\Var^{\rm(SWAP)}_{\rho(\theta)}[\widehat{\ell}(\thv)] = 1 - F_{\thv}^2$, whereas for the Loschmidt Echo test the variance is of the form $\Var^{\rm(LE)}_{\rho(\theta)}[\widehat{\ell}(\thv)] = F_{\thv}\,(1-F_{\thv})$. These distinct forms of variance result in different statistical behaviors of the corresponding estimators.

As another example, consider the task of estimating the purity of a state $\rho(\thv)$ given by $\Tr[\rho^2(\thv)]$. In this case, even in the absence of exponential concentration, estimating the purity from a single copy of $\rho(\thv)$ is inefficient, as it requires exponentially many distinct measurement settings followed by classical post-processing. In contrast, by employing two copies of the state together with an entangled measurement such as the SWAP test, the purity can be estimated efficiently. Thus, changing the POVM and allowing access to multiple copies provides an advantage~\cite{huang2021quantum}. Nevertheless, if we assume again $\rho(\thv)$ to form a $2$-design ensemble, the variance of the purity over the landscape (i.e., the dominator of the ratio) remains exponentially small regardless of the measurement schemes. Consequently, the advantage from employing the entangled POVM vanishes and the problem remains statistically indistinguishable.

Altogether, this raises a natural question: does there exist a POVM that could fundamentally alter the indistinguishability conclusion?
To the best of our knowledge, such a possibility appears unlikely. However, it cannot be definitively ruled out; for example, appropriately designed entangled POVMs for estimating quantities with strong global correlations may potentially help circumventing concentration. That said, the central principle of our message remains: rigorous analysis must be carried out with respect to the POVMs actually employed in practice.

\section{Discussion}
\label{sec:discussion}

Our work provides guidelines for assessing whether a given quantum model can avoid the effects of exponential concentration. These guidelines are grounded in a simple observation: any variational quantum procedure involves estimating certain quantities using quantum hardware followed by classical post-processing. As such, analysis should focus on concentration at the level of outcome probabilities rather than expectation values. Using tools from hypothesis testing, we show that for POVM measurements with a polynomial number of elements, exponential concentration of outcome probabilities implies that the measurement outcomes—together with any post-processing—contain no information about the variables (Theorem~\ref{thm:main-indistinguishable} and Corollary~\ref{coro:no-post-process}). 

We use these guidelines to critically re-examine several training strategies previously discussed as barren-plateau-free, including natural gradient descents~\cite{stokes2020quantum}, sample-based CVaR optimization~\cite{barkoutsos2020improving}, agnostic classical neural network-assisted initialization~\cite{friedrich2022avoiding}, and a rescaled gradient approach~\cite{Teo_2023}, with possible extensions to others such as~\cite{rad2022surviving, Fa_lde_2023, kashif2023resqnets}. More concretely, we argue that while these methods may well have other benefits, they are still prone to the effects of exponential concentration in a similar manner to more standard methods~\cite{mcclean2018barren, arrasmith2020effect,cerezo2020impact}. 

We now discuss two scenarios that are not captured by the general procedure and, consequently, fall outside the scope of the associated guidelines. The first scenario arises when the number of physical quantities scales exponentially with the system size, while the number of POVM elements in each POVM remains polynomial. Specifically, we consider the case where $N_{\ell} \in \Omega(\exp(n))$ and $|\povm^{(i)}| \in \mathcal{O}(\mathrm{poly}(n))$. A practical example of this setting is generative modeling with an explicit loss function, such as the Kullback–Leibler Divergence~(KLD), where the objective is to train a parameterized model to reproduce a target distribution over computational basis bit-strings. Here the explicit loss consists of exponentially many terms; each comparing the probability assigned to an individual bit-string. The POVM associated with estimating each such probability consists of only two elements (i.e., for a bit-string $\vec{x}$, the POVM is $\{ |\vec{x}\rangle\langle\vec{x}|,\; \mathbb{1} - |\vec{x}\rangle\langle\vec{x}| \}$). Moreover, all of these probabilities can, in principle, be estimated simultaneously via measurements in the computational basis. The subtlety, however, as studied in Ref.~\cite{rudolph2023trainability}, is that sampling a polynomial number of bit-strings only allows a polynomial number of terms in the loss to be assigned non-zero values. Since we cannot control which bit-strings are sampled, this results in a mismatch between the sampled loss terms and those present in the training distribution. Consequently, the estimated explicit loss fails to reliably assess the similarity between the model and the target distribution, ultimately inhibiting training.

Strategies that require POVM measurements with exponentially many elements also do not fall within the scope of our guidelines and could, in principle, offer a path toward avoiding the effects of exponential concentration. In such cases, the hypothesis testing analysis breaks down: although each individual outcome probability may be exponentially close to a concentration point, the total number of possible outcomes is also exponential. To highlight this explicitly, Appendix~\ref{app:counter-example} presents a simple counterexample showing that two probability distributions with exponentially many elements—whose individual components are exponentially close—can still be distinguishable, provided information about both distributions is available.
That said, we are not aware of any current proposals in variational quantum computing or quantum machine learning that make use of such exponential POVMs. As discussed in Section~\ref{sec:procedure}, while some procedures may initially appear to require exponential POVM measurements, they are more accurately interpreted as involving multiple polynomial POVMs in disguise. Nonetheless, exploring this gap in our argument remains an intriguing direction for the hunt to vanquish barren plateaus.

\medskip

\section{Acknowledgement}
RAS acknowledges support from the Swiss National Science Foundation [grant number 200021-219329]. ZH acknowledges support from the Sandoz Family Foundation-Monique de Meuron program for Academic Promotion. ST acknowledges Exchange Faculty Travel Grant: TG168033 from Chulalongkorn University, as well as funding from National Research Council of Thailand
(NRCT) [grant number N42A680126].

\bibliography{bibliography.bib,quantum.bib}

\begin{thebibliography}{101}%
\makeatletter
\providecommand \@ifxundefined [1]{%
 \@ifx{#1\undefined}
}%
\providecommand \@ifnum [1]{%
 \ifnum #1\expandafter \@firstoftwo
 \else \expandafter \@secondoftwo
 \fi
}%
\providecommand \@ifx [1]{%
 \ifx #1\expandafter \@firstoftwo
 \else \expandafter \@secondoftwo
 \fi
}%
\providecommand \natexlab [1]{#1}%
\providecommand \enquote  [1]{``#1''}%
\providecommand \bibnamefont  [1]{#1}%
\providecommand \bibfnamefont [1]{#1}%
\providecommand \citenamefont [1]{#1}%
\providecommand \href@noop [0]{\@secondoftwo}%
\providecommand \href [0]{\begingroup \@sanitize@url \@href}%
\providecommand \@href[1]{\@@startlink{#1}\@@href}%
\providecommand \@@href[1]{\endgroup#1\@@endlink}%
\providecommand \@sanitize@url [0]{\catcode `\\12\catcode `\$12\catcode `\&12\catcode `\#12\catcode `\^12\catcode `\_12\catcode `\%12\relax}%
\providecommand \@@startlink[1]{}%
\providecommand \@@endlink[0]{}%
\providecommand \url  [0]{\begingroup\@sanitize@url \@url }%
\providecommand \@url [1]{\endgroup\@href {#1}{\urlprefix }}%
\providecommand \urlprefix  [0]{URL }%
\providecommand \Eprint [0]{\href }%
\providecommand \doibase [0]{https://doi.org/}%
\providecommand \selectlanguage [0]{\@gobble}%
\providecommand \bibinfo  [0]{\@secondoftwo}%
\providecommand \bibfield  [0]{\@secondoftwo}%
\providecommand \translation [1]{[#1]}%
\providecommand \BibitemOpen [0]{}%
\providecommand \bibitemStop [0]{}%
\providecommand \bibitemNoStop [0]{.\EOS\space}%
\providecommand \EOS [0]{\spacefactor3000\relax}%
\providecommand \BibitemShut  [1]{\csname bibitem#1\endcsname}%
\let\auto@bib@innerbib\@empty
\bibitem [{\citenamefont {Cerezo}\ \emph {et~al.}(2021{\natexlab{a}})\citenamefont {Cerezo}, \citenamefont {Arrasmith}, \citenamefont {Babbush}, \citenamefont {Benjamin}, \citenamefont {Endo}, \citenamefont {Fujii}, \citenamefont {McClean}, \citenamefont {Mitarai}, \citenamefont {Yuan}, \citenamefont {Cincio},\ and\ \citenamefont {Coles}}]{cerezo2020variationalreview}%
  \BibitemOpen
  \bibfield  {author} {\bibinfo {author} {\bibfnamefont {M.}~\bibnamefont {Cerezo}}, \bibinfo {author} {\bibfnamefont {A.}~\bibnamefont {Arrasmith}}, \bibinfo {author} {\bibfnamefont {R.}~\bibnamefont {Babbush}}, \bibinfo {author} {\bibfnamefont {S.~C.}\ \bibnamefont {Benjamin}}, \bibinfo {author} {\bibfnamefont {S.}~\bibnamefont {Endo}}, \bibinfo {author} {\bibfnamefont {K.}~\bibnamefont {Fujii}}, \bibinfo {author} {\bibfnamefont {J.~R.}\ \bibnamefont {McClean}}, \bibinfo {author} {\bibfnamefont {K.}~\bibnamefont {Mitarai}}, \bibinfo {author} {\bibfnamefont {X.}~\bibnamefont {Yuan}}, \bibinfo {author} {\bibfnamefont {L.}~\bibnamefont {Cincio}},\ and\ \bibinfo {author} {\bibfnamefont {P.~J.}\ \bibnamefont {Coles}},\ }\bibfield  {title} {\bibinfo {title} {Variational quantum algorithms},\ }\href {https://doi.org/10.1038/s42254-021-00348-9} {\bibfield  {journal} {\bibinfo  {journal} {Nature Reviews Physics}\ }\textbf {\bibinfo {volume} {3}},\ \bibinfo {pages} {625–644} (\bibinfo {year}
  {2021}{\natexlab{a}})}\BibitemShut {NoStop}%
\bibitem [{\citenamefont {Bharti}\ \emph {et~al.}(2022)\citenamefont {Bharti}, \citenamefont {Cervera-Lierta}, \citenamefont {Kyaw}, \citenamefont {Haug}, \citenamefont {Alperin-Lea}, \citenamefont {Anand}, \citenamefont {Degroote}, \citenamefont {Heimonen}, \citenamefont {Kottmann}, \citenamefont {Menke} \emph {et~al.}}]{bharti2021noisy}%
  \BibitemOpen
  \bibfield  {author} {\bibinfo {author} {\bibfnamefont {K.}~\bibnamefont {Bharti}}, \bibinfo {author} {\bibfnamefont {A.}~\bibnamefont {Cervera-Lierta}}, \bibinfo {author} {\bibfnamefont {T.~H.}\ \bibnamefont {Kyaw}}, \bibinfo {author} {\bibfnamefont {T.}~\bibnamefont {Haug}}, \bibinfo {author} {\bibfnamefont {S.}~\bibnamefont {Alperin-Lea}}, \bibinfo {author} {\bibfnamefont {A.}~\bibnamefont {Anand}}, \bibinfo {author} {\bibfnamefont {M.}~\bibnamefont {Degroote}}, \bibinfo {author} {\bibfnamefont {H.}~\bibnamefont {Heimonen}}, \bibinfo {author} {\bibfnamefont {J.~S.}\ \bibnamefont {Kottmann}}, \bibinfo {author} {\bibfnamefont {T.}~\bibnamefont {Menke}}, \emph {et~al.},\ }\bibfield  {title} {\bibinfo {title} {Noisy intermediate-scale quantum algorithms},\ }\href {https://doi.org/10.1103/RevModPhys.94.015004} {\bibfield  {journal} {\bibinfo  {journal} {Reviews of Modern Physics}\ }\textbf {\bibinfo {volume} {94}},\ \bibinfo {pages} {015004} (\bibinfo {year} {2022})}\BibitemShut {NoStop}%
\bibitem [{\citenamefont {Abbas}\ \emph {et~al.}(2023)\citenamefont {Abbas}, \citenamefont {King}, \citenamefont {Huang}, \citenamefont {Huggins}, \citenamefont {Movassagh}, \citenamefont {Gilboa},\ and\ \citenamefont {McClean}}]{abbas2023quantum}%
  \BibitemOpen
  \bibfield  {author} {\bibinfo {author} {\bibfnamefont {A.}~\bibnamefont {Abbas}}, \bibinfo {author} {\bibfnamefont {R.}~\bibnamefont {King}}, \bibinfo {author} {\bibfnamefont {H.-Y.}\ \bibnamefont {Huang}}, \bibinfo {author} {\bibfnamefont {W.~J.}\ \bibnamefont {Huggins}}, \bibinfo {author} {\bibfnamefont {R.}~\bibnamefont {Movassagh}}, \bibinfo {author} {\bibfnamefont {D.}~\bibnamefont {Gilboa}},\ and\ \bibinfo {author} {\bibfnamefont {J.~R.}\ \bibnamefont {McClean}},\ }\bibfield  {title} {\bibinfo {title} {On quantum backpropagation, information reuse, and cheating measurement collapse},\ }\href {https://arxiv.org/abs/2305.13362} {\bibfield  {journal} {\bibinfo  {journal} {arXiv preprint arXiv:2305.13362}\ } (\bibinfo {year} {2023})}\BibitemShut {NoStop}%
\bibitem [{\citenamefont {McArdle}\ \emph {et~al.}(2020)\citenamefont {McArdle}, \citenamefont {Endo}, \citenamefont {Aspuru-Guzik}, \citenamefont {Benjamin},\ and\ \citenamefont {Yuan}}]{mcardle2020quantum}%
  \BibitemOpen
  \bibfield  {author} {\bibinfo {author} {\bibfnamefont {S.}~\bibnamefont {McArdle}}, \bibinfo {author} {\bibfnamefont {S.}~\bibnamefont {Endo}}, \bibinfo {author} {\bibfnamefont {A.}~\bibnamefont {Aspuru-Guzik}}, \bibinfo {author} {\bibfnamefont {S.~C.}\ \bibnamefont {Benjamin}},\ and\ \bibinfo {author} {\bibfnamefont {X.}~\bibnamefont {Yuan}},\ }\bibfield  {title} {\bibinfo {title} {Quantum computational chemistry},\ }\href {https://doi.org/10.1103/RevModPhys.92.015003} {\bibfield  {journal} {\bibinfo  {journal} {Reviews of Modern Physics}\ }\textbf {\bibinfo {volume} {92}},\ \bibinfo {pages} {015003} (\bibinfo {year} {2020})}\BibitemShut {NoStop}%
\bibitem [{\citenamefont {McClean}\ \emph {et~al.}(2016)\citenamefont {McClean}, \citenamefont {Romero}, \citenamefont {Babbush},\ and\ \citenamefont {Aspuru-Guzik}}]{mcclean2016theory}%
  \BibitemOpen
  \bibfield  {author} {\bibinfo {author} {\bibfnamefont {J.~R.}\ \bibnamefont {McClean}}, \bibinfo {author} {\bibfnamefont {J.}~\bibnamefont {Romero}}, \bibinfo {author} {\bibfnamefont {R.}~\bibnamefont {Babbush}},\ and\ \bibinfo {author} {\bibfnamefont {A.}~\bibnamefont {Aspuru-Guzik}},\ }\bibfield  {title} {\bibinfo {title} {The theory of variational hybrid quantum-classical algorithms},\ }\href {https://doi.org/10.1007/978-94-015-8330-5_4} {\bibfield  {journal} {\bibinfo  {journal} {New Journal of Physics}\ }\textbf {\bibinfo {volume} {18}},\ \bibinfo {pages} {023023} (\bibinfo {year} {2016})}\BibitemShut {NoStop}%
\bibitem [{\citenamefont {P{\'e}rez-Salinas}\ \emph {et~al.}(2020)\citenamefont {P{\'e}rez-Salinas}, \citenamefont {Cervera-Lierta}, \citenamefont {Gil-Fuster},\ and\ \citenamefont {Latorre}}]{perez2020data}%
  \BibitemOpen
  \bibfield  {author} {\bibinfo {author} {\bibfnamefont {A.}~\bibnamefont {P{\'e}rez-Salinas}}, \bibinfo {author} {\bibfnamefont {A.}~\bibnamefont {Cervera-Lierta}}, \bibinfo {author} {\bibfnamefont {E.}~\bibnamefont {Gil-Fuster}},\ and\ \bibinfo {author} {\bibfnamefont {J.~I.}\ \bibnamefont {Latorre}},\ }\bibfield  {title} {\bibinfo {title} {Data re-uploading for a universal quantum classifier},\ }\href {https://doi.org/doi.org/10.22331/q-2020-02-06-226} {\bibfield  {journal} {\bibinfo  {journal} {Quantum}\ }\textbf {\bibinfo {volume} {4}},\ \bibinfo {pages} {226} (\bibinfo {year} {2020})}\BibitemShut {NoStop}%
\bibitem [{\citenamefont {Mitarai}\ \emph {et~al.}(2018)\citenamefont {Mitarai}, \citenamefont {Negoro}, \citenamefont {Kitagawa},\ and\ \citenamefont {Fujii}}]{mitarai2018quantum}%
  \BibitemOpen
  \bibfield  {author} {\bibinfo {author} {\bibfnamefont {K.}~\bibnamefont {Mitarai}}, \bibinfo {author} {\bibfnamefont {M.}~\bibnamefont {Negoro}}, \bibinfo {author} {\bibfnamefont {M.}~\bibnamefont {Kitagawa}},\ and\ \bibinfo {author} {\bibfnamefont {K.}~\bibnamefont {Fujii}},\ }\bibfield  {title} {\bibinfo {title} {Quantum circuit learning},\ }\href {https://doi.org/10.1103/PhysRevA.98.032309} {\bibfield  {journal} {\bibinfo  {journal} {Physical Review A}\ }\textbf {\bibinfo {volume} {98}},\ \bibinfo {pages} {032309} (\bibinfo {year} {2018})}\BibitemShut {NoStop}%
\bibitem [{\citenamefont {Schuld}(2021)}]{schuld2021supervised}%
  \BibitemOpen
  \bibfield  {author} {\bibinfo {author} {\bibfnamefont {M.}~\bibnamefont {Schuld}},\ }\bibfield  {title} {\bibinfo {title} {Supervised quantum machine learning models are kernel methods},\ }\href {https://arxiv.org/abs/2101.11020} {\bibfield  {journal} {\bibinfo  {journal} {arXiv preprint arXiv:2101.11020}\ } (\bibinfo {year} {2021})}\BibitemShut {NoStop}%
\bibitem [{\citenamefont {Havl{\'\i}{\v{c}}ek}\ \emph {et~al.}(2019)\citenamefont {Havl{\'\i}{\v{c}}ek}, \citenamefont {C{\'o}rcoles}, \citenamefont {Temme}, \citenamefont {Harrow}, \citenamefont {Kandala}, \citenamefont {Chow},\ and\ \citenamefont {Gambetta}}]{havlivcek2019supervised}%
  \BibitemOpen
  \bibfield  {author} {\bibinfo {author} {\bibfnamefont {V.}~\bibnamefont {Havl{\'\i}{\v{c}}ek}}, \bibinfo {author} {\bibfnamefont {A.~D.}\ \bibnamefont {C{\'o}rcoles}}, \bibinfo {author} {\bibfnamefont {K.}~\bibnamefont {Temme}}, \bibinfo {author} {\bibfnamefont {A.~W.}\ \bibnamefont {Harrow}}, \bibinfo {author} {\bibfnamefont {A.}~\bibnamefont {Kandala}}, \bibinfo {author} {\bibfnamefont {J.~M.}\ \bibnamefont {Chow}},\ and\ \bibinfo {author} {\bibfnamefont {J.~M.}\ \bibnamefont {Gambetta}},\ }\bibfield  {title} {\bibinfo {title} {Supervised learning with quantum-enhanced feature spaces},\ }\href {https://doi.org/10.1038/s41586-019-0980-2} {\bibfield  {journal} {\bibinfo  {journal} {Nature}\ }\textbf {\bibinfo {volume} {567}},\ \bibinfo {pages} {209} (\bibinfo {year} {2019})}\BibitemShut {NoStop}%
\bibitem [{\citenamefont {Thanasilp}\ \emph {et~al.}(2024)\citenamefont {Thanasilp}, \citenamefont {Wang}, \citenamefont {Cerezo},\ and\ \citenamefont {Holmes}}]{thanasilp2022exponential}%
  \BibitemOpen
  \bibfield  {author} {\bibinfo {author} {\bibfnamefont {S.}~\bibnamefont {Thanasilp}}, \bibinfo {author} {\bibfnamefont {S.}~\bibnamefont {Wang}}, \bibinfo {author} {\bibfnamefont {M.}~\bibnamefont {Cerezo}},\ and\ \bibinfo {author} {\bibfnamefont {Z.}~\bibnamefont {Holmes}},\ }\bibfield  {title} {\bibinfo {title} {Exponential concentration in quantum kernel methods},\ }\href {https://doi.org/10.1038/s41467-024-49287-w} {\bibfield  {journal} {\bibinfo  {journal} {Nature Communications}\ }\textbf {\bibinfo {volume} {15}},\ \bibinfo {pages} {5200} (\bibinfo {year} {2024})}\BibitemShut {NoStop}%
\bibitem [{\citenamefont {K{\"u}bler}\ \emph {et~al.}(2021)\citenamefont {K{\"u}bler}, \citenamefont {Buchholz},\ and\ \citenamefont {Sch{\"o}lkopf}}]{kubler2021inductive}%
  \BibitemOpen
  \bibfield  {author} {\bibinfo {author} {\bibfnamefont {J.}~\bibnamefont {K{\"u}bler}}, \bibinfo {author} {\bibfnamefont {S.}~\bibnamefont {Buchholz}},\ and\ \bibinfo {author} {\bibfnamefont {B.}~\bibnamefont {Sch{\"o}lkopf}},\ }\bibfield  {title} {\bibinfo {title} {The inductive bias of quantum kernels},\ }\href {https://proceedings.neurips.cc/paper/2021/hash/69adc1e107f7f7d035d7baf04342e1ca-Abstract.html} {\bibfield  {journal} {\bibinfo  {journal} {Advances in Neural Information Processing Systems}\ }\textbf {\bibinfo {volume} {34}},\ \bibinfo {pages} {12661} (\bibinfo {year} {2021})}\BibitemShut {NoStop}%
\bibitem [{\citenamefont {Huang}\ \emph {et~al.}(2021)\citenamefont {Huang}, \citenamefont {Broughton}, \citenamefont {Mohseni}, \citenamefont {Babbush}, \citenamefont {Boixo}, \citenamefont {Neven},\ and\ \citenamefont {McClean}}]{huang2021power}%
  \BibitemOpen
  \bibfield  {author} {\bibinfo {author} {\bibfnamefont {H.-Y.}\ \bibnamefont {Huang}}, \bibinfo {author} {\bibfnamefont {M.}~\bibnamefont {Broughton}}, \bibinfo {author} {\bibfnamefont {M.}~\bibnamefont {Mohseni}}, \bibinfo {author} {\bibfnamefont {R.}~\bibnamefont {Babbush}}, \bibinfo {author} {\bibfnamefont {S.}~\bibnamefont {Boixo}}, \bibinfo {author} {\bibfnamefont {H.}~\bibnamefont {Neven}},\ and\ \bibinfo {author} {\bibfnamefont {J.~R.}\ \bibnamefont {McClean}},\ }\bibfield  {title} {\bibinfo {title} {Power of data in quantum machine learning},\ }\href {https://doi.org/10.1038/s41467-021-22539-9} {\bibfield  {journal} {\bibinfo  {journal} {Nature {C}ommunications}\ }\textbf {\bibinfo {volume} {12}},\ \bibinfo {pages} {1} (\bibinfo {year} {2021})}\BibitemShut {NoStop}%
\bibitem [{\citenamefont {Gentinetta}\ \emph {et~al.}(2022)\citenamefont {Gentinetta}, \citenamefont {Thomsen}, \citenamefont {Sutter},\ and\ \citenamefont {Woerner}}]{gentinetta2022complexity}%
  \BibitemOpen
  \bibfield  {author} {\bibinfo {author} {\bibfnamefont {G.}~\bibnamefont {Gentinetta}}, \bibinfo {author} {\bibfnamefont {A.}~\bibnamefont {Thomsen}}, \bibinfo {author} {\bibfnamefont {D.}~\bibnamefont {Sutter}},\ and\ \bibinfo {author} {\bibfnamefont {S.}~\bibnamefont {Woerner}},\ }\bibfield  {title} {\bibinfo {title} {The complexity of quantum support vector machines},\ }\href {https://arxiv.org/abs/2203.00031} {\bibfield  {journal} {\bibinfo  {journal} {arXiv preprint arXiv:2203.00031}\ } (\bibinfo {year} {2022})}\BibitemShut {NoStop}%
\bibitem [{\citenamefont {Gan}\ \emph {et~al.}(2023)\citenamefont {Gan}, \citenamefont {Leykam},\ and\ \citenamefont {Thanasilp}}]{gan2023unified}%
  \BibitemOpen
  \bibfield  {author} {\bibinfo {author} {\bibfnamefont {B.~Y.}\ \bibnamefont {Gan}}, \bibinfo {author} {\bibfnamefont {D.}~\bibnamefont {Leykam}},\ and\ \bibinfo {author} {\bibfnamefont {S.}~\bibnamefont {Thanasilp}},\ }\bibfield  {title} {\bibinfo {title} {A unified framework for trace-induced quantum kernels},\ }\bibfield  {journal} {\bibinfo  {journal} {arXiv preprint arXiv:2311.13552}\ }\href {https://doi.org/https://doi.org/10.48550/arXiv.2311.13552} {https://doi.org/10.48550/arXiv.2311.13552} (\bibinfo {year} {2023})\BibitemShut {NoStop}%
\bibitem [{\citenamefont {Zimbor{\'a}s}\ \emph {et~al.}(2025)\citenamefont {Zimbor{\'a}s}, \citenamefont {Koczor}, \citenamefont {Holmes}, \citenamefont {Borrelli}, \citenamefont {Gily{\'e}n}, \citenamefont {Huang}, \citenamefont {Cai}, \citenamefont {Ac{\'\i}n}, \citenamefont {Aolita}, \citenamefont {Banchi} \emph {et~al.}}]{zimboras2025myths}%
  \BibitemOpen
  \bibfield  {author} {\bibinfo {author} {\bibfnamefont {Z.}~\bibnamefont {Zimbor{\'a}s}}, \bibinfo {author} {\bibfnamefont {B.}~\bibnamefont {Koczor}}, \bibinfo {author} {\bibfnamefont {Z.}~\bibnamefont {Holmes}}, \bibinfo {author} {\bibfnamefont {E.-M.}\ \bibnamefont {Borrelli}}, \bibinfo {author} {\bibfnamefont {A.}~\bibnamefont {Gily{\'e}n}}, \bibinfo {author} {\bibfnamefont {H.-Y.}\ \bibnamefont {Huang}}, \bibinfo {author} {\bibfnamefont {Z.}~\bibnamefont {Cai}}, \bibinfo {author} {\bibfnamefont {A.}~\bibnamefont {Ac{\'\i}n}}, \bibinfo {author} {\bibfnamefont {L.}~\bibnamefont {Aolita}}, \bibinfo {author} {\bibfnamefont {L.}~\bibnamefont {Banchi}}, \emph {et~al.},\ }\bibfield  {title} {\bibinfo {title} {Myths around quantum computation before full fault tolerance: What no-go theorems rule out and what they don't},\ }\bibfield  {journal} {\bibinfo  {journal} {arXiv preprint arXiv:2501.05694}\ }\href {https://doi.org/https://doi.org/10.48550/arXiv.2501.05694} {https://doi.org/10.48550/arXiv.2501.05694}
  (\bibinfo {year} {2025})\BibitemShut {NoStop}%
\bibitem [{\citenamefont {Anschuetz}\ and\ \citenamefont {Kiani}(2022)}]{anschuetz2022quantum}%
  \BibitemOpen
  \bibfield  {author} {\bibinfo {author} {\bibfnamefont {E.~R.}\ \bibnamefont {Anschuetz}}\ and\ \bibinfo {author} {\bibfnamefont {B.~T.}\ \bibnamefont {Kiani}},\ }\bibfield  {title} {\bibinfo {title} {Quantum variational algorithms are swamped with traps},\ }\href {https://doi.org/10.1038/s41467-022-35364-5} {\bibfield  {journal} {\bibinfo  {journal} {Nature Communications}\ }\textbf {\bibinfo {volume} {13}},\ \bibinfo {pages} {7760} (\bibinfo {year} {2022})}\BibitemShut {NoStop}%
\bibitem [{\citenamefont {Larocca}\ \emph {et~al.}(2023)\citenamefont {Larocca}, \citenamefont {Ju}, \citenamefont {García-Martín}, \citenamefont {Coles},\ and\ \citenamefont {Cerezo}}]{larocca2021theory}%
  \BibitemOpen
  \bibfield  {author} {\bibinfo {author} {\bibfnamefont {M.}~\bibnamefont {Larocca}}, \bibinfo {author} {\bibfnamefont {N.}~\bibnamefont {Ju}}, \bibinfo {author} {\bibfnamefont {D.}~\bibnamefont {García-Martín}}, \bibinfo {author} {\bibfnamefont {P.~J.}\ \bibnamefont {Coles}},\ and\ \bibinfo {author} {\bibfnamefont {M.}~\bibnamefont {Cerezo}},\ }\bibfield  {title} {\bibinfo {title} {Theory of overparametrization in quantum neural networks},\ }\href {https://doi.org/https://doi.org/10.1038/s43588-023-00467-6} {\bibfield  {journal} {\bibinfo  {journal} {Nature Computational Science}\ }\textbf {\bibinfo {volume} {3}},\ \bibinfo {pages} {542} (\bibinfo {year} {2023})}\BibitemShut {NoStop}%
\bibitem [{\citenamefont {Schreiber}\ \emph {et~al.}(2023)\citenamefont {Schreiber}, \citenamefont {Eisert},\ and\ \citenamefont {Meyer}}]{schreiber2022classical}%
  \BibitemOpen
  \bibfield  {author} {\bibinfo {author} {\bibfnamefont {F.~J.}\ \bibnamefont {Schreiber}}, \bibinfo {author} {\bibfnamefont {J.}~\bibnamefont {Eisert}},\ and\ \bibinfo {author} {\bibfnamefont {J.~J.}\ \bibnamefont {Meyer}},\ }\bibfield  {title} {\bibinfo {title} {Classical surrogates for quantum learning models},\ }\href {https://doi.org/https://doi.org/10.1103/PhysRevLett.131.100803} {\bibfield  {journal} {\bibinfo  {journal} {Physical Review Letters}\ }\textbf {\bibinfo {volume} {131}},\ \bibinfo {pages} {100803} (\bibinfo {year} {2023})}\BibitemShut {NoStop}%
\bibitem [{\citenamefont {Sweke}\ \emph {et~al.}(2025)\citenamefont {Sweke}, \citenamefont {Recio}, \citenamefont {Jerbi}, \citenamefont {Gil-Fuster}, \citenamefont {Fuller}, \citenamefont {Eisert},\ and\ \citenamefont {Meyer}}]{sweke2023potential}%
  \BibitemOpen
  \bibfield  {author} {\bibinfo {author} {\bibfnamefont {R.}~\bibnamefont {Sweke}}, \bibinfo {author} {\bibfnamefont {E.}~\bibnamefont {Recio}}, \bibinfo {author} {\bibfnamefont {S.}~\bibnamefont {Jerbi}}, \bibinfo {author} {\bibfnamefont {E.}~\bibnamefont {Gil-Fuster}}, \bibinfo {author} {\bibfnamefont {B.}~\bibnamefont {Fuller}}, \bibinfo {author} {\bibfnamefont {J.}~\bibnamefont {Eisert}},\ and\ \bibinfo {author} {\bibfnamefont {J.~J.}\ \bibnamefont {Meyer}},\ }\bibfield  {title} {\bibinfo {title} {Potential and limitations of random fourier features for dequantizing quantum machine learning},\ }\href {https://doi.org/10.22331/q-2025-02-20-1640} {\bibfield  {journal} {\bibinfo  {journal} {Quantum}\ }\textbf {\bibinfo {volume} {9}},\ \bibinfo {pages} {1640} (\bibinfo {year} {2025})}\BibitemShut {NoStop}%
\bibitem [{\citenamefont {Sahebi}\ \emph {et~al.}(2025)\citenamefont {Sahebi}, \citenamefont {Barthe}, \citenamefont {Suzuki}, \citenamefont {Holmes},\ and\ \citenamefont {Grossi}}]{sahebi2025dequantization}%
  \BibitemOpen
  \bibfield  {author} {\bibinfo {author} {\bibfnamefont {M.}~\bibnamefont {Sahebi}}, \bibinfo {author} {\bibfnamefont {A.}~\bibnamefont {Barthe}}, \bibinfo {author} {\bibfnamefont {Y.}~\bibnamefont {Suzuki}}, \bibinfo {author} {\bibfnamefont {Z.}~\bibnamefont {Holmes}},\ and\ \bibinfo {author} {\bibfnamefont {M.}~\bibnamefont {Grossi}},\ }\bibfield  {title} {\bibinfo {title} {On dequantization of supervised quantum machine learning via random fourier features},\ }\href {https://arxiv.org/abs/2505.15902} {\bibfield  {journal} {\bibinfo  {journal} {arXiv preprint arXiv:2505.15902}\ } (\bibinfo {year} {2025})}\BibitemShut {NoStop}%
\bibitem [{\citenamefont {Rudolph}\ \emph {et~al.}(2025)\citenamefont {Rudolph}, \citenamefont {Jones}, \citenamefont {Teng}, \citenamefont {Angrisani},\ and\ \citenamefont {Holmes}}]{rudolph2025pauli}%
  \BibitemOpen
  \bibfield  {author} {\bibinfo {author} {\bibfnamefont {M.~S.}\ \bibnamefont {Rudolph}}, \bibinfo {author} {\bibfnamefont {T.}~\bibnamefont {Jones}}, \bibinfo {author} {\bibfnamefont {Y.}~\bibnamefont {Teng}}, \bibinfo {author} {\bibfnamefont {A.}~\bibnamefont {Angrisani}},\ and\ \bibinfo {author} {\bibfnamefont {Z.}~\bibnamefont {Holmes}},\ }\bibfield  {title} {\bibinfo {title} {Pauli propagation: A computational framework for simulating quantum systems},\ }\href {https://arxiv.org/abs/2505.21606} {\bibfield  {journal} {\bibinfo  {journal} {arXiv preprint arXiv:2505.21606}\ } (\bibinfo {year} {2025})}\BibitemShut {NoStop}%
\bibitem [{\citenamefont {Angrisani}\ \emph {et~al.}(2024)\citenamefont {Angrisani}, \citenamefont {Schmidhuber}, \citenamefont {Rudolph}, \citenamefont {Cerezo}, \citenamefont {Holmes},\ and\ \citenamefont {Huang}}]{angrisani2024classically}%
  \BibitemOpen
  \bibfield  {author} {\bibinfo {author} {\bibfnamefont {A.}~\bibnamefont {Angrisani}}, \bibinfo {author} {\bibfnamefont {A.}~\bibnamefont {Schmidhuber}}, \bibinfo {author} {\bibfnamefont {M.~S.}\ \bibnamefont {Rudolph}}, \bibinfo {author} {\bibfnamefont {M.}~\bibnamefont {Cerezo}}, \bibinfo {author} {\bibfnamefont {Z.}~\bibnamefont {Holmes}},\ and\ \bibinfo {author} {\bibfnamefont {H.-Y.}\ \bibnamefont {Huang}},\ }\bibfield  {title} {\bibinfo {title} {Classically estimating observables of noiseless quantum circuits},\ }\href {https://arxiv.org/abs/2409.01706} {\bibfield  {journal} {\bibinfo  {journal} {arXiv preprint arXiv:2409.01706}\ } (\bibinfo {year} {2024})}\BibitemShut {NoStop}%
\bibitem [{\citenamefont {Shin}\ \emph {et~al.}(2024)\citenamefont {Shin}, \citenamefont {Teo},\ and\ \citenamefont {Jeong}}]{shin2024dequantising}%
  \BibitemOpen
  \bibfield  {author} {\bibinfo {author} {\bibfnamefont {S.}~\bibnamefont {Shin}}, \bibinfo {author} {\bibfnamefont {Y.~S.}\ \bibnamefont {Teo}},\ and\ \bibinfo {author} {\bibfnamefont {H.}~\bibnamefont {Jeong}},\ }\bibfield  {title} {\bibinfo {title} {Dequantizing quantum machine learning models using tensor networks},\ }\href {https://doi.org/10.1103/PhysRevResearch.6.023218} {\bibfield  {journal} {\bibinfo  {journal} {Phys. Rev. Res.}\ }\textbf {\bibinfo {volume} {6}},\ \bibinfo {pages} {023218} (\bibinfo {year} {2024})}\BibitemShut {NoStop}%
\bibitem [{\citenamefont {Goh}\ \emph {et~al.}(2023)\citenamefont {Goh}, \citenamefont {Larocca}, \citenamefont {Cincio}, \citenamefont {Cerezo},\ and\ \citenamefont {Sauvage}}]{goh2023lie}%
  \BibitemOpen
  \bibfield  {author} {\bibinfo {author} {\bibfnamefont {M.~L.}\ \bibnamefont {Goh}}, \bibinfo {author} {\bibfnamefont {M.}~\bibnamefont {Larocca}}, \bibinfo {author} {\bibfnamefont {L.}~\bibnamefont {Cincio}}, \bibinfo {author} {\bibfnamefont {M.}~\bibnamefont {Cerezo}},\ and\ \bibinfo {author} {\bibfnamefont {F.}~\bibnamefont {Sauvage}},\ }\bibfield  {title} {\bibinfo {title} {Lie-algebraic classical simulations for quantum computing},\ }\href {https://arxiv.org/abs/2308.01432} {\bibfield  {journal} {\bibinfo  {journal} {arXiv preprint arXiv:2308.01432}\ } (\bibinfo {year} {2023})}\BibitemShut {NoStop}%
\bibitem [{\citenamefont {Kerenidis}\ \emph {et~al.}(2021)\citenamefont {Kerenidis}, \citenamefont {Landman},\ and\ \citenamefont {Mathur}}]{kerenidis2021classical}%
  \BibitemOpen
  \bibfield  {author} {\bibinfo {author} {\bibfnamefont {I.}~\bibnamefont {Kerenidis}}, \bibinfo {author} {\bibfnamefont {J.}~\bibnamefont {Landman}},\ and\ \bibinfo {author} {\bibfnamefont {N.}~\bibnamefont {Mathur}},\ }\bibfield  {title} {\bibinfo {title} {Classical and quantum algorithms for orthogonal neural networks},\ }\href {https://arxiv.org/abs/2106.07198} {\bibfield  {journal} {\bibinfo  {journal} {arXiv preprint arXiv:2106.07198}\ } (\bibinfo {year} {2021})}\BibitemShut {NoStop}%
\bibitem [{\citenamefont {Schuster}\ \emph {et~al.}(2024)\citenamefont {Schuster}, \citenamefont {Yin}, \citenamefont {Gao},\ and\ \citenamefont {Yao}}]{schuster2024polynomial}%
  \BibitemOpen
  \bibfield  {author} {\bibinfo {author} {\bibfnamefont {T.}~\bibnamefont {Schuster}}, \bibinfo {author} {\bibfnamefont {C.}~\bibnamefont {Yin}}, \bibinfo {author} {\bibfnamefont {X.}~\bibnamefont {Gao}},\ and\ \bibinfo {author} {\bibfnamefont {N.~Y.}\ \bibnamefont {Yao}},\ }\bibfield  {title} {\bibinfo {title} {A polynomial-time classical algorithm for noisy quantum circuits},\ }\bibfield  {journal} {\bibinfo  {journal} {arXiv preprint arXiv:2407.12768}\ }\href {https://doi.org/https://doi.org/10.48550/arXiv.2407.12768} {https://doi.org/10.48550/arXiv.2407.12768} (\bibinfo {year} {2024})\BibitemShut {NoStop}%
\bibitem [{\citenamefont {Fontana}\ \emph {et~al.}(2025)\citenamefont {Fontana}, \citenamefont {Rudolph}, \citenamefont {Duncan}, \citenamefont {Rungger},\ and\ \citenamefont {C{\^\i}rstoiu}}]{fontana2023classical}%
  \BibitemOpen
  \bibfield  {author} {\bibinfo {author} {\bibfnamefont {E.}~\bibnamefont {Fontana}}, \bibinfo {author} {\bibfnamefont {M.~S.}\ \bibnamefont {Rudolph}}, \bibinfo {author} {\bibfnamefont {R.}~\bibnamefont {Duncan}}, \bibinfo {author} {\bibfnamefont {I.}~\bibnamefont {Rungger}},\ and\ \bibinfo {author} {\bibfnamefont {C.}~\bibnamefont {C{\^\i}rstoiu}},\ }\bibfield  {title} {\bibinfo {title} {Classical simulations of noisy variational quantum circuits},\ }\href {https://doi.org/https://doi.org/10.1038/s41534-024-00955-1} {\bibfield  {journal} {\bibinfo  {journal} {npj Quantum Information}\ }\textbf {\bibinfo {volume} {11}},\ \bibinfo {pages} {1} (\bibinfo {year} {2025})}\BibitemShut {NoStop}%
\bibitem [{\citenamefont {Larocca}\ \emph {et~al.}(2025)\citenamefont {Larocca}, \citenamefont {Thanasilp}, \citenamefont {Wang}, \citenamefont {Sharma}, \citenamefont {Biamonte}, \citenamefont {Coles}, \citenamefont {Cincio}, \citenamefont {McClean}, \citenamefont {Holmes},\ and\ \citenamefont {Cerezo}}]{larocca2024review}%
  \BibitemOpen
  \bibfield  {author} {\bibinfo {author} {\bibfnamefont {M.}~\bibnamefont {Larocca}}, \bibinfo {author} {\bibfnamefont {S.}~\bibnamefont {Thanasilp}}, \bibinfo {author} {\bibfnamefont {S.}~\bibnamefont {Wang}}, \bibinfo {author} {\bibfnamefont {K.}~\bibnamefont {Sharma}}, \bibinfo {author} {\bibfnamefont {J.}~\bibnamefont {Biamonte}}, \bibinfo {author} {\bibfnamefont {P.~J.}\ \bibnamefont {Coles}}, \bibinfo {author} {\bibfnamefont {L.}~\bibnamefont {Cincio}}, \bibinfo {author} {\bibfnamefont {J.~R.}\ \bibnamefont {McClean}}, \bibinfo {author} {\bibfnamefont {Z.}~\bibnamefont {Holmes}},\ and\ \bibinfo {author} {\bibfnamefont {M.}~\bibnamefont {Cerezo}},\ }\bibfield  {title} {\bibinfo {title} {A review of barren plateaus in variational quantum computing},\ }\href {https://doi.org/10.1038/s42254-025-00813-9} {\bibfield  {journal} {\bibinfo  {journal} {Nature Reviews Physics}\ }\textbf {\bibinfo {volume} {3}},\ \bibinfo {pages} {625–644} (\bibinfo {year} {2025})}\BibitemShut {NoStop}%
\bibitem [{\citenamefont {McClean}\ \emph {et~al.}(2018)\citenamefont {McClean}, \citenamefont {Boixo}, \citenamefont {Smelyanskiy}, \citenamefont {Babbush},\ and\ \citenamefont {Neven}}]{mcclean2018barren}%
  \BibitemOpen
  \bibfield  {author} {\bibinfo {author} {\bibfnamefont {J.~R.}\ \bibnamefont {McClean}}, \bibinfo {author} {\bibfnamefont {S.}~\bibnamefont {Boixo}}, \bibinfo {author} {\bibfnamefont {V.~N.}\ \bibnamefont {Smelyanskiy}}, \bibinfo {author} {\bibfnamefont {R.}~\bibnamefont {Babbush}},\ and\ \bibinfo {author} {\bibfnamefont {H.}~\bibnamefont {Neven}},\ }\bibfield  {title} {\bibinfo {title} {Barren plateaus in quantum neural network training landscapes},\ }\href {https://doi.org/10.1038/s41467-018-07090-4} {\bibfield  {journal} {\bibinfo  {journal} {Nature {C}ommunications}\ }\textbf {\bibinfo {volume} {9}},\ \bibinfo {pages} {1} (\bibinfo {year} {2018})}\BibitemShut {NoStop}%
\bibitem [{\citenamefont {Fontana}\ \emph {et~al.}(2024)\citenamefont {Fontana}, \citenamefont {Herman}, \citenamefont {Chakrabarti}, \citenamefont {Kumar}, \citenamefont {Yalovetzky}, \citenamefont {Heredge}, \citenamefont {Sureshbabu},\ and\ \citenamefont {Pistoia}}]{fontana2023theadjoint}%
  \BibitemOpen
  \bibfield  {author} {\bibinfo {author} {\bibfnamefont {E.}~\bibnamefont {Fontana}}, \bibinfo {author} {\bibfnamefont {D.}~\bibnamefont {Herman}}, \bibinfo {author} {\bibfnamefont {S.}~\bibnamefont {Chakrabarti}}, \bibinfo {author} {\bibfnamefont {N.}~\bibnamefont {Kumar}}, \bibinfo {author} {\bibfnamefont {R.}~\bibnamefont {Yalovetzky}}, \bibinfo {author} {\bibfnamefont {J.}~\bibnamefont {Heredge}}, \bibinfo {author} {\bibfnamefont {S.~H.}\ \bibnamefont {Sureshbabu}},\ and\ \bibinfo {author} {\bibfnamefont {M.}~\bibnamefont {Pistoia}},\ }\bibfield  {title} {\bibinfo {title} {Characterizing barren plateaus in quantum ansätze with the adjoint representation},\ }\href {https://doi.org/10.1038/s41467-024-49910-w} {\bibfield  {journal} {\bibinfo  {journal} {Nature Communications}\ }\textbf {\bibinfo {volume} {15}},\ \bibinfo {pages} {7171} (\bibinfo {year} {2024})}\BibitemShut {NoStop}%
\bibitem [{\citenamefont {Ragone}\ \emph {et~al.}(2024)\citenamefont {Ragone}, \citenamefont {Bakalov}, \citenamefont {Sauvage}, \citenamefont {Kemper}, \citenamefont {Ortiz~Marrero}, \citenamefont {Larocca},\ and\ \citenamefont {Cerezo}}]{ragone2023unified}%
  \BibitemOpen
  \bibfield  {author} {\bibinfo {author} {\bibfnamefont {M.}~\bibnamefont {Ragone}}, \bibinfo {author} {\bibfnamefont {B.~N.}\ \bibnamefont {Bakalov}}, \bibinfo {author} {\bibfnamefont {F.}~\bibnamefont {Sauvage}}, \bibinfo {author} {\bibfnamefont {A.~F.}\ \bibnamefont {Kemper}}, \bibinfo {author} {\bibfnamefont {C.}~\bibnamefont {Ortiz~Marrero}}, \bibinfo {author} {\bibfnamefont {M.}~\bibnamefont {Larocca}},\ and\ \bibinfo {author} {\bibfnamefont {M.}~\bibnamefont {Cerezo}},\ }\bibfield  {title} {\bibinfo {title} {A lie algebraic theory of barren plateaus for deep parameterized quantum circuits},\ }\href {https://doi.org/10.1038/s41467-024-49909-3} {\bibfield  {journal} {\bibinfo  {journal} {Nature Communications}\ }\textbf {\bibinfo {volume} {15}},\ \bibinfo {pages} {7172} (\bibinfo {year} {2024})}\BibitemShut {NoStop}%
\bibitem [{\citenamefont {Cerezo}\ \emph {et~al.}(2023)\citenamefont {Cerezo}, \citenamefont {Larocca}, \citenamefont {Garc{\'\i}a-Mart{\'\i}n}, \citenamefont {Diaz}, \citenamefont {Braccia}, \citenamefont {Fontana}, \citenamefont {Rudolph}, \citenamefont {Bermejo}, \citenamefont {Ijaz}, \citenamefont {Thanasilp} \emph {et~al.}}]{cerezo2023does}%
  \BibitemOpen
  \bibfield  {author} {\bibinfo {author} {\bibfnamefont {M.}~\bibnamefont {Cerezo}}, \bibinfo {author} {\bibfnamefont {M.}~\bibnamefont {Larocca}}, \bibinfo {author} {\bibfnamefont {D.}~\bibnamefont {Garc{\'\i}a-Mart{\'\i}n}}, \bibinfo {author} {\bibfnamefont {N.~L.}\ \bibnamefont {Diaz}}, \bibinfo {author} {\bibfnamefont {P.}~\bibnamefont {Braccia}}, \bibinfo {author} {\bibfnamefont {E.}~\bibnamefont {Fontana}}, \bibinfo {author} {\bibfnamefont {M.~S.}\ \bibnamefont {Rudolph}}, \bibinfo {author} {\bibfnamefont {P.}~\bibnamefont {Bermejo}}, \bibinfo {author} {\bibfnamefont {A.}~\bibnamefont {Ijaz}}, \bibinfo {author} {\bibfnamefont {S.}~\bibnamefont {Thanasilp}}, \emph {et~al.},\ }\bibfield  {title} {\bibinfo {title} {Does provable absence of barren plateaus imply classical simulability? {O}r, why we need to rethink variational quantum computing},\ }\href {https://arxiv.org/abs/2312.09121} {\bibfield  {journal} {\bibinfo  {journal} {arXiv preprint arXiv:2312.09121}\ } (\bibinfo {year} {2023})}\BibitemShut
  {NoStop}%
\bibitem [{\citenamefont {Suzuki}\ and\ \citenamefont {Li}(2023)}]{suzuki2023effect}%
  \BibitemOpen
  \bibfield  {author} {\bibinfo {author} {\bibfnamefont {Y.}~\bibnamefont {Suzuki}}\ and\ \bibinfo {author} {\bibfnamefont {M.}~\bibnamefont {Li}},\ }\bibfield  {title} {\bibinfo {title} {Effect of alternating layered ansatzes on trainability of projected quantum kernel},\ }\href {https://arxiv.org/abs/2310.00361} {\bibfield  {journal} {\bibinfo  {journal} {arXiv preprint arXiv:2310.00361}\ } (\bibinfo {year} {2023})}\BibitemShut {NoStop}%
\bibitem [{\citenamefont {Xiong}\ \emph {et~al.}(2023)\citenamefont {Xiong}, \citenamefont {Facelli}, \citenamefont {Sahebi}, \citenamefont {Agnel}, \citenamefont {Chotibut}, \citenamefont {Thanasilp},\ and\ \citenamefont {Holmes}}]{xiong2023fundamental}%
  \BibitemOpen
  \bibfield  {author} {\bibinfo {author} {\bibfnamefont {W.}~\bibnamefont {Xiong}}, \bibinfo {author} {\bibfnamefont {G.}~\bibnamefont {Facelli}}, \bibinfo {author} {\bibfnamefont {M.}~\bibnamefont {Sahebi}}, \bibinfo {author} {\bibfnamefont {O.}~\bibnamefont {Agnel}}, \bibinfo {author} {\bibfnamefont {T.}~\bibnamefont {Chotibut}}, \bibinfo {author} {\bibfnamefont {S.}~\bibnamefont {Thanasilp}},\ and\ \bibinfo {author} {\bibfnamefont {Z.}~\bibnamefont {Holmes}},\ }\bibfield  {title} {\bibinfo {title} {On fundamental aspects of quantum extreme learning machines},\ }\href {https://arxiv.org/abs/2312.15124} {\bibfield  {journal} {\bibinfo  {journal} {arXiv preprint arXiv:2312.15124}\ } (\bibinfo {year} {2023})}\BibitemShut {NoStop}%
\bibitem [{\citenamefont {Xiong}\ \emph {et~al.}(2025)\citenamefont {Xiong}, \citenamefont {Holmes}, \citenamefont {Angrisani}, \citenamefont {Suzuki}, \citenamefont {Chotibut},\ and\ \citenamefont {Thanasilp}}]{xiong2025role}%
  \BibitemOpen
  \bibfield  {author} {\bibinfo {author} {\bibfnamefont {W.}~\bibnamefont {Xiong}}, \bibinfo {author} {\bibfnamefont {Z.}~\bibnamefont {Holmes}}, \bibinfo {author} {\bibfnamefont {A.}~\bibnamefont {Angrisani}}, \bibinfo {author} {\bibfnamefont {Y.}~\bibnamefont {Suzuki}}, \bibinfo {author} {\bibfnamefont {T.}~\bibnamefont {Chotibut}},\ and\ \bibinfo {author} {\bibfnamefont {S.}~\bibnamefont {Thanasilp}},\ }\bibfield  {title} {\bibinfo {title} {Role of scrambling and noise in temporal information processing with quantum systems},\ }\href {https://arxiv.org/abs/2505.10080} {\bibfield  {journal} {\bibinfo  {journal} {arXiv preprint arXiv:2505.10080}\ } (\bibinfo {year} {2025})}\BibitemShut {NoStop}%
\bibitem [{\citenamefont {Shaydulin}\ and\ \citenamefont {Wild}(2022)}]{shaydulin2021importance}%
  \BibitemOpen
  \bibfield  {author} {\bibinfo {author} {\bibfnamefont {R.}~\bibnamefont {Shaydulin}}\ and\ \bibinfo {author} {\bibfnamefont {S.~M.}\ \bibnamefont {Wild}},\ }\bibfield  {title} {\bibinfo {title} {Importance of kernel bandwidth in quantum machine learning},\ }\href {https://doi.org/10.1103/PhysRevA.106.042407} {\bibfield  {journal} {\bibinfo  {journal} {Physical Review A}\ }\textbf {\bibinfo {volume} {106}},\ \bibinfo {pages} {042407} (\bibinfo {year} {2022})}\BibitemShut {NoStop}%
\bibitem [{\citenamefont {Arrasmith}\ \emph {et~al.}(2022)\citenamefont {Arrasmith}, \citenamefont {Holmes}, \citenamefont {Cerezo},\ and\ \citenamefont {Coles}}]{arrasmith2021equivalence}%
  \BibitemOpen
  \bibfield  {author} {\bibinfo {author} {\bibfnamefont {A.}~\bibnamefont {Arrasmith}}, \bibinfo {author} {\bibfnamefont {Z.}~\bibnamefont {Holmes}}, \bibinfo {author} {\bibfnamefont {M.}~\bibnamefont {Cerezo}},\ and\ \bibinfo {author} {\bibfnamefont {P.~J.}\ \bibnamefont {Coles}},\ }\bibfield  {title} {\bibinfo {title} {Equivalence of quantum barren plateaus to cost concentration and narrow gorges},\ }\href {https://doi.org/10.1088/2058-9565/ac7d06} {\bibfield  {journal} {\bibinfo  {journal} {Quantum Science and Technology}\ }\textbf {\bibinfo {volume} {7}},\ \bibinfo {pages} {045015} (\bibinfo {year} {2022})}\BibitemShut {NoStop}%
\bibitem [{\citenamefont {Arrasmith}\ \emph {et~al.}(2021)\citenamefont {Arrasmith}, \citenamefont {Cerezo}, \citenamefont {Czarnik}, \citenamefont {Cincio},\ and\ \citenamefont {Coles}}]{arrasmith2020effect}%
  \BibitemOpen
  \bibfield  {author} {\bibinfo {author} {\bibfnamefont {A.}~\bibnamefont {Arrasmith}}, \bibinfo {author} {\bibfnamefont {M.}~\bibnamefont {Cerezo}}, \bibinfo {author} {\bibfnamefont {P.}~\bibnamefont {Czarnik}}, \bibinfo {author} {\bibfnamefont {L.}~\bibnamefont {Cincio}},\ and\ \bibinfo {author} {\bibfnamefont {P.~J.}\ \bibnamefont {Coles}},\ }\bibfield  {title} {\bibinfo {title} {Effect of barren plateaus on gradient-free optimization},\ }\href {https://doi.org/10.22331/q-2021-10-05-558} {\bibfield  {journal} {\bibinfo  {journal} {Quantum}\ }\textbf {\bibinfo {volume} {5}},\ \bibinfo {pages} {558} (\bibinfo {year} {2021})}\BibitemShut {NoStop}%
\bibitem [{\citenamefont {Larocca}\ \emph {et~al.}(2022{\natexlab{a}})\citenamefont {Larocca}, \citenamefont {Sauvage}, \citenamefont {Sbahi}, \citenamefont {Verdon}, \citenamefont {Coles},\ and\ \citenamefont {Cerezo}}]{larocca2022group}%
  \BibitemOpen
  \bibfield  {author} {\bibinfo {author} {\bibfnamefont {M.}~\bibnamefont {Larocca}}, \bibinfo {author} {\bibfnamefont {F.}~\bibnamefont {Sauvage}}, \bibinfo {author} {\bibfnamefont {F.~M.}\ \bibnamefont {Sbahi}}, \bibinfo {author} {\bibfnamefont {G.}~\bibnamefont {Verdon}}, \bibinfo {author} {\bibfnamefont {P.~J.}\ \bibnamefont {Coles}},\ and\ \bibinfo {author} {\bibfnamefont {M.}~\bibnamefont {Cerezo}},\ }\bibfield  {title} {\bibinfo {title} {Group-invariant quantum machine learning},\ }\href {https://doi.org/10.1103/PRXQuantum.3.030341} {\bibfield  {journal} {\bibinfo  {journal} {PRX Quantum}\ }\textbf {\bibinfo {volume} {3}},\ \bibinfo {pages} {030341} (\bibinfo {year} {2022}{\natexlab{a}})}\BibitemShut {NoStop}%
\bibitem [{\citenamefont {Schatzki}\ \emph {et~al.}(2024)\citenamefont {Schatzki}, \citenamefont {Larocca}, \citenamefont {Nguyen}, \citenamefont {Sauvage},\ and\ \citenamefont {Cerezo}}]{schatzki2022theoretical}%
  \BibitemOpen
  \bibfield  {author} {\bibinfo {author} {\bibfnamefont {L.}~\bibnamefont {Schatzki}}, \bibinfo {author} {\bibfnamefont {M.}~\bibnamefont {Larocca}}, \bibinfo {author} {\bibfnamefont {Q.~T.}\ \bibnamefont {Nguyen}}, \bibinfo {author} {\bibfnamefont {F.}~\bibnamefont {Sauvage}},\ and\ \bibinfo {author} {\bibfnamefont {M.}~\bibnamefont {Cerezo}},\ }\bibfield  {title} {\bibinfo {title} {Theoretical guarantees for permutation-equivariant quantum neural networks},\ }\href {https://doi.org/10.1038/s41534-024-00804-1} {\bibfield  {journal} {\bibinfo  {journal} {npj Quantum Information}\ }\textbf {\bibinfo {volume} {10}},\ \bibinfo {pages} {12} (\bibinfo {year} {2024})}\BibitemShut {NoStop}%
\bibitem [{\citenamefont {Cerezo}\ \emph {et~al.}(2021{\natexlab{b}})\citenamefont {Cerezo}, \citenamefont {Sone}, \citenamefont {Volkoff}, \citenamefont {Cincio},\ and\ \citenamefont {Coles}}]{cerezo2020cost}%
  \BibitemOpen
  \bibfield  {author} {\bibinfo {author} {\bibfnamefont {M.}~\bibnamefont {Cerezo}}, \bibinfo {author} {\bibfnamefont {A.}~\bibnamefont {Sone}}, \bibinfo {author} {\bibfnamefont {T.}~\bibnamefont {Volkoff}}, \bibinfo {author} {\bibfnamefont {L.}~\bibnamefont {Cincio}},\ and\ \bibinfo {author} {\bibfnamefont {P.~J.}\ \bibnamefont {Coles}},\ }\bibfield  {title} {\bibinfo {title} {Cost function dependent barren plateaus in shallow parametrized quantum circuits},\ }\href {https://doi.org/10.1038/s41467-021-21728-w} {\bibfield  {journal} {\bibinfo  {journal} {Nature {C}ommunications}\ }\textbf {\bibinfo {volume} {12}},\ \bibinfo {pages} {1} (\bibinfo {year} {2021}{\natexlab{b}})}\BibitemShut {NoStop}%
\bibitem [{\citenamefont {Larocca}\ \emph {et~al.}(2022{\natexlab{b}})\citenamefont {Larocca}, \citenamefont {Czarnik}, \citenamefont {Sharma}, \citenamefont {Muraleedharan}, \citenamefont {Coles},\ and\ \citenamefont {Cerezo}}]{larocca2021diagnosing}%
  \BibitemOpen
  \bibfield  {author} {\bibinfo {author} {\bibfnamefont {M.}~\bibnamefont {Larocca}}, \bibinfo {author} {\bibfnamefont {P.}~\bibnamefont {Czarnik}}, \bibinfo {author} {\bibfnamefont {K.}~\bibnamefont {Sharma}}, \bibinfo {author} {\bibfnamefont {G.}~\bibnamefont {Muraleedharan}}, \bibinfo {author} {\bibfnamefont {P.~J.}\ \bibnamefont {Coles}},\ and\ \bibinfo {author} {\bibfnamefont {M.}~\bibnamefont {Cerezo}},\ }\bibfield  {title} {\bibinfo {title} {Diagnosing {B}arren {P}lateaus with {T}ools from {Q}uantum {O}ptimal {C}ontrol},\ }\href {https://doi.org/10.22331/q-2022-09-29-824} {\bibfield  {journal} {\bibinfo  {journal} {{Quantum}}\ }\textbf {\bibinfo {volume} {6}},\ \bibinfo {pages} {824} (\bibinfo {year} {2022}{\natexlab{b}})}\BibitemShut {NoStop}%
\bibitem [{\citenamefont {Letcher}\ \emph {et~al.}(2024)\citenamefont {Letcher}, \citenamefont {Woerner},\ and\ \citenamefont {Zoufal}}]{letcher2023tight}%
  \BibitemOpen
  \bibfield  {author} {\bibinfo {author} {\bibfnamefont {A.}~\bibnamefont {Letcher}}, \bibinfo {author} {\bibfnamefont {S.}~\bibnamefont {Woerner}},\ and\ \bibinfo {author} {\bibfnamefont {C.}~\bibnamefont {Zoufal}},\ }\bibfield  {title} {\bibinfo {title} {Tight and efficient gradient bounds for parameterized quantum circuits},\ }\href {https://quantum-journal.org/papers/q-2024-09-25-1484/} {\bibfield  {journal} {\bibinfo  {journal} {Quantum}\ }\textbf {\bibinfo {volume} {8}},\ \bibinfo {pages} {1484} (\bibinfo {year} {2024})}\BibitemShut {NoStop}%
\bibitem [{\citenamefont {Basheer}\ \emph {et~al.}(2022)\citenamefont {Basheer}, \citenamefont {Feng}, \citenamefont {Ferrie},\ and\ \citenamefont {Li}}]{basheer2022alternating}%
  \BibitemOpen
  \bibfield  {author} {\bibinfo {author} {\bibfnamefont {A.}~\bibnamefont {Basheer}}, \bibinfo {author} {\bibfnamefont {Y.}~\bibnamefont {Feng}}, \bibinfo {author} {\bibfnamefont {C.}~\bibnamefont {Ferrie}},\ and\ \bibinfo {author} {\bibfnamefont {S.}~\bibnamefont {Li}},\ }\bibfield  {title} {\bibinfo {title} {Alternating layered variational quantum circuits can be classically optimized efficiently using classical shadows},\ }\href {https://arxiv.org/abs/2208.11623} {\bibfield  {journal} {\bibinfo  {journal} {arXiv preprint arXiv:2208.11623}\ } (\bibinfo {year} {2022})}\BibitemShut {NoStop}%
\bibitem [{\citenamefont {Napp}(2022)}]{napp2022quantifying}%
  \BibitemOpen
  \bibfield  {author} {\bibinfo {author} {\bibfnamefont {J.}~\bibnamefont {Napp}},\ }\bibfield  {title} {\bibinfo {title} {Quantifying the barren plateau phenomenon for a model of unstructured variational ans\"{a}tze},\ }\href {https://arxiv.org/abs/2203.06174} {\bibfield  {journal} {\bibinfo  {journal} {arXiv preprint arXiv:2203.06174}\ } (\bibinfo {year} {2022})}\BibitemShut {NoStop}%
\bibitem [{\citenamefont {Zhang}\ \emph {et~al.}(2024)\citenamefont {Zhang}, \citenamefont {Liu},\ and\ \citenamefont {Zhang}}]{zhang2023absence}%
  \BibitemOpen
  \bibfield  {author} {\bibinfo {author} {\bibfnamefont {H.-K.}\ \bibnamefont {Zhang}}, \bibinfo {author} {\bibfnamefont {S.}~\bibnamefont {Liu}},\ and\ \bibinfo {author} {\bibfnamefont {S.-X.}\ \bibnamefont {Zhang}},\ }\bibfield  {title} {\bibinfo {title} {Absence of barren plateaus in finite local-depth circuits with long-range entanglement},\ }\href {https://doi.org/10.1103/PhysRevLett.132.150603} {\bibfield  {journal} {\bibinfo  {journal} {Physical Review Letters}\ }\textbf {\bibinfo {volume} {132}},\ \bibinfo {pages} {150603} (\bibinfo {year} {2024})}\BibitemShut {NoStop}%
\bibitem [{\citenamefont {Pesah}\ \emph {et~al.}(2021)\citenamefont {Pesah}, \citenamefont {Cerezo}, \citenamefont {Wang}, \citenamefont {Volkoff}, \citenamefont {Sornborger},\ and\ \citenamefont {Coles}}]{pesah2020absence}%
  \BibitemOpen
  \bibfield  {author} {\bibinfo {author} {\bibfnamefont {A.}~\bibnamefont {Pesah}}, \bibinfo {author} {\bibfnamefont {M.}~\bibnamefont {Cerezo}}, \bibinfo {author} {\bibfnamefont {S.}~\bibnamefont {Wang}}, \bibinfo {author} {\bibfnamefont {T.}~\bibnamefont {Volkoff}}, \bibinfo {author} {\bibfnamefont {A.~T.}\ \bibnamefont {Sornborger}},\ and\ \bibinfo {author} {\bibfnamefont {P.~J.}\ \bibnamefont {Coles}},\ }\bibfield  {title} {\bibinfo {title} {Absence of barren plateaus in quantum convolutional neural networks},\ }\href {https://doi.org/10.1103/PhysRevX.11.041011} {\bibfield  {journal} {\bibinfo  {journal} {Physical Review X}\ }\textbf {\bibinfo {volume} {11}},\ \bibinfo {pages} {041011} (\bibinfo {year} {2021})}\BibitemShut {NoStop}%
\bibitem [{\citenamefont {Monbroussou}\ \emph {et~al.}(2023)\citenamefont {Monbroussou}, \citenamefont {Landman}, \citenamefont {Grilo}, \citenamefont {Kukla},\ and\ \citenamefont {Kashefi}}]{monbroussou2023trainability}%
  \BibitemOpen
  \bibfield  {author} {\bibinfo {author} {\bibfnamefont {L.}~\bibnamefont {Monbroussou}}, \bibinfo {author} {\bibfnamefont {J.}~\bibnamefont {Landman}}, \bibinfo {author} {\bibfnamefont {A.~B.}\ \bibnamefont {Grilo}}, \bibinfo {author} {\bibfnamefont {R.}~\bibnamefont {Kukla}},\ and\ \bibinfo {author} {\bibfnamefont {E.}~\bibnamefont {Kashefi}},\ }\bibfield  {title} {\bibinfo {title} {Trainability and expressivity of hamming-weight preserving quantum circuits for machine learning},\ }\href {https://arxiv.org/abs/2309.15547} {\bibfield  {journal} {\bibinfo  {journal} {arXiv preprint arXiv:2309.15547}\ } (\bibinfo {year} {2023})}\BibitemShut {NoStop}%
\bibitem [{\citenamefont {Raj}\ \emph {et~al.}(2023)\citenamefont {Raj}, \citenamefont {Kerenidis}, \citenamefont {Shekhar}, \citenamefont {Wood}, \citenamefont {Dee}, \citenamefont {Chakrabarti}, \citenamefont {Chen}, \citenamefont {Herman}, \citenamefont {Hu}, \citenamefont {Minssen} \emph {et~al.}}]{cherrat2023quantum}%
  \BibitemOpen
  \bibfield  {author} {\bibinfo {author} {\bibfnamefont {S.}~\bibnamefont {Raj}}, \bibinfo {author} {\bibfnamefont {I.}~\bibnamefont {Kerenidis}}, \bibinfo {author} {\bibfnamefont {A.}~\bibnamefont {Shekhar}}, \bibinfo {author} {\bibfnamefont {B.}~\bibnamefont {Wood}}, \bibinfo {author} {\bibfnamefont {J.}~\bibnamefont {Dee}}, \bibinfo {author} {\bibfnamefont {S.}~\bibnamefont {Chakrabarti}}, \bibinfo {author} {\bibfnamefont {R.}~\bibnamefont {Chen}}, \bibinfo {author} {\bibfnamefont {D.}~\bibnamefont {Herman}}, \bibinfo {author} {\bibfnamefont {S.}~\bibnamefont {Hu}}, \bibinfo {author} {\bibfnamefont {P.}~\bibnamefont {Minssen}}, \emph {et~al.},\ }\bibfield  {title} {\bibinfo {title} {Quantum deep hedging},\ }\href {https://doi.org/10.22331/q-2023-11-29-1191} {\bibfield  {journal} {\bibinfo  {journal} {Quantum}\ }\textbf {\bibinfo {volume} {7}},\ \bibinfo {pages} {1191} (\bibinfo {year} {2023})}\BibitemShut {NoStop}%
\bibitem [{\citenamefont {Diaz}\ \emph {et~al.}(2023)\citenamefont {Diaz}, \citenamefont {Garc{\'\i}a-Mart{\'\i}n}, \citenamefont {Kazi}, \citenamefont {Larocca},\ and\ \citenamefont {Cerezo}}]{diaz2023showcasing}%
  \BibitemOpen
  \bibfield  {author} {\bibinfo {author} {\bibfnamefont {N.~L.}\ \bibnamefont {Diaz}}, \bibinfo {author} {\bibfnamefont {D.}~\bibnamefont {Garc{\'\i}a-Mart{\'\i}n}}, \bibinfo {author} {\bibfnamefont {S.}~\bibnamefont {Kazi}}, \bibinfo {author} {\bibfnamefont {M.}~\bibnamefont {Larocca}},\ and\ \bibinfo {author} {\bibfnamefont {M.}~\bibnamefont {Cerezo}},\ }\bibfield  {title} {\bibinfo {title} {Showcasing a barren plateau theory beyond the dynamical lie algebra},\ }\href {https://arxiv.org/abs/2310.11505} {\bibfield  {journal} {\bibinfo  {journal} {arXiv preprint arXiv:2310.11505}\ } (\bibinfo {year} {2023})}\BibitemShut {NoStop}%
\bibitem [{\citenamefont {Mele}\ \emph {et~al.}(2024)\citenamefont {Mele}, \citenamefont {Angrisani}, \citenamefont {Ghosh}, \citenamefont {Khatri}, \citenamefont {Eisert}, \citenamefont {Fran{\c{c}}a},\ and\ \citenamefont {Quek}}]{mele2024noise}%
  \BibitemOpen
  \bibfield  {author} {\bibinfo {author} {\bibfnamefont {A.~A.}\ \bibnamefont {Mele}}, \bibinfo {author} {\bibfnamefont {A.}~\bibnamefont {Angrisani}}, \bibinfo {author} {\bibfnamefont {S.}~\bibnamefont {Ghosh}}, \bibinfo {author} {\bibfnamefont {S.}~\bibnamefont {Khatri}}, \bibinfo {author} {\bibfnamefont {J.}~\bibnamefont {Eisert}}, \bibinfo {author} {\bibfnamefont {D.~S.}\ \bibnamefont {Fran{\c{c}}a}},\ and\ \bibinfo {author} {\bibfnamefont {Y.}~\bibnamefont {Quek}},\ }\bibfield  {title} {\bibinfo {title} {Noise-induced shallow circuits and absence of barren plateaus},\ }\href {https://arxiv.org/abs/2403.13927} {\bibfield  {journal} {\bibinfo  {journal} {arXiv preprint arXiv:2403.13927}\ } (\bibinfo {year} {2024})}\BibitemShut {NoStop}%
\bibitem [{\citenamefont {Deshpande}\ \emph {et~al.}(2024)\citenamefont {Deshpande}, \citenamefont {Hinsche}, \citenamefont {Najafi}, \citenamefont {Sharma}, \citenamefont {Sweke},\ and\ \citenamefont {Zoufal}}]{deshpande2024dynamic}%
  \BibitemOpen
  \bibfield  {author} {\bibinfo {author} {\bibfnamefont {A.}~\bibnamefont {Deshpande}}, \bibinfo {author} {\bibfnamefont {M.}~\bibnamefont {Hinsche}}, \bibinfo {author} {\bibfnamefont {S.}~\bibnamefont {Najafi}}, \bibinfo {author} {\bibfnamefont {K.}~\bibnamefont {Sharma}}, \bibinfo {author} {\bibfnamefont {R.}~\bibnamefont {Sweke}},\ and\ \bibinfo {author} {\bibfnamefont {C.}~\bibnamefont {Zoufal}},\ }\bibfield  {title} {\bibinfo {title} {Dynamic parameterized quantum circuits: expressive and barren-plateau free},\ }\bibfield  {journal} {\bibinfo  {journal} {arXiv preprint arXiv:2411.05760}\ }\href {https://doi.org/10.48550/arXiv.2411.05760} {10.48550/arXiv.2411.05760} (\bibinfo {year} {2024})\BibitemShut {NoStop}%
\bibitem [{\citenamefont {Srimahajariyapong}\ \emph {et~al.}(2025)\citenamefont {Srimahajariyapong}, \citenamefont {Thanasilp},\ and\ \citenamefont {Chotibut}}]{srimahajariyapong2025connecting}%
  \BibitemOpen
  \bibfield  {author} {\bibinfo {author} {\bibfnamefont {K.}~\bibnamefont {Srimahajariyapong}}, \bibinfo {author} {\bibfnamefont {S.}~\bibnamefont {Thanasilp}},\ and\ \bibinfo {author} {\bibfnamefont {T.}~\bibnamefont {Chotibut}},\ }\bibfield  {title} {\bibinfo {title} {Connecting phases of matter to the flatness of the loss landscape in analog variational quantum algorithms},\ }\href {https://arxiv.org/abs/2506.13865} {\bibfield  {journal} {\bibinfo  {journal} {arXiv preprint arXiv:2506.13865}\ } (\bibinfo {year} {2025})}\BibitemShut {NoStop}%
\bibitem [{\citenamefont {Mhiri}\ \emph {et~al.}(2025)\citenamefont {Mhiri}, \citenamefont {Puig}, \citenamefont {Lerch}, \citenamefont {Rudolph}, \citenamefont {Chotibut}, \citenamefont {Thanasilp},\ and\ \citenamefont {Holmes}}]{mhiri2025unifying}%
  \BibitemOpen
  \bibfield  {author} {\bibinfo {author} {\bibfnamefont {H.}~\bibnamefont {Mhiri}}, \bibinfo {author} {\bibfnamefont {R.}~\bibnamefont {Puig}}, \bibinfo {author} {\bibfnamefont {S.}~\bibnamefont {Lerch}}, \bibinfo {author} {\bibfnamefont {M.~S.}\ \bibnamefont {Rudolph}}, \bibinfo {author} {\bibfnamefont {T.}~\bibnamefont {Chotibut}}, \bibinfo {author} {\bibfnamefont {S.}~\bibnamefont {Thanasilp}},\ and\ \bibinfo {author} {\bibfnamefont {Z.}~\bibnamefont {Holmes}},\ }\bibfield  {title} {\bibinfo {title} {A unifying account of warm start guarantees for patches of quantum landscapes},\ }\bibfield  {journal} {\bibinfo  {journal} {arXiv preprint arXiv:2502.07889}\ }\href {https://doi.org/https://doi.org/10.48550/arXiv.2502.07889} {https://doi.org/10.48550/arXiv.2502.07889} (\bibinfo {year} {2025})\BibitemShut {NoStop}%
\bibitem [{\citenamefont {Mele}\ \emph {et~al.}(2022)\citenamefont {Mele}, \citenamefont {Mbeng}, \citenamefont {Santoro}, \citenamefont {Collura},\ and\ \citenamefont {Torta}}]{mele2022avoiding}%
  \BibitemOpen
  \bibfield  {author} {\bibinfo {author} {\bibfnamefont {A.~A.}\ \bibnamefont {Mele}}, \bibinfo {author} {\bibfnamefont {G.~B.}\ \bibnamefont {Mbeng}}, \bibinfo {author} {\bibfnamefont {G.~E.}\ \bibnamefont {Santoro}}, \bibinfo {author} {\bibfnamefont {M.}~\bibnamefont {Collura}},\ and\ \bibinfo {author} {\bibfnamefont {P.}~\bibnamefont {Torta}},\ }\bibfield  {title} {\bibinfo {title} {Avoiding barren plateaus via transferability of smooth solutions in a {H}amiltonian variational ansatz},\ }\href {https://doi.org/10.1103/PhysRevA.106.L060401} {\bibfield  {journal} {\bibinfo  {journal} {Physical Review A}\ }\textbf {\bibinfo {volume} {106}},\ \bibinfo {pages} {L060401} (\bibinfo {year} {2022})}\BibitemShut {NoStop}%
\bibitem [{\citenamefont {Puig}\ \emph {et~al.}(2025)\citenamefont {Puig}, \citenamefont {Drudis}, \citenamefont {Thanasilp},\ and\ \citenamefont {Holmes}}]{puig2024variational}%
  \BibitemOpen
  \bibfield  {author} {\bibinfo {author} {\bibfnamefont {R.}~\bibnamefont {Puig}}, \bibinfo {author} {\bibfnamefont {M.}~\bibnamefont {Drudis}}, \bibinfo {author} {\bibfnamefont {S.}~\bibnamefont {Thanasilp}},\ and\ \bibinfo {author} {\bibfnamefont {Z.}~\bibnamefont {Holmes}},\ }\bibfield  {title} {\bibinfo {title} {Variational quantum simulation: A case study for understanding warm starts},\ }\href {https://doi.org/10.1103/PRXQuantum.6.010317} {\bibfield  {journal} {\bibinfo  {journal} {PRX Quantum}\ }\textbf {\bibinfo {volume} {6}},\ \bibinfo {pages} {010317} (\bibinfo {year} {2025})}\BibitemShut {NoStop}%
\bibitem [{\citenamefont {Chang}\ \emph {et~al.}(2024)\citenamefont {Chang}, \citenamefont {Thanasilp}, \citenamefont {Saux}, \citenamefont {Vallecorsa},\ and\ \citenamefont {Grossi}}]{chang2024latent}%
  \BibitemOpen
  \bibfield  {author} {\bibinfo {author} {\bibfnamefont {S.~Y.}\ \bibnamefont {Chang}}, \bibinfo {author} {\bibfnamefont {S.}~\bibnamefont {Thanasilp}}, \bibinfo {author} {\bibfnamefont {B.~L.}\ \bibnamefont {Saux}}, \bibinfo {author} {\bibfnamefont {S.}~\bibnamefont {Vallecorsa}},\ and\ \bibinfo {author} {\bibfnamefont {M.}~\bibnamefont {Grossi}},\ }\bibfield  {title} {\bibinfo {title} {Latent style-based quantum gan for high-quality image generation},\ }\href {https://arxiv.org/abs/2406.02668} {\bibfield  {journal} {\bibinfo  {journal} {arXiv preprint arXiv:2406.02668}\ } (\bibinfo {year} {2024})}\BibitemShut {NoStop}%
\bibitem [{\citenamefont {Wang}\ \emph {et~al.}(2024)\citenamefont {Wang}, \citenamefont {Qi}, \citenamefont {Ferrie},\ and\ \citenamefont {Dong}}]{wang2023trainability}%
  \BibitemOpen
  \bibfield  {author} {\bibinfo {author} {\bibfnamefont {Y.}~\bibnamefont {Wang}}, \bibinfo {author} {\bibfnamefont {B.}~\bibnamefont {Qi}}, \bibinfo {author} {\bibfnamefont {C.}~\bibnamefont {Ferrie}},\ and\ \bibinfo {author} {\bibfnamefont {D.}~\bibnamefont {Dong}},\ }\bibfield  {title} {\bibinfo {title} {Trainability enhancement of parameterized quantum circuits via reduced-domain parameter initialization},\ }\href {https://doi.org/10.1103/PhysRevApplied.22.054005} {\bibfield  {journal} {\bibinfo  {journal} {Physical Review Applied}\ }\textbf {\bibinfo {volume} {22}},\ \bibinfo {pages} {054005} (\bibinfo {year} {2024})}\BibitemShut {NoStop}%
\bibitem [{\citenamefont {Park}\ and\ \citenamefont {Killoran}(2024)}]{park2023hamiltonian}%
  \BibitemOpen
  \bibfield  {author} {\bibinfo {author} {\bibfnamefont {C.-Y.}\ \bibnamefont {Park}}\ and\ \bibinfo {author} {\bibfnamefont {N.}~\bibnamefont {Killoran}},\ }\bibfield  {title} {\bibinfo {title} {Hamiltonian variational ansatz without barren plateaus},\ }\href {https://doi.org/10.22331/q-2024-02-01-1239} {\bibfield  {journal} {\bibinfo  {journal} {Quantum}\ }\textbf {\bibinfo {volume} {8}},\ \bibinfo {pages} {1239} (\bibinfo {year} {2024})}\BibitemShut {NoStop}%
\bibitem [{\citenamefont {Park}\ \emph {et~al.}(2024)\citenamefont {Park}, \citenamefont {Kang},\ and\ \citenamefont {Huh}}]{park2024hardware}%
  \BibitemOpen
  \bibfield  {author} {\bibinfo {author} {\bibfnamefont {C.-Y.}\ \bibnamefont {Park}}, \bibinfo {author} {\bibfnamefont {M.}~\bibnamefont {Kang}},\ and\ \bibinfo {author} {\bibfnamefont {J.}~\bibnamefont {Huh}},\ }\bibfield  {title} {\bibinfo {title} {Hardware-efficient ansatz without barren plateaus in any depth},\ }\href {https://arxiv.org/abs/2403.04844} {\bibfield  {journal} {\bibinfo  {journal} {arXiv preprint arXiv:2403.04844}\ } (\bibinfo {year} {2024})}\BibitemShut {NoStop}%
\bibitem [{\citenamefont {Zhang}\ \emph {et~al.}(2022)\citenamefont {Zhang}, \citenamefont {Liu}, \citenamefont {Hsieh},\ and\ \citenamefont {Tao}}]{zhang2022escaping}%
  \BibitemOpen
  \bibfield  {author} {\bibinfo {author} {\bibfnamefont {K.}~\bibnamefont {Zhang}}, \bibinfo {author} {\bibfnamefont {L.}~\bibnamefont {Liu}}, \bibinfo {author} {\bibfnamefont {M.-H.}\ \bibnamefont {Hsieh}},\ and\ \bibinfo {author} {\bibfnamefont {D.}~\bibnamefont {Tao}},\ }\bibfield  {title} {\bibinfo {title} {Escaping from the barren plateau via {G}aussian initializations in deep variational quantum circuits},\ }in\ \href {https://openreview.net/forum?id=jXgbJdQ2YIy} {\emph {\bibinfo {booktitle} {Advances in Neural Information Processing Systems}}}\ (\bibinfo {year} {2022})\BibitemShut {NoStop}%
\bibitem [{\citenamefont {Tangpanitanon}\ \emph {et~al.}(2020)\citenamefont {Tangpanitanon}, \citenamefont {Thanasilp}, \citenamefont {Dangniam}, \citenamefont {Lemonde},\ and\ \citenamefont {Angelakis}}]{tangpanitanon2020expressibility}%
  \BibitemOpen
  \bibfield  {author} {\bibinfo {author} {\bibfnamefont {J.}~\bibnamefont {Tangpanitanon}}, \bibinfo {author} {\bibfnamefont {S.}~\bibnamefont {Thanasilp}}, \bibinfo {author} {\bibfnamefont {N.}~\bibnamefont {Dangniam}}, \bibinfo {author} {\bibfnamefont {M.-A.}\ \bibnamefont {Lemonde}},\ and\ \bibinfo {author} {\bibfnamefont {D.~G.}\ \bibnamefont {Angelakis}},\ }\bibfield  {title} {\bibinfo {title} {Expressibility and trainability of parametrized analog quantum systems for machine learning applications},\ }\href {https://doi.org/10.1103/PhysRevResearch.2.043364} {\bibfield  {journal} {\bibinfo  {journal} {Physical Review Research}\ }\textbf {\bibinfo {volume} {2}},\ \bibinfo {pages} {043364} (\bibinfo {year} {2020})}\BibitemShut {NoStop}%
\bibitem [{\citenamefont {Shi}\ and\ \citenamefont {Shang}(2024)}]{shi2024avoiding}%
  \BibitemOpen
  \bibfield  {author} {\bibinfo {author} {\bibfnamefont {X.}~\bibnamefont {Shi}}\ and\ \bibinfo {author} {\bibfnamefont {Y.}~\bibnamefont {Shang}},\ }\bibfield  {title} {\bibinfo {title} {Avoiding barren plateaus via {G}aussian mixture model},\ }\href {https://arxiv.org/abs/2402.13501} {\bibfield  {journal} {\bibinfo  {journal} {arXiv preprint arXiv:2402.13501}\ } (\bibinfo {year} {2024})}\BibitemShut {NoStop}%
\bibitem [{\citenamefont {Cao}\ \emph {et~al.}(2024)\citenamefont {Cao}, \citenamefont {Zhou}, \citenamefont {Tannu}, \citenamefont {Shannon},\ and\ \citenamefont {Joynt}}]{cao2024exploiting}%
  \BibitemOpen
  \bibfield  {author} {\bibinfo {author} {\bibfnamefont {C.}~\bibnamefont {Cao}}, \bibinfo {author} {\bibfnamefont {Y.}~\bibnamefont {Zhou}}, \bibinfo {author} {\bibfnamefont {S.}~\bibnamefont {Tannu}}, \bibinfo {author} {\bibfnamefont {N.}~\bibnamefont {Shannon}},\ and\ \bibinfo {author} {\bibfnamefont {R.}~\bibnamefont {Joynt}},\ }\bibfield  {title} {\bibinfo {title} {Exploiting many-body localization for scalable variational quantum simulation},\ }\href {https://arxiv.org/abs/2404.17560} {\bibfield  {journal} {\bibinfo  {journal} {arXiv preprint arXiv:2404.17560}\ } (\bibinfo {year} {2024})}\BibitemShut {NoStop}%
\bibitem [{\citenamefont {Grant}\ \emph {et~al.}(2019)\citenamefont {Grant}, \citenamefont {Wossnig}, \citenamefont {Ostaszewski},\ and\ \citenamefont {Benedetti}}]{grant2019initialization}%
  \BibitemOpen
  \bibfield  {author} {\bibinfo {author} {\bibfnamefont {E.}~\bibnamefont {Grant}}, \bibinfo {author} {\bibfnamefont {L.}~\bibnamefont {Wossnig}}, \bibinfo {author} {\bibfnamefont {M.}~\bibnamefont {Ostaszewski}},\ and\ \bibinfo {author} {\bibfnamefont {M.}~\bibnamefont {Benedetti}},\ }\bibfield  {title} {\bibinfo {title} {An initialization strategy for addressing barren plateaus in parametrized quantum circuits},\ }\href {https://doi.org/10.22331/q-2019-12-09-214} {\bibfield  {journal} {\bibinfo  {journal} {Quantum}\ }\textbf {\bibinfo {volume} {3}},\ \bibinfo {pages} {214} (\bibinfo {year} {2019})}\BibitemShut {NoStop}%
\bibitem [{\citenamefont {Friedrich}\ and\ \citenamefont {Maziero}(2022)}]{friedrich2022avoiding}%
  \BibitemOpen
  \bibfield  {author} {\bibinfo {author} {\bibfnamefont {L.}~\bibnamefont {Friedrich}}\ and\ \bibinfo {author} {\bibfnamefont {J.}~\bibnamefont {Maziero}},\ }\bibfield  {title} {\bibinfo {title} {Avoiding barren plateaus with classical deep neural networks},\ }\href {https://doi.org/10.1103/PhysRevA.106.042433} {\bibfield  {journal} {\bibinfo  {journal} {Physical Review A}\ }\textbf {\bibinfo {volume} {106}},\ \bibinfo {pages} {042433} (\bibinfo {year} {2022})}\BibitemShut {NoStop}%
\bibitem [{\citenamefont {Miao}\ \emph {et~al.}(2024)\citenamefont {Miao}, \citenamefont {Hsieh},\ and\ \citenamefont {Zhang}}]{miao2024neural}%
  \BibitemOpen
  \bibfield  {author} {\bibinfo {author} {\bibfnamefont {J.}~\bibnamefont {Miao}}, \bibinfo {author} {\bibfnamefont {C.-Y.}\ \bibnamefont {Hsieh}},\ and\ \bibinfo {author} {\bibfnamefont {S.-X.}\ \bibnamefont {Zhang}},\ }\bibfield  {title} {\bibinfo {title} {Neural-network-encoded variational quantum algorithms},\ }\href {https://journals.aps.org/prapplied/abstract/10.1103/PhysRevApplied.21.014053} {\bibfield  {journal} {\bibinfo  {journal} {Physical Review Applied}\ }\textbf {\bibinfo {volume} {21}},\ \bibinfo {pages} {014053} (\bibinfo {year} {2024})}\BibitemShut {NoStop}%
\bibitem [{\citenamefont {Rad}\ \emph {et~al.}(2022)\citenamefont {Rad}, \citenamefont {Seif},\ and\ \citenamefont {Linke}}]{rad2022surviving}%
  \BibitemOpen
  \bibfield  {author} {\bibinfo {author} {\bibfnamefont {A.}~\bibnamefont {Rad}}, \bibinfo {author} {\bibfnamefont {A.}~\bibnamefont {Seif}},\ and\ \bibinfo {author} {\bibfnamefont {N.~M.}\ \bibnamefont {Linke}},\ }\bibfield  {title} {\bibinfo {title} {Surviving the barren plateau in variational quantum circuits with {B}ayesian learning initialization},\ }\href {https://arxiv.org/abs/2203.02464} {\bibfield  {journal} {\bibinfo  {journal} {arXiv preprint arXiv:2203.02464}\ } (\bibinfo {year} {2022})}\BibitemShut {NoStop}%
\bibitem [{\citenamefont {Fa{\'{\i}}lde}\ \emph {et~al.}(2023)\citenamefont {Fa{\'{\i}}lde}, \citenamefont {Viqueira}, \citenamefont {Juane},\ and\ \citenamefont {G{\'{o}}mez}}]{Fa_lde_2023}%
  \BibitemOpen
  \bibfield  {author} {\bibinfo {author} {\bibfnamefont {D.}~\bibnamefont {Fa{\'{\i}}lde}}, \bibinfo {author} {\bibfnamefont {J.~D.}\ \bibnamefont {Viqueira}}, \bibinfo {author} {\bibfnamefont {M.~M.}\ \bibnamefont {Juane}},\ and\ \bibinfo {author} {\bibfnamefont {A.}~\bibnamefont {G{\'{o}}mez}},\ }\bibfield  {title} {\bibinfo {title} {Using differential evolution to avoid local minima in variational quantum algorithms},\ }\bibfield  {journal} {\bibinfo  {journal} {Scientific Reports}\ }\textbf {\bibinfo {volume} {13}},\ \href {https://doi.org/10.1038/s41598-023-43404-3} {10.1038/s41598-023-43404-3} (\bibinfo {year} {2023})\BibitemShut {NoStop}%
\bibitem [{\citenamefont {Kashif}\ and\ \citenamefont {Al-kuwari}(2023)}]{kashif2023resqnets}%
  \BibitemOpen
  \bibfield  {author} {\bibinfo {author} {\bibfnamefont {M.}~\bibnamefont {Kashif}}\ and\ \bibinfo {author} {\bibfnamefont {S.}~\bibnamefont {Al-kuwari}},\ }\href@noop {} {\bibinfo {title} {Resqnets: A residual approach for mitigating barren plateaus in quantum neural networks}} (\bibinfo {year} {2023}),\ \Eprint {https://arxiv.org/abs/2305.03527} {arXiv:2305.03527 [quant-ph]} \BibitemShut {NoStop}%
\bibitem [{\citenamefont {Barkoutsos}\ \emph {et~al.}(2020)\citenamefont {Barkoutsos}, \citenamefont {Nannicini}, \citenamefont {Robert}, \citenamefont {Tavernelli},\ and\ \citenamefont {Woerner}}]{barkoutsos2020improving}%
  \BibitemOpen
  \bibfield  {author} {\bibinfo {author} {\bibfnamefont {P.~K.}\ \bibnamefont {Barkoutsos}}, \bibinfo {author} {\bibfnamefont {G.}~\bibnamefont {Nannicini}}, \bibinfo {author} {\bibfnamefont {A.}~\bibnamefont {Robert}}, \bibinfo {author} {\bibfnamefont {I.}~\bibnamefont {Tavernelli}},\ and\ \bibinfo {author} {\bibfnamefont {S.}~\bibnamefont {Woerner}},\ }\bibfield  {title} {\bibinfo {title} {Improving variational quantum optimization using cvar},\ }\href {https://quantum-journal.org/papers/q-2020-04-20-256/} {\bibfield  {journal} {\bibinfo  {journal} {Quantum}\ }\textbf {\bibinfo {volume} {4}},\ \bibinfo {pages} {256} (\bibinfo {year} {2020})}\BibitemShut {NoStop}%
\bibitem [{\citenamefont {Stokes}\ \emph {et~al.}(2020)\citenamefont {Stokes}, \citenamefont {Izaac}, \citenamefont {Killoran},\ and\ \citenamefont {Carleo}}]{stokes2020quantum}%
  \BibitemOpen
  \bibfield  {author} {\bibinfo {author} {\bibfnamefont {J.}~\bibnamefont {Stokes}}, \bibinfo {author} {\bibfnamefont {J.}~\bibnamefont {Izaac}}, \bibinfo {author} {\bibfnamefont {N.}~\bibnamefont {Killoran}},\ and\ \bibinfo {author} {\bibfnamefont {G.}~\bibnamefont {Carleo}},\ }\bibfield  {title} {\bibinfo {title} {Quantum natural gradient},\ }\href {https://doi.org/10.22331/q-2020-05-25-269} {\bibfield  {journal} {\bibinfo  {journal} {Quantum}\ }\textbf {\bibinfo {volume} {4}},\ \bibinfo {pages} {269} (\bibinfo {year} {2020})}\BibitemShut {NoStop}%
\bibitem [{\citenamefont {Marrero}\ \emph {et~al.}(2021)\citenamefont {Marrero}, \citenamefont {Kieferov{\'a}},\ and\ \citenamefont {Wiebe}}]{marrero2020entanglement}%
  \BibitemOpen
  \bibfield  {author} {\bibinfo {author} {\bibfnamefont {C.~O.}\ \bibnamefont {Marrero}}, \bibinfo {author} {\bibfnamefont {M.}~\bibnamefont {Kieferov{\'a}}},\ and\ \bibinfo {author} {\bibfnamefont {N.}~\bibnamefont {Wiebe}},\ }\bibfield  {title} {\bibinfo {title} {Entanglement-induced barren plateaus},\ }\href {https://doi.org/10.1103/PRXQuantum.2.040316} {\bibfield  {journal} {\bibinfo  {journal} {PRX Quantum}\ }\textbf {\bibinfo {volume} {2}},\ \bibinfo {pages} {040316} (\bibinfo {year} {2021})}\BibitemShut {NoStop}%
\bibitem [{\citenamefont {Sharma}\ \emph {et~al.}(2022)\citenamefont {Sharma}, \citenamefont {Cerezo}, \citenamefont {Cincio},\ and\ \citenamefont {Coles}}]{sharma2020trainability}%
  \BibitemOpen
  \bibfield  {author} {\bibinfo {author} {\bibfnamefont {K.}~\bibnamefont {Sharma}}, \bibinfo {author} {\bibfnamefont {M.}~\bibnamefont {Cerezo}}, \bibinfo {author} {\bibfnamefont {L.}~\bibnamefont {Cincio}},\ and\ \bibinfo {author} {\bibfnamefont {P.~J.}\ \bibnamefont {Coles}},\ }\bibfield  {title} {\bibinfo {title} {Trainability of dissipative perceptron-based quantum neural networks},\ }\href {https://doi.org/10.1103/PhysRevLett.128.180505} {\bibfield  {journal} {\bibinfo  {journal} {Physical Review Letters}\ }\textbf {\bibinfo {volume} {128}},\ \bibinfo {pages} {180505} (\bibinfo {year} {2022})}\BibitemShut {NoStop}%
\bibitem [{\citenamefont {Patti}\ \emph {et~al.}(2021)\citenamefont {Patti}, \citenamefont {Najafi}, \citenamefont {Gao},\ and\ \citenamefont {Yelin}}]{patti2020entanglement}%
  \BibitemOpen
  \bibfield  {author} {\bibinfo {author} {\bibfnamefont {T.~L.}\ \bibnamefont {Patti}}, \bibinfo {author} {\bibfnamefont {K.}~\bibnamefont {Najafi}}, \bibinfo {author} {\bibfnamefont {X.}~\bibnamefont {Gao}},\ and\ \bibinfo {author} {\bibfnamefont {S.~F.}\ \bibnamefont {Yelin}},\ }\bibfield  {title} {\bibinfo {title} {Entanglement devised barren plateau mitigation},\ }\href {https://doi.org/10.1103/PhysRevResearch.3.033090} {\bibfield  {journal} {\bibinfo  {journal} {Physical Review Research}\ }\textbf {\bibinfo {volume} {3}},\ \bibinfo {pages} {033090} (\bibinfo {year} {2021})}\BibitemShut {NoStop}%
\bibitem [{\citenamefont {Wang}\ \emph {et~al.}(2021)\citenamefont {Wang}, \citenamefont {Fontana}, \citenamefont {Cerezo}, \citenamefont {Sharma}, \citenamefont {Sone}, \citenamefont {Cincio},\ and\ \citenamefont {Coles}}]{wang2020noise}%
  \BibitemOpen
  \bibfield  {author} {\bibinfo {author} {\bibfnamefont {S.}~\bibnamefont {Wang}}, \bibinfo {author} {\bibfnamefont {E.}~\bibnamefont {Fontana}}, \bibinfo {author} {\bibfnamefont {M.}~\bibnamefont {Cerezo}}, \bibinfo {author} {\bibfnamefont {K.}~\bibnamefont {Sharma}}, \bibinfo {author} {\bibfnamefont {A.}~\bibnamefont {Sone}}, \bibinfo {author} {\bibfnamefont {L.}~\bibnamefont {Cincio}},\ and\ \bibinfo {author} {\bibfnamefont {P.~J.}\ \bibnamefont {Coles}},\ }\bibfield  {title} {\bibinfo {title} {Noise-induced barren plateaus in variational quantum algorithms},\ }\href {https://doi.org/10.1038/s41467-021-27045-6} {\bibfield  {journal} {\bibinfo  {journal} {Nature Communications}\ }\textbf {\bibinfo {volume} {12}},\ \bibinfo {pages} {1} (\bibinfo {year} {2021})}\BibitemShut {NoStop}%
\bibitem [{\citenamefont {Holmes}\ \emph {et~al.}(2022)\citenamefont {Holmes}, \citenamefont {Sharma}, \citenamefont {Cerezo},\ and\ \citenamefont {Coles}}]{holmes2021connecting}%
  \BibitemOpen
  \bibfield  {author} {\bibinfo {author} {\bibfnamefont {Z.}~\bibnamefont {Holmes}}, \bibinfo {author} {\bibfnamefont {K.}~\bibnamefont {Sharma}}, \bibinfo {author} {\bibfnamefont {M.}~\bibnamefont {Cerezo}},\ and\ \bibinfo {author} {\bibfnamefont {P.~J.}\ \bibnamefont {Coles}},\ }\bibfield  {title} {\bibinfo {title} {Connecting ansatz expressibility to gradient magnitudes and barren plateaus},\ }\href {https://doi.org/10.1103/PRXQuantum.3.010313} {\bibfield  {journal} {\bibinfo  {journal} {PRX Quantum}\ }\textbf {\bibinfo {volume} {3}},\ \bibinfo {pages} {010313} (\bibinfo {year} {2022})}\BibitemShut {NoStop}%
\bibitem [{\citenamefont {Khatri}\ \emph {et~al.}(2019)\citenamefont {Khatri}, \citenamefont {LaRose}, \citenamefont {Poremba}, \citenamefont {Cincio}, \citenamefont {Sornborger},\ and\ \citenamefont {Coles}}]{khatri2019quantum}%
  \BibitemOpen
  \bibfield  {author} {\bibinfo {author} {\bibfnamefont {S.}~\bibnamefont {Khatri}}, \bibinfo {author} {\bibfnamefont {R.}~\bibnamefont {LaRose}}, \bibinfo {author} {\bibfnamefont {A.}~\bibnamefont {Poremba}}, \bibinfo {author} {\bibfnamefont {L.}~\bibnamefont {Cincio}}, \bibinfo {author} {\bibfnamefont {A.~T.}\ \bibnamefont {Sornborger}},\ and\ \bibinfo {author} {\bibfnamefont {P.~J.}\ \bibnamefont {Coles}},\ }\bibfield  {title} {\bibinfo {title} {Quantum-assisted quantum compiling},\ }\href {https://doi.org/10.22331/q-2019-05-13-140} {\bibfield  {journal} {\bibinfo  {journal} {Quantum}\ }\textbf {\bibinfo {volume} {3}},\ \bibinfo {pages} {140} (\bibinfo {year} {2019})}\BibitemShut {NoStop}%
\bibitem [{\citenamefont {Rudolph}\ \emph {et~al.}(2024)\citenamefont {Rudolph}, \citenamefont {Lerch}, \citenamefont {Thanasilp}, \citenamefont {Kiss}, \citenamefont {Shaya}, \citenamefont {Vallecorsa}, \citenamefont {Grossi},\ and\ \citenamefont {Holmes}}]{rudolph2023trainability}%
  \BibitemOpen
  \bibfield  {author} {\bibinfo {author} {\bibfnamefont {M.~S.}\ \bibnamefont {Rudolph}}, \bibinfo {author} {\bibfnamefont {S.}~\bibnamefont {Lerch}}, \bibinfo {author} {\bibfnamefont {S.}~\bibnamefont {Thanasilp}}, \bibinfo {author} {\bibfnamefont {O.}~\bibnamefont {Kiss}}, \bibinfo {author} {\bibfnamefont {O.}~\bibnamefont {Shaya}}, \bibinfo {author} {\bibfnamefont {S.}~\bibnamefont {Vallecorsa}}, \bibinfo {author} {\bibfnamefont {M.}~\bibnamefont {Grossi}},\ and\ \bibinfo {author} {\bibfnamefont {Z.}~\bibnamefont {Holmes}},\ }\bibfield  {title} {\bibinfo {title} {Trainability barriers and opportunities in quantum generative modeling},\ }\href {https://www.nature.com/articles/s41534-024-00902-0} {\bibfield  {journal} {\bibinfo  {journal} {npj Quantum Information}\ }\textbf {\bibinfo {volume} {10}},\ \bibinfo {pages} {116} (\bibinfo {year} {2024})}\BibitemShut {NoStop}%
\bibitem [{\citenamefont {Kieferova}\ \emph {et~al.}(2021)\citenamefont {Kieferova}, \citenamefont {Carlos},\ and\ \citenamefont {Wiebe}}]{kieferova2021quantum}%
  \BibitemOpen
  \bibfield  {author} {\bibinfo {author} {\bibfnamefont {M.}~\bibnamefont {Kieferova}}, \bibinfo {author} {\bibfnamefont {O.~M.}\ \bibnamefont {Carlos}},\ and\ \bibinfo {author} {\bibfnamefont {N.}~\bibnamefont {Wiebe}},\ }\bibfield  {title} {\bibinfo {title} {Quantum generative training using r\'{e}nyi divergences},\ }\href {https://arxiv.org/abs/2106.09567} {\bibfield  {journal} {\bibinfo  {journal} {arXiv preprint arXiv:2106.09567}\ } (\bibinfo {year} {2021})}\BibitemShut {NoStop}%
\bibitem [{\citenamefont {Thanaslip}\ \emph {et~al.}(2023)\citenamefont {Thanaslip}, \citenamefont {Wang}, \citenamefont {Nghiem}, \citenamefont {Coles},\ and\ \citenamefont {Cerezo}}]{thanaslip2021subtleties}%
  \BibitemOpen
  \bibfield  {author} {\bibinfo {author} {\bibfnamefont {S.}~\bibnamefont {Thanaslip}}, \bibinfo {author} {\bibfnamefont {S.}~\bibnamefont {Wang}}, \bibinfo {author} {\bibfnamefont {N.~A.}\ \bibnamefont {Nghiem}}, \bibinfo {author} {\bibfnamefont {P.~J.}\ \bibnamefont {Coles}},\ and\ \bibinfo {author} {\bibfnamefont {M.}~\bibnamefont {Cerezo}},\ }\bibfield  {title} {\bibinfo {title} {Subtleties in the trainability of quantum machine learning models},\ }\href {https://doi.org/10.1007/s42484-023-00103-6} {\bibfield  {journal} {\bibinfo  {journal} {Quantum Machine Intelligence}\ }\textbf {\bibinfo {volume} {5}},\ \bibinfo {pages} {21} (\bibinfo {year} {2023})}\BibitemShut {NoStop}%
\bibitem [{\citenamefont {Holmes}\ \emph {et~al.}(2021)\citenamefont {Holmes}, \citenamefont {Arrasmith}, \citenamefont {Yan}, \citenamefont {Coles}, \citenamefont {Albrecht},\ and\ \citenamefont {Sornborger}}]{holmes2021barren}%
  \BibitemOpen
  \bibfield  {author} {\bibinfo {author} {\bibfnamefont {Z.}~\bibnamefont {Holmes}}, \bibinfo {author} {\bibfnamefont {A.}~\bibnamefont {Arrasmith}}, \bibinfo {author} {\bibfnamefont {B.}~\bibnamefont {Yan}}, \bibinfo {author} {\bibfnamefont {P.~J.}\ \bibnamefont {Coles}}, \bibinfo {author} {\bibfnamefont {A.}~\bibnamefont {Albrecht}},\ and\ \bibinfo {author} {\bibfnamefont {A.~T.}\ \bibnamefont {Sornborger}},\ }\bibfield  {title} {\bibinfo {title} {Barren plateaus preclude learning scramblers},\ }\href {https://doi.org/10.1103/PhysRevLett.126.190501} {\bibfield  {journal} {\bibinfo  {journal} {Physical Review Letters}\ }\textbf {\bibinfo {volume} {126}},\ \bibinfo {pages} {190501} (\bibinfo {year} {2021})}\BibitemShut {NoStop}%
\bibitem [{\citenamefont {Mart{\'\i}n}\ \emph {et~al.}(2023)\citenamefont {Mart{\'\i}n}, \citenamefont {Plekhanov},\ and\ \citenamefont {Lubasch}}]{martin2022barren}%
  \BibitemOpen
  \bibfield  {author} {\bibinfo {author} {\bibfnamefont {E.~C.}\ \bibnamefont {Mart{\'\i}n}}, \bibinfo {author} {\bibfnamefont {K.}~\bibnamefont {Plekhanov}},\ and\ \bibinfo {author} {\bibfnamefont {M.}~\bibnamefont {Lubasch}},\ }\bibfield  {title} {\bibinfo {title} {Barren plateaus in quantum tensor network optimization},\ }\href {https://doi.org/10.22331/q-2023-04-13-974} {\bibfield  {journal} {\bibinfo  {journal} {Quantum}\ }\textbf {\bibinfo {volume} {7}},\ \bibinfo {pages} {974} (\bibinfo {year} {2023})}\BibitemShut {NoStop}%
\bibitem [{\citenamefont {Anschuetz}(2024)}]{anschuetz2024unified}%
  \BibitemOpen
  \bibfield  {author} {\bibinfo {author} {\bibfnamefont {E.~R.}\ \bibnamefont {Anschuetz}},\ }\bibfield  {title} {\bibinfo {title} {A unified theory of quantum neural network loss landscapes},\ }\href {https://arxiv.org/abs/2408.11901} {\bibfield  {journal} {\bibinfo  {journal} {arXiv preprint arXiv:2408.11901}\ } (\bibinfo {year} {2024})}\BibitemShut {NoStop}%
\bibitem [{\citenamefont {Crognaletti}\ \emph {et~al.}(2024)\citenamefont {Crognaletti}, \citenamefont {Grossi},\ and\ \citenamefont {Bassi}}]{crognaletti2024estimates}%
  \BibitemOpen
  \bibfield  {author} {\bibinfo {author} {\bibfnamefont {G.}~\bibnamefont {Crognaletti}}, \bibinfo {author} {\bibfnamefont {M.}~\bibnamefont {Grossi}},\ and\ \bibinfo {author} {\bibfnamefont {A.}~\bibnamefont {Bassi}},\ }\bibfield  {title} {\bibinfo {title} {Estimates of loss function concentration in noisy parametrized quantum circuits},\ }\href {https://arxiv.org/abs/2410.01893} {\bibfield  {journal} {\bibinfo  {journal} {arXiv preprint arXiv:2410.01893}\ } (\bibinfo {year} {2024})}\BibitemShut {NoStop}%
\bibitem [{\citenamefont {Mao}\ \emph {et~al.}(2024)\citenamefont {Mao}, \citenamefont {Tian},\ and\ \citenamefont {Sun}}]{mao2023barren}%
  \BibitemOpen
  \bibfield  {author} {\bibinfo {author} {\bibfnamefont {R.}~\bibnamefont {Mao}}, \bibinfo {author} {\bibfnamefont {G.}~\bibnamefont {Tian}},\ and\ \bibinfo {author} {\bibfnamefont {X.}~\bibnamefont {Sun}},\ }\bibfield  {title} {\bibinfo {title} {Towards determining the presence of barren plateaus in some chemically inspired variational quantum algorithms},\ }\href {https://www.nature.com/articles/s42005-024-01798-0} {\bibfield  {journal} {\bibinfo  {journal} {Communications Physics}\ }\textbf {\bibinfo {volume} {7}},\ \bibinfo {pages} {342} (\bibinfo {year} {2024})}\BibitemShut {NoStop}%
\bibitem [{\citenamefont {Teo}(2023)}]{Teo_2023}%
  \BibitemOpen
  \bibfield  {author} {\bibinfo {author} {\bibfnamefont {Y.~S.}\ \bibnamefont {Teo}},\ }\bibfield  {title} {\bibinfo {title} {Optimized numerical gradient and hessian estimation for variational quantum algorithms},\ }\bibfield  {journal} {\bibinfo  {journal} {Physical Review A}\ }\textbf {\bibinfo {volume} {107}},\ \href {https://doi.org/10.1103/physreva.107.042421} {10.1103/physreva.107.042421} (\bibinfo {year} {2023})\BibitemShut {NoStop}%
\bibitem [{\citenamefont {Fujii}\ and\ \citenamefont {Nakajima}(2017)}]{fujii2017harnessing}%
  \BibitemOpen
  \bibfield  {author} {\bibinfo {author} {\bibfnamefont {K.}~\bibnamefont {Fujii}}\ and\ \bibinfo {author} {\bibfnamefont {K.}~\bibnamefont {Nakajima}},\ }\bibfield  {title} {\bibinfo {title} {Harnessing disordered-ensemble quantum dynamics for machine learning},\ }\href {https://doi.org/10.1103/PhysRevApplied.8.024030} {\bibfield  {journal} {\bibinfo  {journal} {Physical Review Applied}\ }\textbf {\bibinfo {volume} {8}},\ \bibinfo {pages} {024030} (\bibinfo {year} {2017})}\BibitemShut {NoStop}%
\bibitem [{\citenamefont {Nakajima}\ \emph {et~al.}(2019)\citenamefont {Nakajima}, \citenamefont {Fujii}, \citenamefont {Negoro}, \citenamefont {Mitarai},\ and\ \citenamefont {Kitagawa}}]{nakajima2019boosting}%
  \BibitemOpen
  \bibfield  {author} {\bibinfo {author} {\bibfnamefont {K.}~\bibnamefont {Nakajima}}, \bibinfo {author} {\bibfnamefont {K.}~\bibnamefont {Fujii}}, \bibinfo {author} {\bibfnamefont {M.}~\bibnamefont {Negoro}}, \bibinfo {author} {\bibfnamefont {K.}~\bibnamefont {Mitarai}},\ and\ \bibinfo {author} {\bibfnamefont {M.}~\bibnamefont {Kitagawa}},\ }\bibfield  {title} {\bibinfo {title} {Boosting computational power through spatial multiplexing in quantum reservoir computing},\ }\href {https://doi.org/10.1103/PhysRevApplied.11.034021} {\bibfield  {journal} {\bibinfo  {journal} {Physical Review Applied}\ }\textbf {\bibinfo {volume} {11}},\ \bibinfo {pages} {034021} (\bibinfo {year} {2019})}\BibitemShut {NoStop}%
\bibitem [{\citenamefont {Mujal}\ \emph {et~al.}(2021)\citenamefont {Mujal}, \citenamefont {Mart{\'{\i} }nez-Pe{\~{n}}a}, \citenamefont {Nokkala}, \citenamefont {Garc{\'{\i}}a-Beni}, \citenamefont {Giorgi}, \citenamefont {Soriano},\ and\ \citenamefont {Zambrini}}]{mujal2021opportunities}%
  \BibitemOpen
  \bibfield  {author} {\bibinfo {author} {\bibfnamefont {P.}~\bibnamefont {Mujal}}, \bibinfo {author} {\bibfnamefont {R.}~\bibnamefont {Mart{\'{\i} }nez-Pe{\~{n}}a}}, \bibinfo {author} {\bibfnamefont {J.}~\bibnamefont {Nokkala}}, \bibinfo {author} {\bibfnamefont {J.}~\bibnamefont {Garc{\'{\i}}a-Beni}}, \bibinfo {author} {\bibfnamefont {G.~L.}\ \bibnamefont {Giorgi}}, \bibinfo {author} {\bibfnamefont {M.~C.}\ \bibnamefont {Soriano}},\ and\ \bibinfo {author} {\bibfnamefont {R.}~\bibnamefont {Zambrini}},\ }\bibfield  {title} {\bibinfo {title} {Opportunities in quantum reservoir computing and extreme learning machines},\ }\href {https://doi.org/10.1002/qute.202100027} {\bibfield  {journal} {\bibinfo  {journal} {Advanced Quantum Technologies}\ }\textbf {\bibinfo {volume} {4}},\ \bibinfo {pages} {2100027} (\bibinfo {year} {2021})}\BibitemShut {NoStop}%
\bibitem [{\citenamefont {Ghosh}\ \emph {et~al.}(2020)\citenamefont {Ghosh}, \citenamefont {Opala}, \citenamefont {Matuszewski}, \citenamefont {Paterek},\ and\ \citenamefont {Liew}}]{ghosh2020reconstructing}%
  \BibitemOpen
  \bibfield  {author} {\bibinfo {author} {\bibfnamefont {S.}~\bibnamefont {Ghosh}}, \bibinfo {author} {\bibfnamefont {A.}~\bibnamefont {Opala}}, \bibinfo {author} {\bibfnamefont {M.}~\bibnamefont {Matuszewski}}, \bibinfo {author} {\bibfnamefont {T.}~\bibnamefont {Paterek}},\ and\ \bibinfo {author} {\bibfnamefont {T.~C.}\ \bibnamefont {Liew}},\ }\bibfield  {title} {\bibinfo {title} {Reconstructing quantum states with quantum reservoir networks},\ }\href {https://doi.org/10.1109/TNNLS.2020.3009716} {\bibfield  {journal} {\bibinfo  {journal} {IEEE Transactions on Neural Networks and Learning Systems}\ }\textbf {\bibinfo {volume} {32}},\ \bibinfo {pages} {3148} (\bibinfo {year} {2020})}\BibitemShut {NoStop}%
\bibitem [{\citenamefont {Hu}\ \emph {et~al.}(2023)\citenamefont {Hu}, \citenamefont {Angelatos}, \citenamefont {Khan}, \citenamefont {Vives}, \citenamefont {T{\"{u}}reci}, \citenamefont {Bello}, \citenamefont {Rowlands}, \citenamefont {Ribeill},\ and\ \citenamefont {Tureci}}]{hu2023tackling}%
  \BibitemOpen
  \bibfield  {author} {\bibinfo {author} {\bibfnamefont {F.}~\bibnamefont {Hu}}, \bibinfo {author} {\bibfnamefont {G.}~\bibnamefont {Angelatos}}, \bibinfo {author} {\bibfnamefont {S.~A.}\ \bibnamefont {Khan}}, \bibinfo {author} {\bibfnamefont {M.}~\bibnamefont {Vives}}, \bibinfo {author} {\bibfnamefont {E.}~\bibnamefont {T{\"{u}}reci}}, \bibinfo {author} {\bibfnamefont {L.}~\bibnamefont {Bello}}, \bibinfo {author} {\bibfnamefont {G.~E.}\ \bibnamefont {Rowlands}}, \bibinfo {author} {\bibfnamefont {G.~J.}\ \bibnamefont {Ribeill}},\ and\ \bibinfo {author} {\bibfnamefont {H.~E.}\ \bibnamefont {Tureci}},\ }\bibfield  {title} {\bibinfo {title} {Tackling sampling noise in physical systems for machine learning applications: Fundamental limits and eigentasks},\ }\href {https://doi.org/10.1103/PhysRevX.13.041020} {\bibfield  {journal} {\bibinfo  {journal} {Physical Review X}\ }\textbf {\bibinfo {volume} {13}},\ \bibinfo {pages} {041020} (\bibinfo {year} {2023})}\BibitemShut {NoStop}%
\bibitem [{\citenamefont {Schuld}\ \emph {et~al.}(2019)\citenamefont {Schuld}, \citenamefont {Bergholm}, \citenamefont {Gogolin}, \citenamefont {Izaac},\ and\ \citenamefont {Killoran}}]{schuld2019evaluating}%
  \BibitemOpen
  \bibfield  {author} {\bibinfo {author} {\bibfnamefont {M.}~\bibnamefont {Schuld}}, \bibinfo {author} {\bibfnamefont {V.}~\bibnamefont {Bergholm}}, \bibinfo {author} {\bibfnamefont {C.}~\bibnamefont {Gogolin}}, \bibinfo {author} {\bibfnamefont {J.}~\bibnamefont {Izaac}},\ and\ \bibinfo {author} {\bibfnamefont {N.}~\bibnamefont {Killoran}},\ }\bibfield  {title} {\bibinfo {title} {Evaluating analytic gradients on quantum hardware},\ }\href {https://doi.org/10.1103/PhysRevA.99.032331} {\bibfield  {journal} {\bibinfo  {journal} {Physical Review A}\ }\textbf {\bibinfo {volume} {99}},\ \bibinfo {pages} {032331} (\bibinfo {year} {2019})}\BibitemShut {NoStop}%
\bibitem [{\citenamefont {Singh}\ \emph {et~al.}(2023)\citenamefont {Singh}, \citenamefont {Majumder},\ and\ \citenamefont {Mishra}}]{singh2023benchmarking}%
  \BibitemOpen
  \bibfield  {author} {\bibinfo {author} {\bibfnamefont {H.}~\bibnamefont {Singh}}, \bibinfo {author} {\bibfnamefont {S.}~\bibnamefont {Majumder}},\ and\ \bibinfo {author} {\bibfnamefont {S.}~\bibnamefont {Mishra}},\ }\bibfield  {title} {\bibinfo {title} {Benchmarking of different optimizers in the variational quantum algorithms for applications in quantum chemistry},\ }\href {https://pubs.aip.org/aip/jcp/article-abstract/159/4/044117/2904876/Benchmarking-of-different-optimizers-in-the?redirectedFrom=fulltext} {\bibfield  {journal} {\bibinfo  {journal} {The Journal of Chemical Physics}\ }\textbf {\bibinfo {volume} {159}} (\bibinfo {year} {2023})}\BibitemShut {NoStop}%
\bibitem [{\citenamefont {Rudolph}\ \emph {et~al.}(2021)\citenamefont {Rudolph}, \citenamefont {Sim}, \citenamefont {Raza}, \citenamefont {Stechly}, \citenamefont {McClean}, \citenamefont {Anschuetz}, \citenamefont {Serrano},\ and\ \citenamefont {Perdomo-Ortiz}}]{rudolph2021orqviz}%
  \BibitemOpen
  \bibfield  {author} {\bibinfo {author} {\bibfnamefont {M.~S.}\ \bibnamefont {Rudolph}}, \bibinfo {author} {\bibfnamefont {S.}~\bibnamefont {Sim}}, \bibinfo {author} {\bibfnamefont {A.}~\bibnamefont {Raza}}, \bibinfo {author} {\bibfnamefont {M.}~\bibnamefont {Stechly}}, \bibinfo {author} {\bibfnamefont {J.~R.}\ \bibnamefont {McClean}}, \bibinfo {author} {\bibfnamefont {E.~R.}\ \bibnamefont {Anschuetz}}, \bibinfo {author} {\bibfnamefont {L.}~\bibnamefont {Serrano}},\ and\ \bibinfo {author} {\bibfnamefont {A.}~\bibnamefont {Perdomo-Ortiz}},\ }\bibfield  {title} {\bibinfo {title} {Orqviz: visualizing high-dimensional landscapes in variational quantum algorithms},\ }\href {https://doi.org/10.48550/arXiv.2111.04695} {\bibfield  {journal} {\bibinfo  {journal} {arXiv preprint arXiv:2111.04695}\ } (\bibinfo {year} {2021})}\BibitemShut {NoStop}%
\bibitem [{\citenamefont {Wierichs}\ \emph {et~al.}(2020)\citenamefont {Wierichs}, \citenamefont {Gogolin},\ and\ \citenamefont {Kastoryano}}]{wierichs2020avoiding}%
  \BibitemOpen
  \bibfield  {author} {\bibinfo {author} {\bibfnamefont {D.}~\bibnamefont {Wierichs}}, \bibinfo {author} {\bibfnamefont {C.}~\bibnamefont {Gogolin}},\ and\ \bibinfo {author} {\bibfnamefont {M.}~\bibnamefont {Kastoryano}},\ }\bibfield  {title} {\bibinfo {title} {Avoiding local minima in variational quantum eigensolvers with the natural gradient optimizer},\ }\href {https://doi.org/10.1103/PhysRevA.104.032401} {\bibfield  {journal} {\bibinfo  {journal} {Physical Review Research}\ }\textbf {\bibinfo {volume} {2}},\ \bibinfo {pages} {043246} (\bibinfo {year} {2020})}\BibitemShut {NoStop}%
\bibitem [{\citenamefont {Haug}\ \emph {et~al.}(2021)\citenamefont {Haug}, \citenamefont {Bharti},\ and\ \citenamefont {Kim}}]{Haug_2021}%
  \BibitemOpen
  \bibfield  {author} {\bibinfo {author} {\bibfnamefont {T.}~\bibnamefont {Haug}}, \bibinfo {author} {\bibfnamefont {K.}~\bibnamefont {Bharti}},\ and\ \bibinfo {author} {\bibfnamefont {M.}~\bibnamefont {Kim}},\ }\bibfield  {title} {\bibinfo {title} {Capacity and quantum geometry of parametrized quantum circuits},\ }\bibfield  {journal} {\bibinfo  {journal} {{PRX} Quantum}\ }\textbf {\bibinfo {volume} {2}},\ \href {https://doi.org/10.1103/prxquantum.2.040309} {10.1103/prxquantum.2.040309} (\bibinfo {year} {2021})\BibitemShut {NoStop}%
\bibitem [{\citenamefont {Huang}\ \emph {et~al.}(2020)\citenamefont {Huang}, \citenamefont {Kueng},\ and\ \citenamefont {Preskill}}]{huang2020predicting}%
  \BibitemOpen
  \bibfield  {author} {\bibinfo {author} {\bibfnamefont {H.-Y.}\ \bibnamefont {Huang}}, \bibinfo {author} {\bibfnamefont {R.}~\bibnamefont {Kueng}},\ and\ \bibinfo {author} {\bibfnamefont {J.}~\bibnamefont {Preskill}},\ }\bibfield  {title} {\bibinfo {title} {Predicting many properties of a quantum system from very few measurements},\ }\href {https://doi.org/10.1038/s41567-020-0932-7} {\bibfield  {journal} {\bibinfo  {journal} {Nature Physics}\ }\textbf {\bibinfo {volume} {16}},\ \bibinfo {pages} {1050} (\bibinfo {year} {2020})}\BibitemShut {NoStop}%
\bibitem [{\citenamefont {Elben}\ \emph {et~al.}(2022)\citenamefont {Elben}, \citenamefont {Flammia}, \citenamefont {Huang}, \citenamefont {Kueng}, \citenamefont {Preskill}, \citenamefont {Vermersch},\ and\ \citenamefont {Zoller}}]{elben2022randomized}%
  \BibitemOpen
  \bibfield  {author} {\bibinfo {author} {\bibfnamefont {A.}~\bibnamefont {Elben}}, \bibinfo {author} {\bibfnamefont {S.~T.}\ \bibnamefont {Flammia}}, \bibinfo {author} {\bibfnamefont {H.-Y.}\ \bibnamefont {Huang}}, \bibinfo {author} {\bibfnamefont {R.}~\bibnamefont {Kueng}}, \bibinfo {author} {\bibfnamefont {J.}~\bibnamefont {Preskill}}, \bibinfo {author} {\bibfnamefont {B.}~\bibnamefont {Vermersch}},\ and\ \bibinfo {author} {\bibfnamefont {P.}~\bibnamefont {Zoller}},\ }\bibfield  {title} {\bibinfo {title} {The randomized measurement toolbox},\ }\bibfield  {journal} {\bibinfo  {journal} {Nature Review Physics}\ }\href {https://doi.org/10.1038/s42254-022-00535-2} {10.1038/s42254-022-00535-2} (\bibinfo {year} {2022})\BibitemShut {NoStop}%
\bibitem [{\citenamefont {{Huang}}\ \emph {et~al.}(2022)\citenamefont {{Huang}}, \citenamefont {{Broughton}}, \citenamefont {{Cotler}}, \citenamefont {{Chen}}, \citenamefont {{Li}}, \citenamefont {{Mohseni}}, \citenamefont {{Neven}}, \citenamefont {{Babbush}}, \citenamefont {{Kueng}}, \citenamefont {{Preskill}},\ and\ \citenamefont {{McClean}}}]{huang2021quantum}%
  \BibitemOpen
  \bibfield  {author} {\bibinfo {author} {\bibfnamefont {H.-Y.}\ \bibnamefont {{Huang}}}, \bibinfo {author} {\bibfnamefont {M.}~\bibnamefont {{Broughton}}}, \bibinfo {author} {\bibfnamefont {J.}~\bibnamefont {{Cotler}}}, \bibinfo {author} {\bibfnamefont {S.}~\bibnamefont {{Chen}}}, \bibinfo {author} {\bibfnamefont {J.}~\bibnamefont {{Li}}}, \bibinfo {author} {\bibfnamefont {M.}~\bibnamefont {{Mohseni}}}, \bibinfo {author} {\bibfnamefont {H.}~\bibnamefont {{Neven}}}, \bibinfo {author} {\bibfnamefont {R.}~\bibnamefont {{Babbush}}}, \bibinfo {author} {\bibfnamefont {R.}~\bibnamefont {{Kueng}}}, \bibinfo {author} {\bibfnamefont {J.}~\bibnamefont {{Preskill}}},\ and\ \bibinfo {author} {\bibfnamefont {J.~R.}\ \bibnamefont {{McClean}}},\ }\bibfield  {title} {\bibinfo {title} {Quantum advantage in learning from experiments},\ }\href {https://doi.org/10.1126/science.abn7293} {\bibfield  {journal} {\bibinfo  {journal} {Science}\ }\textbf {\bibinfo {volume} {376}},\ \bibinfo {pages} {1182} (\bibinfo {year}
  {2022})}\BibitemShut {NoStop}%
\bibitem [{\citenamefont {Cerezo}\ and\ \citenamefont {Coles}(2021)}]{cerezo2020impact}%
  \BibitemOpen
  \bibfield  {author} {\bibinfo {author} {\bibfnamefont {M.}~\bibnamefont {Cerezo}}\ and\ \bibinfo {author} {\bibfnamefont {P.~J.}\ \bibnamefont {Coles}},\ }\bibfield  {title} {\bibinfo {title} {Higher order derivatives of quantum neural networks with barren plateaus},\ }\href {https://doi.org/10.1088/2058-9565/abf51a} {\bibfield  {journal} {\bibinfo  {journal} {Quantum Science and Technology}\ }\textbf {\bibinfo {volume} {6}},\ \bibinfo {pages} {035006} (\bibinfo {year} {2021})}\BibitemShut {NoStop}%
\end{thebibliography}%
\onecolumngrid
\newpage

\appendix
{\huge \noindent \textbf{Appendix}}

\bigskip

\part{}
\parttoc 

\bigskip

\section{Hypothesis Testing}

Here, we provide preliminaries on binary hypothesis testing, a necessary tool for proving the statistical indistinguishability of measurement outcome probability distributions, as well as of post-processing applied to those outcomes. We present theoretical statements concerning the success of hypothesis testing in the single-sample case (Appendix~\ref{app:prelim-one-sample}) and the multi-sample case (Appendix~\ref{app:prelim-many-samples}). Finally, in Appendix~\ref{app:prelim-stat-indis}, we rigorously define the notion of statistical indistinguishability.

\subsection{One sample}\label{app:prelim-one-sample}
\begin{lemma}[]
    Consider two probability distributions $\P$ and $\P^\prime$ over some finite set $\mathcal{I}$. Suppose we are given a single sample $\SC$ drawn from either $\P$ or $\P^\prime$ with equal probability. We have the following two hypotheses:
    \begin{itemize}
        \item Null Hypothesis $ \mathcal{H}_0 $: $\SC$ is drawn from $\P$.
        \item Alternative hypothesis $ \mathcal{H}_1 $: $\SC$ is drawn from $\P^\prime$.
    \end{itemize}
    The probability of correctly deciding the true hypothesis is given by 
    \begin{equation}
        \text{Pr}[`` {\text{right decision between } \mathcal{H}_0 \, \text{and} \, \mathcal{H}_1} "] = \frac{1}{2} + \frac{\|\P - \P^\prime\|_1}{4} \;\;,
    \end{equation}
    where $\|\P - \P^\prime\|_1 = \sum_{s \in \mathcal{I}} |p(s) - p^\prime(s)|$ is the 1-norm between the two distributions.
\label{OneSampleHypothesisTestingLemma}    
\end{lemma}
\begin{proof}
    For two given distributions, denote $\AC$ as the largest subset such that $p(s) > p^\prime(s)$ for all $s\in \mathcal{A}$~\footnote{if the two distributions are identical, this $\AC$ is simply an empty set and the probability of making the right decision is simply $1/2$.}. The optimal test is to choose that the given sample $\SC$ is drawn from $\P$ (null hypothesis) if it falls in the region i.e., $\SC \in \mathcal{A}$ and guess $\P^\prime $ (alternative hypothesis), otherwise. The probability of choosing the correct hypothesis can be expressed as 
    \begin{align}
        \text{Pr}[`` {\text{right decision between } \mathcal{H}_0 \, \text{and} \, \mathcal{H}_1}"] &= \text{Pr}(\SC \in \mathcal{A}| \SC \sim \P) \text{Pr}(\SC \sim \P) + \text{Pr}(\SC \notin \mathcal{A}| \SC \sim \P^\prime) \text{Pr}(\SC \sim \P^\prime) \\
        &= \frac{1}{2} \big[\text{Pr}(\SC \in \mathcal{A}| \SC \sim \P) + \text{Pr}(\SC \notin \mathcal{A}| \SC \sim \P^\prime)\big] \\
        & = \frac{1}{2} \Bigg[\sum_{s \in \mathcal{A}} p(s)+\sum_{s \notin \mathcal{A}} p^\prime(s)\Bigg], \label{eq:proof-1sample-01}
    \end{align}
    where the second equality is due to the sample being equally likely to be drawn from either $ \P$ or $ \P^\prime$. In the last equality, we use the fact that the given sample is from $\P$, the probability that this sample takes any value within the region $ \mathcal{A}$ is simple $\sum_{s \in \mathcal{A}} p(s) $, and similarly for $ s \notin \mathcal{A}$. 

    The 1-norm between $\P$ and $\P^\prime$ is the 1-norm between the probability vectors, and it can be written as 
    \begin{align}
            \| \P - \P^\prime \|_1 &= \sum_{s \in \mathcal{I}} |p(s) - p^\prime(s)| \\
            & = \sum_{s \in \mathcal{A}} \big(p(s) - p^\prime(s)\big) + \sum_{s \notin \mathcal{A}} \big(p^\prime(s) - p(s)\big) \;\;,
    \end{align}
    where the second equality is due to the definition of the region $\mathcal{A}$. Lastly, we notice that 
    \begin{align}
        \frac{2 + \| \P - \P^\prime \|_1}{2} & = \frac{1}{2} \bigg( \sum_{s \in \mathcal{I}} p(s) + \sum_{s \in \mathcal{I}} p^\prime(s)  + \| \P - \P^\prime \|_1 \bigg) \\
        & = \sum_{s \in \mathcal{A}} p(s) + \sum_{s \notin \mathcal{A}} p^\prime(s)\;\;,
    \end{align}
    which, upon substituting back to Eq.~\eqref{eq:proof-1sample-01} completes the proof.
\end{proof}

\subsection{Many samples}\label{app:prelim-many-samples}

We proceed to the scenario with $N$ many samples where all samples are drawn from either $\P$ or $\P^\prime$. The goal remains similar in the sense that given this set of samples we have to make a decision which of the two distributions the samples are drawn from. To see how Lemma~\ref{OneSampleHypothesisTestingLemma} applies in this setting, we consider product distributions of $\P^{\otimes N}$ and $\P^{\prime \otimes N}$, where a \textit{single} sample from the product distribution corresponds to $N$ samples from $\P$ and $\P^{\prime}$ respectively:
\begin{definition}[Product distribution]
    Consider the probability distribution $\P$ over some finite set $\IC = \{ s_1,s_2,...,s_{|\IC|}\}$ such that
    \begin{align}
        \P = \big( \, p(s_1), p(s_2), \dots, p(s_{|\IC|})\, \big) \;\;,
    \end{align}
    where $p(s_i)$ is the probability of obtaining $s_i$. The $N$-fold product distribution of $\P$ is defined as $\P^{\otimes N}$ where $\ot$ is a tensor product~\footnote{For $A = (a_1,a_2)$ and $B=(b_1,b_2)$, we have $A\ot B = (a_1 b_1, a_1b_2, a_2b_1,a_2b_2)$.} i.e.,
    \begin{align}
        \P^{\ot N} & = \big( \, p(s_1),  \dots, p(s_{|\IC|})\, \big) \ot  \big( \, p(s_1),  \dots, p(s_{|\IC|})\, \big) \ot ... \ot  \big( \, p(s_1),  \dots, p(s_{|\IC|})\, \big) \\
        & = \big( p^N(s_1), p^{N-1}(s_1)p(s_2),...,p^N(s_{|\IC|}) \big) \;.
    \end{align}
\end{definition}

We now present Lemma~\ref{lemma:1-norm-bound}, which provides a bound on the one-norm between two distributions. We then proceed to Proposition~\ref{ManySampleHypothesisTestingLemma}, which establishes an upper bound on the success probability of the binary hypothesis testing task with many samples.
\begin{lemma}[]\label{lemma:1-norm-bound}
    Consider the probability distributions $ \P_i $ and  $ \P^\prime_i $ over some finite set $ \mathcal{I}_i $ for $ i \in \{ 1, 2, \dots, N\} $. Then, the 1-norm between the product distributions $ \P = \bigotimes_{i = 1}^{N} \P_i $ and $ \P^\prime = \bigotimes_{i = 1}^{N} \P^\prime_i $ over the finite set $ \bigotimes_{i =1}^{N} \mathcal{I}_i $ can be upper bounded as 
    \begin{equation}
        \| \P - \P^\prime \|_1 \leq \sum_{i = 1 }^{N} \| \P_i - \P^\prime_i \|_1 \;\;.
    \end{equation}
\label{ManySampleNormLemma}    
\end{lemma}
\begin{proof}
    The 1-norm between the two product distributions can be bounded as 
    \begin{align}
            \| \P - \P^\prime \|_1 &=  \left\| \bigotimes_{i = 1}^{N} \P_i - \bigotimes_{i = 1}^{N} \P^\prime_i \right\|_1 \\
            & = \left\| \bigotimes_{i = 1}^{N} \P_i - \P^\prime_1 \otimes \left[\bigotimes_{i = 2}^{N} \P_i\right] + \P^\prime_1 \otimes \left[\bigotimes_{i = 2}^{N} \P_i\right] - [\P^\prime_1 \otimes \P^\prime_2] \otimes \left[\bigotimes_{i = 3}^{N} \P_i\right] + [\P^\prime_1 \otimes \P^\prime_2] \otimes \left[\bigotimes_{i = 3}^{N} \P_i\right] + \dots \right. \nonumber \\
            & \;\;\;\;\;\; \left.- \left[\bigotimes_{i = 1}^{N-1} \P^\prime_i\right] \otimes \P_N  + \left[\bigotimes_{i = 1}^{N-1} \P^\prime_i\right] \otimes \P_N - \bigotimes_{i = 1}^{N} \P^\prime_i \right\|_1 \\
            & \leq \left\| \bigotimes_{i = 1}^{N} \P_i - \P^\prime_1 \otimes \left[ \bigotimes_{i = 2}^{N} \P_i \right]\right\|_1 + \left\| \P^\prime_1 \otimes \left[ \bigotimes_{i = 2}^{N} \P_i \right] - [\P^\prime_1 \otimes \P^\prime_2] \otimes \left[ \bigotimes_{i = 3}^{N} \P_i \right] \right\|_1 +  \nonumber \\
            & \;\;\;\;\;\;  \left\| [\P^\prime_1 \otimes \P^\prime_2] \otimes \left[ \bigotimes_{i = 3}^{N} \P_i \right] - [\P^\prime_1 \otimes \P^\prime_2 \otimes \P^\prime_3] \otimes \left[ \bigotimes_{i = 4}^{N} \P_i \right] \right\|_1 + \dots + \left\|\P_N \otimes \left[ \bigotimes_{i = 1}^{N-1} \P^\prime_i \right] - \bigotimes_{i = 1}^{N} \P^\prime_i \right\|_1 \; ,
    \end{align}
    where in the second equality, we have added and subtracted terms, the first inequality is simply the triangle inequality. We further proceed the bound as follows:
    \begin{align}
             \| \P - \P^\prime \|_1& \leq 
            \| \P_1 - \P^\prime_1\|_1 \, \left\|\bigotimes_{i = 2}^{N} \P_i \right\|_1 + 
            \| \P^\prime_1\|_1 \,\| \P_2 - \P^\prime_2\|_1 \, \left\|\bigotimes_{i = 3}^{N} \P_i \right\|_1 + \dots + 
            \left\|\bigotimes_{i = 1}^{N-1} \P^\prime_i \right\|_1 \, \| \P_N - \P^\prime_N\|_1
            \\
            & = \| \P_1 - \P^\prime_1\|_1 \, \prod_{i = 2}^{N}\| \P_i \|_1 + \| \P^\prime_1\|_1 \,\| \P_2 - \P^\prime_2\|_1 \, \prod_{i = 3}^{N}\| \P_i \|_1 + \dots + \prod_{i = 1}^{N-1}\| \P^\prime_i \|_1 \, \| \P_N - \P^\prime_N\|_1 \\
            & = \sum_{i = 1}^{N} \| \P_i - \P^\prime_i \|_1 \;\;,
    \end{align}
    where we use the following properties that the 1-norm of each individual distribution sums up to $ 1 $ i.e., $ \| \P_i \|_1 = \| \P^\prime_i \|_1 = 1 $. This completes the proof.
\end{proof}

\medskip

\begin{proposition}[]
    Consider the probability distributions $ \P = \P_0^{\otimes N} $ and  $ \P^\prime = \P_0^{\prime \otimes N} $ over the finite set $ \mathcal{I}^{\otimes N}_0 $. Suppose we are given $N$ samples (denoted as $\SC$) drawn from either $\P_0$ or $\P^\prime_0$ with equal probabilities. We have the following two hypotheses:
    \begin{itemize}
        \item Null Hypothesis $ \mathcal{H}_0 $: $\SC$ is drawn from $\P$.
        \item Alternative hypothesis $ \mathcal{H}_1 $: $\SC$ is drawn from $\P^\prime$.
    \end{itemize}The probability of correctly choosing the hypothesis is upper bounded as 
    \begin{equation}
        \text{Pr}[`` {\text{right decision between } \mathcal{H}_0 \, \text{and} \, \mathcal{H}_1} "] \leq \frac{1}{2} + \frac{N \|\P_0 - \P^\prime_0\|_1}{4}.
    \end{equation}
\label{ManySampleHypothesisTestingLemma}
\end{proposition}
\begin{proof}
    We use Lemma \ref{OneSampleHypothesisTestingLemma} to find the probability of making a correct hypothesis as 
    \begin{equation}
        \text{Pr}[`` {\text{right decision between } \mathcal{H}_0 \, \text{and} \, \mathcal{H}_1}"] = \frac{1}{2} + \frac{ \|\P - \P^\prime\|_1}{4}.
    \end{equation}
    Then according to Lemma \ref{ManySampleNormLemma}, we have $\|\P - \P^\prime\|_1 \leq N \|\P_0 - \P^\prime_0\|_1$. So we can find an upper bound for the probability of the right hypothesis as
    \begin{align}
        \text{Pr}[``{\text{right decision between } \mathcal{H}_0 \, \text{and} \, \mathcal{H}_1}"] & = \frac{1}{2} + \frac{ \|\P_0^{\otimes N} - \P_0^{\prime \otimes N}\|_1}{4} \\ 
        & \leq \frac{1}{2} + \frac{N\|\P_0 - \P^\prime_0\|_1}{4} \;\;.
    \end{align}
\end{proof}

\subsection{Statistical indistinguishability}\label{app:prelim-stat-indis}

We define statistical indistinguishability based on the success probability of a binary hypothesis testing task. Furthermore, we introduce the notion of statistical indistinguishability at two levels. Namely, Definition~\ref{StatisticalIndistinguishabilityDef} concerns the level of distributions or samples, while Definition~\ref{StatisticalIndistinguishabilityofOutputsDef} addresses the level of outputs.

\begin{definition}[$\varepsilon$-statistical indistinguishability of distributions]
    Two probability distributions $\P$ and $\P'$ are $\varepsilon$-statistically indistinguishable with $N$ samples if a binary hypothesis test cannot be passed with probability at least $0.5 +$$\varepsilon$. That is, given a set of $N$ samples $\SC$ drawn from either $\P$ or $\P'$ (with an equal probability), consider the following hypotheses
    \begin{itemize}
        \item Null hypothesis $\HC_0$: $\SC$ is drawn from $\P$\,,
        \item Alternative hypothesis $\HC_1$: $\SC$ is drawn from $\P'$ \,,
    \end{itemize}
    where $\P$ and $\P'$ are statistically indistinguishable (with $N$ samples) if for any algorithm the probability of correctly identifying the correct hypothesis satisfies:
    \begin{align}
         {\rm Pr}[``{\rm right \; decision \; between \, } \HC_0 \, {\rm and} \, \HC_1"] \leq 0.5 + \varepsilon \;,
    \end{align}
    where $\varepsilon$ is some arbitrary small number chosen such that the upper bound is close to random guessing e.g., $\varepsilon = 0.01$.
\label{StatisticalIndistinguishabilityDef}
\end{definition}

\begin{definition}[$\varepsilon$-statistical indistinguishability of outputs]
    Consider a map $\Phi:\mathbb{R}^N\rightarrow \mathbb{R}^M$ (with $M$ being the dimension of the output) and two distributions $\P$ and $\P'$ which are $\varepsilon$-statistically indistinguishable under $N$ samples according to Definition~\ref{StatisticalIndistinguishabilityDef}. Draw $N$ respective samples from $\P$ and $\P'$, which we respectively denote as $\SC_{\P}$ and $\SC_{\P'}$.  We say that $\Phi(\SC_{\P})$ and $\Phi(\SC_{\P'})$ are $\varepsilon$-statistically indistinguishable outputs.
\label{StatisticalIndistinguishabilityofOutputsDef}
\end{definition}

\section{Indistinguishability from probability concentration}
\label{app:theorem1-proof}

In this Appendix, we provide the formal statements of the theoretical results in the main text together with their detailed proofs. In particular, we have:
\begin{itemize}
    \item The formal version of Theorem~\ref{thm:main-indistinguishable} is presented in Theorem~\ref{thm:main-indistinguishable-formal}.
    \item The formal version of Corollary~\ref{coro:no-post-process} is presented in Corollary~\ref{coro:indistinguishability-outputs-formal}.
    \item The formal version of Corollary~\ref{coro:random-walk} is presented in Corollary~\ref{coro:random-walk-formal}.
\end{itemize}

We begin with general theoretical results demonstrating that the practical consequence of outcome probability concentration manifests as statistical indistinguishability—both at the level of distributions/samples, as shown in Theorem~\ref{thm:main-indistinguishable}, and at the level of outputs, as shown in Corollary~\ref{coro:indistinguishability-outputs-formal}.

\begin{theorem}[Formal version of Theorem~\ref{thm:main-indistinguishable}]\label{thm:main-indistinguishable-formal}
    Consider a parametrized $n$-qubit state $\rhot$ and a POVM set $\povm = \{\pov_k\}_{k=1}^{|\povm|}$ with polynomial elements $|\povm| \in \OC(\poly(n))$. The associated outcome probability distribution $\P_{\alv}$ is of the form 
    \begin{align}
        \P_{\alv} =  \left( p_1(\alv), p_2(\alv), \cdots, p_{|M|}(\alv)  \right)\;,
    \end{align}
    with $p_k(\alv) = \Tr[\rhot M_k]$ as the probability of obtaining the measurement associated with the element $M_k$. Now, assume the exponential concentration of outcome probabilities as in Definition~\ref{def:probability-concentration-main}, with some $\beta \in \OC(\exp(-n))$, for all possible outcomes with the fixed probability distribution 
    \begin{align}
        \P_{\rm fixed} =  \left( \mu_1, \mu_2, \mu_3, \dots, \mu_{|\povm|} \right) \;,
    \end{align}
    where $\mu_k$ is the concentration point of $p_k(\alv)$. For any $\alv$, with high probability exponentially close to $1$, the distributions $\P_{\alv}$ and $\P_{\rm fixed}$ are statistically indistinguishable with polynomial samples $N \in \OC(\poly(n))$ according to Definition~\ref{StatisticalIndistinguishabilityofSamples}. 
    That is, let $\SC_N$ be polynomial-sized samples drawn from either $\P_{\alv}$ or $\P_{\rm fixed}$ (with equal probability) and consider the following hypotheses:
    \begin{itemize}
        \item Null Hypothesis $ \mathcal{H}_0 $: $\SC_N$ is drawn from $\P_{\alv}$
        \item Alternative hypothesis $ \mathcal{H}_1 $: $\SC_N$ is drawn from $\P_{\rm fixed}$
    \end{itemize}
    Then, with probability at least $1-\delta$ over the choice of $\alv$ such that $\delta \in \OC({\rm exp}(-n))$, the probability of correctly deciding the true hypothesis is given by 
    \begin{align}
        \text{Pr}[`` {\text{right decision between } \mathcal{H}_0 \, \text{and} \, \mathcal{H}_1}"] \leq \frac{1}{2} + \varepsilon \; ,
    \end{align}
    where $\varepsilon \in \OC(\exp(-n))$. More particularly, we have $\delta = |\povm|\sqrt{\beta}$ and $\varepsilon = \frac{N|\povm|\beta^{\rm1/4}}{4}$.
\label{StatisticalIndistinguishabilityofSamples}
\end{theorem}
\begin{proof}
Given an $N$-sized set of samples $\SC_N$, the probability of successfully making the decision can be upper bounded with Proposition \ref{ManySampleHypothesisTestingLemma} as
\begin{equation}
    \text{Pr}[`` {\text{right decision between } \mathcal{H}_0 \, \text{and} \, \mathcal{H}_1}"] \leq \frac{1}{2} + \frac{N \|\P_{\alv} - \P_{\rm fixed}\|_1}{4}.
\label{prguesscorrectinproof}
\end{equation}
To compute $ \|\P_{\alv} - \P_{\rm fixed}\|_1 $ we can use the definition of 1-norm between two probability distributions and write 
\begin{align}
    \|\P_{\alv} - \P_{\rm fixed}\|_1 = \sum_{k = 1}^{|\povm|} |{p_k(\alv) - \mu_k}|,
\label{norminproof}
\end{align}
where $|\povm|$ is the number of POVM operators. 

To further proceed, we use the definition of the exponential concentration of outcome probabilities as in Definition \ref{def:probability-concentration-main}. For each outcome probability, we have 
\begin{align}
    \Pr_{\alv}\left( |p_k(\alv) - \mu_k| \geq \delta' \right) \leq \frac{\beta}{\delta'^2} \;\;,\;\; \beta \in \OC\left( \frac{1}{b^n}\right) \;\;\;,
\label{ProbExpConcentrationinproof}
\end{align}
with some $b > 1$ and $\mu_k$ is some concentration point independent of $\alv$. 
We then choose $\delta' = \beta^{\frac{1}{4}}$ and invert the inequality of Eq. (\ref{ProbExpConcentrationinproof}), leading to 
\begin{align}
    \Pr_{\alv}\left( |p_k(\alv) - \mu_k| \leq \beta^{\frac{1}{4}} \right) \geq 1 - \sqrt{\beta} \;,\;\; \beta \in \OC\left( \frac{1}{b^n}\right) \;.
\label{ProbExpConcentrationinproof2}
\end{align}
That is, we have that each outcome probability is exponentially close to the concentration point with high probability, exponentially close to $1$. 

We now show that this is sufficient to imply exponential vanishing of the one-norm. In particular, this can be shown by using the union bound over all outcome probabilities. 
Let $E_k$ be the event that $ |p_k(\alv) - \mu_k| \leq \beta^{\frac{1}{4}} $. From Eq. (\ref{ProbExpConcentrationinproof2}), we have
\begin{align}
    \Pr_{\alv}\left( E_k \right) \geq 1 - \sqrt{\beta} \;,\;\; \beta \in \OC\left( \frac{1}{b^n}\right) \;.
\end{align}
Now, the probability that all $E_k$ occur can be bounded with the union bound as 
\begin{align}
    \Pr_{\alv}\left( \bigcap_{k=1}^{|\povm|} E_k \right) &= 1 - \Pr_{\alv}\left( \bigcup_{k=1}^{|\povm|} \bar{E}_k \right) \\
    & \geq 1 - \sum_{k=1}^{|\povm|}\Pr_{\alv}\left(\bar{E}_k \right) \\
    & \geq 1 - |\povm|\sqrt{\beta} \;,
\end{align}
where $\bar{E}_k$ is conjugate event of $E_k$, we use union bound in the second line and use $\Pr_{\alv}\left(\bar{E}_k \right) \leq \sqrt{\beta}$ by reversing the inequality in Eq.~\eqref{ProbExpConcentrationinproof2}. So, with probability at least $1 - |\povm|\sqrt{\beta} $ over the parameters, we have 
\begin{align}
    \|\P_{\alv} - \P_{\rm fixed}\|_1 = \sum_{k = 1}^{|\povm|} |{p_k(\alv) - \mu_k}| \leq |\povm| \beta^{\frac{1}{4}},
\label{norminproof2}
\end{align}
By putting all together in Eq.~\eqref{prguesscorrectinproof}, this leads to, with the probability at least $1 - |\povm| \sqrt{\beta}$, the probability of success is bounded as,
\begin{equation}
    \text{Pr}[`` {\text{right decision between } \mathcal{H}_0 \, \text{and} \, \mathcal{H}_1}"] \leq \frac{1}{2} + \frac{N |\povm| \beta^{\frac{1}{4}}}{4} = \frac{1}{2} + \varepsilon \;,  
\label{prguesscorrectinproof2}
\end{equation}
where $\varepsilon = \frac{N |\povm| \beta^{\frac{1}{4}}}{4} \in \OC(c'^{-n}) $ for some $c'>1$ since $N, |\povm| \in \OC(\poly(n))$. In addition, by denoting $\delta = |\povm| \sqrt{\beta}$, we also have $\delta \in \OC(\poly(n))$.
That is, with probability at least $1-\delta$, the distributions $\P_{\alv}$ and $\P_{\rm fixed}$ are statistically indistinguishable with polynomial samples according to Definition~\ref{StatisticalIndistinguishabilityDef}. This completes the proof.
\end{proof}

\begin{corollary}[No post-processing, formal]\label{coro:indistinguishability-outputs-formal}
     Consider a map $\Phi:\mathbb{R}^N\rightarrow \mathbb{R}^M$ (with $M$ being the dimension of the output) on the set of measurement outcomes. Under the same assumptions as in Theorem~\ref{StatisticalIndistinguishabilityofSamples}, consider two $N$-sized sets of samples $\SC_{N}(\alv)$ and $\SC_{N,{\rm fixed}}$ where samples in $\SC_{N}(\alv)$ (and $\SC_{N,{\rm fixed}}$) are drawn from $\P_{\alv}$ (and $\P_{\rm fixed}$). With probability at least $1-\delta$ such that $\delta \in \OC(\exp(-n))$, the outputs $ \Phi(\SC_{N}(\alv)) $ and $ \Phi(\SC_{N,{\rm fixed}})$ are statistically indistinguishable as in Definition~\ref{StatisticalIndistinguishabilityofOutputsDef}.    
\end{corollary}

\begin{proof}
From Theorem~\ref{StatisticalIndistinguishabilityofSamples}, the distributions $\P_{\alv}$ and $\P_{\rm fixed}$ are statistically indistinguishable with probability at least $1-\delta$ with $\delta \in \OC(\exp(-n))$. If the outputs from the processing map $ \Phi(\SC_{N}(\alv)) $ and $ \Phi(\SC_{N,{\rm fixed}})$ were distinguishable, this would imply a strategy to reliable distinguish $\P_{\alv}$ and $\P_{\rm fixed}$. Hence, by contradiction, it must not be possible to distinguish the outputs from the processing map on the samples $\SC_{N}(\alv)$ and $\SC_{N,{\rm fixed}}$. 
\end{proof}

\medskip

The implication of Corollary~\ref{coro:indistinguishability-outputs-formal} is that post-processing the obtained samples cannot overcome the limitations imposed by exponential concentration. This conclusion holds for arbitrary procedures. To illustrate a concrete example, we show that training the loss on a featureless landscape using a standard gradient-based approach results in a random walk over the landscape.

\medskip

\begin{corollary}[Random walk via gradient descent, formal]\label{coro:random-walk-formal}
Consider a parametrized quantum state $\rho(\thv)$ that depends on some trainable parameters $\thv$ and a loss function of the form $\LC(\thv) = \sum_{i=1}^{N_L} c_i \ell_i(\thv)$ where each $\ell_i(\thv) = \Tr[\rho(\thv)O_i]$ with some parametrized state $\rho(\thv)$, some Pauli operator $O_i$, and $N_\ell \in\OC(\poly(n))$. Further, consider $\rho(\thv))$ is generated from some parametrized circuit such that the parameter shift rule is applied. Training the loss with the standard gradient descent algorithm with a random initialization for polynomial training iterations using overall polynomial measurement shots is statisitcially indistinguishable from a random walk with high probability $1 - c$ for some $c \in \OC(\exp(-n))$. That is, for a given iteration, the updated parameters $\thv^{(\rm new)}$ for the next iteration are statistically indistinguishable from the update rule given as 
\begin{align}
    \thv^{\rm (new)} = \thv^{\rm (current)} + \rw \;\;,
\end{align}
where $\thv^{(\rm current)}$ are parameter values for the current iteration, and $\rw$ is a vector where each component is an instance of some parameter-independent random variable. In particular, the $k^{\rm th}$ component $[\rw]_k$ is of the form
\begin{align}
    [\rw]_k= - \frac{\eta}{2}\sum_{i=1}^{N_L} c_i\left[\sum_{j=1}^N\left(\frac{z_{ijk}}{N}\right) - \sum_{j=1}^N \left( \frac{z'_{ijk}}{N}\right) \right] \;,
\end{align}
where each individual $z_{ijk}, z'_{ijk}$ is a random variable which takes a value $+1$ or $-1$ with equal probability.

\label{untrainabilityproposition}
\end{corollary} 

\begin{proof}

To prove that the optimization trajectory behaves as a random walk with high probability, we proceed in two main steps: (1) show that each individual optimization step exhibits random walk behavior; and (2) demonstrate that, over the course of training, the entire optimization trajectory remains consistent with a random walk. This can be established by repeatedly applying Corollary~\ref{coro:indistinguishability-outputs-formal} to various relevant quantities and invoking the union bound.

\medskip

\paragraph*{\underline{(1).~Show that each individual optimization step exhibits random walk behavior.}} Consider a loss function of the form $\LC (\thv) = \sum_{i = 1}^{N_{L}}{c_i \ell_i}(\thv) = \sum_{i=1}^{N_L} c_i \Tr[\rho (\thv) O_i]$ where $N_\ell \in\OC(\poly(n))$ and the set of POVMs $\{ \povm^{(i)}\}_{i = 1}^{N_L}$. We denote 
\begin{align}\label{eq:proof-rw-estimated-loss}
    \widehat{\LC}(\thv) = \sum_{i=1}^{N_L} c_i \widehat{\ell}_i(\thv)
\end{align}
as the empirical estimate of the full loss function at a given parameter setting $\boldsymbol{\theta}$. To estimate the full loss, we are required to estimate each individual term $\widehat{\ell}_i(\thv)$ by performing POVM $\povm^{(i)}$ resulting in a set of measurement outcomes $\SC_N^{(i)}(\thv)$ consisting of polynomial measurement shots i.e., $N\in \OC(\poly(n))$. Therefore, estimating the full loss function at a single parameter point generally requires $N_L$ such sample sets, one for each POVM. 

To estimate the loss gradients, we need to repeatedly estimate the loss at different points on the landscape. For gradient descent optimization with the parameter shift rule with $N_p$ parameters $\thv = (\theta_1,\theta_2,...,\theta_{N_p})$ and learning rate $\eta$, the updated parameter given the current parameter values $\thv^{\rm (current)}$ can be described by the general procedure (-- similar to  Eq.~\eqref{eq:loss-gradient-evaluation} in Section~\ref{sec:procedure}) as:
\begin{align}\label{eq:proof-rw-grad}
  \Phi_{\PC}\left( \left\{ \SC_N^{(ijk)} (\alv_{ijk})\right\}_{i,j,k }\right)= \thv^{\rm(current)} - \frac{\eta}{2}  \sum_{k=1}^{N_p} \left[\widehat{\LC}\left(\thv^{\rm(current)} + \frac{\pi}{2}\hat{e}_k\right) -  \widehat{\LC}\left(\thv^{\rm(current)}  - \frac{\pi}{2}\hat{e}_k\right)\right]\hat{e}_k \;,
\end{align}
where $\hat{e}_k$ is the unit vector in the direction of the $k^{\rm th}$ parameter component. Here, $\widehat{\mathcal{L}}\left(\boldsymbol{\theta} \pm \frac{\pi}{2} \hat{e}_k\right)$ denotes the estimated loss in Eq.~\eqref{eq:proof-rw-estimated-loss} evaluated at the shifted parameter values $\boldsymbol{\theta} \pm \frac{\pi}{2} \hat{e}_k$. Furthermore, we have the collection $\left\{ \mathcal{S}_N^{(ijk)}(\boldsymbol{\alpha}_{ijk}) \right\}_{ijk}$, where each $\mathcal{S}_N^{(ijk)}(\boldsymbol{\alpha}_{ijk})$ corresponds to the measurement outcomes used to estimate an individual term in the full loss function at a specific parameter setting. The total number of such measurement sets is $N_{\ell} = 2 N_{\rm p} N_L \in\OC(\poly(n))$. Since the indexing notation may appear somewhat dense, let us unpack the meaning of each index:
\begin{itemize}
    \item The index $i \in \{1, \dots, N_L\}$ labels the individual term $\ell_i$ in the full loss function.
    \item The index $k \in \{1, \dots, N_{\rm p}\}$ denotes the component of the parameter vector $\boldsymbol{\theta}$ being shifted.
    \item The index $j \in \{1, 2\}$ specifies the direction of the parameter shift: $j = 1$ corresponds to a $+$ shift, and $j = 2$ corresponds to a $-$ shift.
\end{itemize}
\noindent That is, we have $\mathcal{S}_N^{(i1k)}(\boldsymbol{\alpha}_{i1k})$ corresponding to the estimation of $\ell_i(\boldsymbol{\theta} + \frac{\pi}{2} \hat{e}_k)$, and $\mathcal{S}_N^{(i2k)}(\boldsymbol{\alpha}_{i2k})$ for $\ell_i(\boldsymbol{\theta} - \frac{\pi}{2} \hat{e}_k)$.

Now, we can further re-express Eq.~\eqref{eq:proof-rw-grad} in terms of individual measurement outcomes as
\begin{align}\label{eq:proof-rw-updated-params-lambda}
      \Phi_{\PC}\left( \left\{ \SC_N^{(ijk)} (\alv_{ijk})\right\}_{i,j,k }\right)= \thv^{\rm(current)} - \frac{\eta}{2}  \sum_{k=1}^{N_p}\sum_{i=1}^{N_L}c_i \left[\sum_{q=1}^N\frac{\lambda_{i1kq}}{N} - \sum_{q=1}^N\frac{\lambda_{i2kq}}{N}  \right]\hat{e}_k \;,
\end{align}
where $\lambda_{ijkq}$ is the $q^{\rm th}$ measurement outcome from $\SC_N^{(ijk)} (\alv_{ijk})$, which takes a value $``+1"$ with probability $p^{(i)}_+ (\alv_{ijk}) = (1+ \ell_i(\alv_{ijk}))/2$, and a value $``-1"$ with probability $p^{(i)}_- (\alv_{ijk}) =(1- \ell_i(\alv_{ijk}))/2$.

\medskip

Next, we consider that each individual POVM exponentially concentrate according to Definition~\ref{def:probability-concentration-main} and hence each set of measurement outcomes $\SC_N^{(ijk)} (\alv_{ijk})$ is, with high probability exponentially close to $1$, statistically indistinguishable from another set $\SC_{N,{\rm fixed}}^{(ijk)}$ where each individual outcome takes a value $``+1"$ or $``-1"$ with equal probability, according to Theorem~\ref{thm:main-indistinguishable-formal}. Consequently, by invoking Corollary~\ref{coro:indistinguishability-outputs-formal}, each estimate $\widehat{\ell}_i(\alv_{ijk})$ is, with high probability exponentially close to $1$, indistinguishable from another $\alv_{ijk}$-independent random variable of the form
\begin{align}
    \ell_{\rm fixed}^{(ijk)} = \sum_{q} \frac{z_{ijkq}}{N} \;,
\end{align}
where each $z_{ijkq}$ takes a value $``+1"$ or $``-1"$ with an equal probability. 

Now, the key proof strategy is to show that when considering all outcome measurement sets $\{\SC_N^{(ijk)} (\alv_{ijk})\}_{ijk}$ of size $N_{\ell} = 2 N_pN_L \in \OC(\poly(n))$, the probability of the updated parameters in Eq.~\eqref{eq:proof-rw-updated-params-lambda} being statistically indistinguishable remains exponentially close to $1$. In order to do so, we can invoke the union bound. In particular, denote $A_{ijk}$ as an event that $\SC_N^{(ijk)} (\alv_{ijk})$ is statically indistinguishable from $\SC_{N,{\rm fixed}}^{(ijk)}$ which, from Theorem~\ref{thm:main-indistinguishable-formal}, happens with the probability 
\begin{align}\label{prob-of-A}
    {\rm Pr}_{\thv}(A_{ijk}) & \ge 1 - |\MC^{(i)}| \sqrt{\beta^{(i)}} \\
    & = 1- 2\sqrt{\beta^{(i)}} \;,
\end{align} 
where $\beta^{(i)}$ is some exponentially vanishing value associated with the exponential concentration of the POVM $\MC^{(i)}$ i.e., $\beta^{(i)} \in \OC(\exp(-n))$. Note that in our setting with Pauli operator measurements, each POVM has two elements, that is $|\povm^{(i)}| = 2 $ for all POVM measurements. 

We are ready to bound the probability of all $\{\SC_N^{(ijk)} (\alv_{ijk})\}_{ijk}$ being indistinguishable using the union bound, leading to
\begin{align}
    \Pr_{\thv}{\left(\bigcap_{i,j,k} A_{ijk}\right)} &= 1 - \Pr_{\thv}{\left(\bigcup_{i,j,k} \bar{A}_{ijk}\right)} \\
    & \geq 1 - \sum_{i,j,k} \Pr_{\thv}{\left(\bar{A}_{ijk}\right)} \\
    & \geq 1 - \sum_{i,j,k} 2 \sqrt{\beta^{(i)}} \\
    & \geq 1 - 2 N_\ell  \sqrt{\beta^{*}} \; \; ,
    \label{eq:proof-rw-prob-rw-1step}
\end{align}
where $\bar{A}_{ijk}$ is the conjugate event of $A_{ijk}$, in the first inequality we use the union bound, in the second inequality we use $\Pr_{\thv}{\left(\bar{A}_{ijk}\right)} \leq |\povm^{(ijk)}| \sqrt{\beta^{(ijk)}}$ by reversing the inequality in Eq.~\ref{prob-of-A}. To reach the last line, we denote $\beta^{*}$ is the maximum value of $\{\beta^{(i)}\}_{i = 1}^{N_L}$. Since $N_\ell \in \OC(\poly(n))$, we have that $\delta' =2N_l \sqrt{\beta^*} \in \OC(\exp(-n))$. That is, with the probability at least $1-\delta'$ such that $\delta' \in \OC(\exp(-n))$, all $\{\SC_N^{(ijk)} (\alv_{ijk})\}_{ijk}$ are indistinguishable from all $\{\SC_{N,{\rm fixed}}^{(ijk)} \}_{ijk}$. Following directly, with the same probability, the updated parameter $\Phi_{\PC}\left( \left\{ \SC_N^{(ijk)} (\alv_{ijk})\right\}_{i,j,k }\right)$ is statistically indistinguishable from 
\begin{align}\label{eq:proof-rw-final-rand}
    \Phi_{\PC}\left( \left\{ \SC_{N,{\rm fixed}}^{(ijk)}\right\}_{i,j,k }\right) = \thv^{\rm(current)} - \frac{\eta}{2}  \sum_{k=1}^{N_p}\sum_{i=1}^{N_L}c_i \left[\sum_{q=1}^N\frac{z_{i1kq}}{N} - \sum_{q=1}^N\frac{z_{i2kq}}{N}  \right]\hat{e}_k \;\;,
\end{align}
which means that the update is statistically indistinguishable from a parameter-independent random variable when determining the next step, effectively forming a random walk for a single training iteration with probability at least exponentially close to $1$. Note that the random walk $\rw$ introduced in Corollary~\ref{coro:random-walk-formal} can be directly identified from Eq.~\eqref{eq:proof-rw-final-rand}.

\medskip

\paragraph*{\underline{(2).~Demonstrate that, throughout training, the entire optimization trajectory remains consistent with a random walk.}} 
We know that for each training step, the probability of a random walk is bounded as shown in Eq.~\eqref{eq:proof-rw-prob-rw-1step}. To demonstrate that the entire training trajectory consisting of $N_{\mathrm{step}} \in \mathcal{O}(\mathrm{poly}(n))$ steps forms a random walk, we invoke the union bound. Denote $B_k$ be the event that an $k^{\rm th}$ training step resembles a random walk. We have that
\begin{align}
    \Pr_{\thv^{(1)}, \thv^{(2)}, \dots, \thv^{(N_{\rm step})}}{\left(\bigcap_{k = 1}^{N_{\rm step}} B_k\right)} &= 1 - \Pr_{\thv^{(1)}, \thv^{(2)}, \dots, \thv^{(N_{\rm step})}}{\left(\bigcup_{k = 1}^{N_{\rm step}} \bar{B}_k\right)} \\
    & \geq 1 - \sum_{k = 1}^{N_{\rm step}} \Pr_{\thv^{(k)}}{\left(\bar{B}_k\right)} \\
    & \geq 1 - \sum_{k = 1}^{N_{\rm step}} 2 N_{\ell} \sqrt{\beta^{*}} \\
    & = 1 - 2 N_{\rm step} N_{\ell} \sqrt{\beta^{*}} \; \;,
\end{align}
where we use union bound in the second line and use $\Pr_{\thv^{(k)}}{\left(\bar{B}_k\right)} \leq  2 N_{\ell} \sqrt{\beta^{*}}$ by reversing the inequality in Eq.~\eqref{eq:proof-rw-prob-rw-1step}. That is, with probability $1-\delta''$ such that $\delta'' = 2 N_{\rm step} N_{\ell} \sqrt{\beta^{*}} \in \OC(\exp(-n))$, the whole training trajectory is statistically indistinguishable from a random walk. This completes the proof.

\end{proof}

\section{Further details of numerical simulation}
\label{app:numerical-details}

Here, we describe the numerical setup in more detail. To specify the setup, we detail the parametrized circuit, the measurement operator, the initialization of the circuit parameters, the estimation map used to evaluate the loss function, and the optimization strategy.

In all methods, the quantum circuit consists of a single layer of single-qubit Pauli-$X$ rotations applied to all qubits,
\begin{equation}
    U(\thv) = \prod_{i = 1}^{n} e^{- i \theta_i X_i} \;,
    \label{unitary-for-numerics}
\end{equation}
where $X_i$ denotes the Pauli-$X$ operator acting on the $i^{\mathrm{th}}$ qubit. In addition, although the exact form may vary slightly, the observables are all global, acting non-trivially on $k$ qubits such that $k\in\Omega(n)$. This combination is well known to induce barren plateaus due to the global nature of the loss function~\cite{cerezo2020cost}.

\medskip

\paragraph*{\underline{Quantum natural gradient descent optimization.}}
Quantum natural gradient descent is a quantum analogue of natural gradient descent. In this method, the optimization direction in the loss landscape is determined with respect to the underlying quantum information geometry, which is characterized by the real part of the Quantum Geometric Tensor (QGT). As a result, this approach can yield more effective optimization steps than standard gradient descent. The parameters update rule, with learning rate $\eta$, is given by
\begin{equation}
    \thv^{(\rm new)} = \thv^{(\rm current)} - \eta  \, g^{+}(\thv^{(\rm current)}) \, \nabla\widehat{\LC}(\thv^{(\rm current)})\;,  
\end{equation}
where $g^{+}(\boldsymbol{\theta}^{(\mathrm{current})})$ denotes the pseudo-inverse of the Fubini–Study metric tensor $g(\boldsymbol{\theta})$, which can be expressed in terms of QGT $G(\boldsymbol{\theta})$, and $\nabla \widehat{\mathcal{L}}(\boldsymbol{\theta}^{(\mathrm{current})})$ represents the gradient of the loss function at that step. 
To obtain the Quantum Geometric Tensor (QGT), let us be general and consider a parameterized quantum circuit composed of $L$ layers of non-commuting parameterized unitaries (-- we will simplify to the circuit in Eq.~\eqref{unitary-for-numerics} later)
\begin{align}
    U(\thv) & = \prod_{i = 1}^{L} V_i(\thv) \, W_i \;,
\end{align}
where $W_i$ are non-parametrised unitaries and $V_i(\thv)$ are parametrised unitaries, with $n_l$ parameters at $l^{\rm th}$ layer, of the form
\begin{equation}
    V_l(\thv) = \prod_{i = 1}^{n_l} e^{-i \theta^{(l)}_i \sigma_i^{(l)}} \; ,
\end{equation}
where $\thv = \{ \thv^{(1)}, \thv^{(2)}  , \dots, \thv^{{(L)}}\}$ with $\thv^{(i)} = \{ \theta^{(i)}_1, \theta^{(i)}_2, \dots, \theta^{(i)}_{n_l}\}$, and $\{\sigma_i^{(l)}\}_{l}$ are a set of commuting generators such that $\left(\sigma_i^{(l)}\right)^2= \mathbb{1}$. Then, as shown in Ref.~\cite{stokes2020quantum}, the QGT can be approximated by its block-diagonalized form with $L$ blocks, and a $l^{\rm th}$ block is an $n_l \times n_l$ matrix whose elements are given by
\begin{equation}
    G_{ij}^{(l)}(\thv) = \langle \psi_l(\thv) | \sigma^{(l)}_{i} \, \sigma^{(l)}_{j}| \psi_l(\thv) \rangle - \langle \psi_l(\thv) | \sigma^{(l)}_{i}| \psi_l(\thv) \rangle \, \langle \psi_l(\thv) | \sigma^{(l)}_{j}| \psi_l(\thv) \rangle \;,
\end{equation}
where $|\psi_l(\thv)\rangle$ is the state of circuit at at $l^{\rm th}$ layer, i.e., 
\begin{equation}
    |\psi_l(\thv)\rangle = \prod_{i = 1}^{l} V_i(\thv) \, W_i |\psi_0\rangle \;,
\end{equation}
with an initial state $|\psi_0\rangle$. At the end, the block-diagonal QGT leads to a block-diagonal Fubini-Study metric tensor, 
\begin{equation}
    g_{ij}^{(l)}(\thv) = \Re{[G_{ij}^{(l)}(\thv)]} \;.
\end{equation}

\medskip

In our numerical setup, the parameterized quantum circuit is simplified to the form given in Eq.~\ref{unitary-for-numerics}, and the measurement operator is taken to be a global Pauli-$Z$ observable. In other words, in the loss function of the form $\LC(\thv) = \Tr[H \rho(\thv)]$, the Hamiltonian is given by
\begin{equation}
    H = \bigotimes_{i = 1}^{n} Z_i \;.
\end{equation}
The circuit parameters are sampled from a uniform random distribution. The loss function is estimated using the empirical average as the post-processing map. For the optimization strategy, the pseudo-inverse of the Fubini–Study metric is computed to be $g^{+}(\thv) = 4 \, \mathbb{1}$, making the update rule equivalent to a gradient descent method with an effective learning rate of $4\eta$.

\medskip

\paragraph*{\underline{Sample-based CVaR optimization.}}
For this method, the measurement operator is a linear combination of four global Pauli-$Z$ terms, and the Hamiltonian is given by
\begin{equation}
    H = c_1 \, \left[\bigotimes_{i = 1}^{n} Z_i\right]  + c_2 \, \left[\bigotimes_{i = 1}^{n-1} Z_i\right] \otimes \mathbb{1}  +
    c_3 \, \left[\bigotimes_{i = 1}^{n-2} Z_i\right] \otimes \mathbb{1} \otimes \mathbb{1} +
    c_4 \, \left[\bigotimes_{i = 1}^{n-3} Z_i \right] \otimes \mathbb{1} \otimes \mathbb{1} \otimes \mathbb{1} \;,
    \label{CVaR-Hamiltonian}
\end{equation}
where $\{c_i\}_{i=1}^4$ are some constant coefficients and $H$ has $16$ distinct eigenvalues $\{ \povv_i\}_{i =1}^{16}$. 

The initial parameters of the circuit are sampled from a uniform random distribution, similar to the previous method. The loss function is estimated using the CVaR processing map:
\begin{equation}
    \widehat{\ell}_i(\thv) = \Phi_{\rm CVaR}^{(\gamma)} \left( \SC_N(\thv) \right) = \frac{1}{\lceil{\gamma N}\rceil} \sum_{j=0}^{\lceil{\gamma N}\rceil} \oc_j \;,
\end{equation}
where each $\oc_j$ is one of the eigenvalues $\{ \povv_i \}_i$. For the optimization strategy, we use the gradient descent method with learning rate $\eta$.

\medskip

\paragraph*{\underline{Classical neural network assisted initialization.}}
In this method, the key difference lies in the initialization of the circuit parameters. Unlike the previous approaches, where the circuit parameters are sampled from random distributions, here the parameters are initialized using the output of a classical neural network. The output vector of this neural network is used to initialize the parameters of the quantum circuit. 

The measurement operator is a global Pauli-$Z$ measurement, and the loss function is estimated using the empirical average processing map. Finally, the optimization strategy involves applying gradient descent to train the parameters of the classical neural network, rather than those of the quantum circuit.

\medskip

\paragraph*{\underline{Re-scaled Parameter Shift (RPS) rule optimization.}}
In this method~\cite{Teo_2023}, it is claimed that using a scaled parameter-shift rule to compute gradients can improve the training of loss functions that suffer from barren plateaus. This corresponds to using the scaled parameter-shift rule with a scaling factor that minimizes the mean squared error in the estimation of both the loss function and its gradients.

In our numerical results, the measurement operator is a global Pauli-$Z$ operator, the initial circuit parameters are sampled from the uniform distribution, and the loss function is estimated using the empirical average processing map. For the optimization strategy, we employ the scaled parameter-shift rule to compute gradients within the gradient descent algorithm, using learning rate $\eta$. The scaling factor is given by 
\begin{equation}
    \lambda = \frac{d N}{2d^2 + N d - 2} \;,
\end{equation}
where $N$ is the number of measurement shots and $d = 2^n$ is the Hilbert space dimension.

\section{Examples for a probability distribution with exponentially many elements}
\label{app:counter-example}
\subsection{Counter example}
We provide a counterexample of a distribution with exponentially many elements that exhibits exponential concentration according to Definition~\ref{def:probability-concentration-main}, yet remains distinguishable from its fixed distribution.

Consider two probability distributions $\P_{\alv}$ and $\P_{\rm fixed}$ over a finite set $\mathcal{X}$ with size $M$ such that $M\in \Omega(\exp(n))$. In other words, we have two  distributions $\P_{\alv} = (p_1(\alv), p_2(\alv), \dots, p_M(\alv))$ and $\P_{\rm fixed} = (\mu_1, \mu_2, \dots, \mu_M)$ where the categories are named as $\{ 1, 2, \dots, M\}$, and the probability of each category $i$ is given with $p_i(\alv)$ and $\mu_i$ respectively. 
We then perform a binary hypothesis test. That is, suppose we are given a set of $N$ samples $\mathcal{S}$ all drawn from either $\P_{\alv}$ or $\P_{\rm fixed}$ with equal probabilities. We aim to distinguish between the following two hypotheses:
\begin{itemize}
    \item Null hypothesis $\mathcal{H}_0$: $\mathcal{S} \sim \P_{\alv}$
    \item Alternative hypothesis $\mathcal{H}_1$: $\mathcal{S} \sim \P_{\rm fixed}$
\end{itemize}
According to Proposition \ref{ManySampleHypothesisTestingLemma}, the probability of correctly identifying the true hypothesis is upper bounded by
\begin{equation}
        \text{Pr}[``{\text{right decision between } \mathcal{H}_0 \, \text{and} \, \mathcal{H}_1} "] \leq \frac{1}{2} + \frac{N \|\P_{\alv} - \P_{\rm fixed}\|_1}{4}.
\label{correctanswer_probability_loose_bound}
\end{equation}

Now, let us assume that the probability distribution $\P_{\alv}$ takes the following form:
\begin{equation}
    p_i(\alv) = \begin{cases}
        \mu_i + f(\alv) & ;\; i \in \text{odd} \\
        \mu_i - f(\alv) & ;\; i \in \text{even}
    \end{cases}
\end{equation}
where $\mu_i = \mathbb{E}_{\alv}[p_i(\alv)]$. 
The one-norm distance between $\P_{\alv}$ and $\P_{\rm fixed}$ can be written as 
\begin{equation}
    \| \P_{\alv} - \P_{\rm fixed} \|_1 = \sum_{i = 1}^{M} | p_i(\alv) - \mu_i | = \sum_{i = 1}^{M} |f(\alv)| = M |f(\alv)|.
\end{equation}
Considering the exponential concentration of all outcome probabilities, we can derive the condition that must be satisfied by the function $f(\alv)$.  In particular, each outcome probability $p_i(\alv)$ exhibits concentration if
\begin{equation}
    \Var_{\alv}[p_i(\alv)] \leq \beta, 
\end{equation}
where $\beta \in \mathcal{O}(1/b^n)$ for some $b>1$. The variance of these probabilities can be expressed as a function of $f(\alv)$, and the corresponding condition on $f(\alv)$ is given by
\begin{equation}
    \Var_{\alv}[p_i(\alv)] = \Var_{\alv}[\mu_i \pm f(\alv)] = \Var_{\alv}[f(\alv)] = \mathbb{E}_{\alv}[f^2(\alv)] \leq \beta.
\end{equation}

To illustrate the counterexample, consider the specific setting where  $\P_{\rm fixed}$ is a uniform probability distribution, that is for all outcomes $\mathbb{E}_{\alv}[p_i(\alv)] = \frac{1}{M}$, and $f(\alv) = \frac{1}{M}$. In this case, the problem reduces to the hypothesis test between the following two probability distributions: $\P_{\alv} = (p_1 (\alv), p_2(\alv), \dots, p_M(\alv))$ with
\begin{equation}
    p_i(\alv) = \begin{cases}
        \frac{2}{M}  & ;\; i \in \text{odd} \\
        0 & ;\; i  \in \text{even} \;\;,
    \end{cases}
\end{equation}
and $\P_{\rm fixed} = (\mu_1, \mu_2, \dots, \mu_M)$ which is the uniform distribution with
\begin{equation}
    \mu_i = \frac{1}{M}\;\;, \; \;\forall m \;. 
\end{equation}
In this case, we have $\| \P_{\alv} - \P_{\rm fixed}\|_1 = 1$. Using Eq. (\ref{correctanswer_probability_loose_bound}), the probability of correctly distinguishing between the two hypotheses is upper-bounded by
\begin{align}
    \text{Pr}[`` {\text{right decision between } \mathcal{H}_0 \, \text{and} \, \mathcal{H}_1} "] \leq 1/2 + N/4 \;.
\end{align}
This is a loose bound and does not provide any meaningful constraint on the number of samples required to succeed in the hypothesis test.

However, we can provide a simple test for our hypotheses.  Consider the test $\mathcal{T}$, which chooses the null hypothesis ($\mathcal{T} = 0$) if all samples $s_i$ are odd, and chooses the alternative hypothesis $\mathcal{H}_1$ otherwise (i.e., if at least one even number appears among the samples). The probability of error in this test can be written as
\begin{align}
    p_{\text{error}} & = \text{Pr}[\mathcal{S} \sim \P_{\alv}^{\otimes N}] \, \text{Pr}[\mathcal{T} = 1 |\mathcal{S} \sim \P_{\alv}^{\otimes N}] + \text{Pr}[\mathcal{S} \sim \P_{\rm fixed}^{\otimes N}] \, \text{Pr}[\mathcal{T} = 0 |\mathcal{S} \sim \P_{\rm fixed}^{\otimes N}] \\
    & = \frac{1}{2} \, \text{Pr}[\mathcal{T} = 1 |\mathcal{S} \sim \P_{\alv}^{\otimes N}] + \frac{1}{2} \, \text{Pr}[\mathcal{T} = 0 |\mathcal{S} \sim \P_{\rm fixed}^{\otimes N}] \\
    & =  \frac{1}{2^{N+1}} \;\;,
\end{align}
where $\text{Pr}[\mathcal{T} = 0|\mathcal{S} \sim \P_{\rm fixed}^{\otimes N}]$ corresponds to the probability that all sampled values are odd (i.e., no even numbers appear) when the samples are drawn from the distribution $\P_{\rm fixed}$. On the other hand, $\text{Pr}[\mathcal{T} = 1 | \mathcal{S} \sim \P_{\alv}^{\otimes N}] = 0$ in our setting, since $\P_{\alv}$ assigns zero probability to all even outcomes. Here, one can observe that the error probability decays exponentially with respect to $N$, allowing the two distributions to be distinguished with high confidence using a realistic number of samples.

\subsection{Indistinguishable example}
We also present another example where the distribution remains indistinguishable, despite having exponentially many elements. To see this, consider the same setup, but now choose $f(\alv) = \frac{1}{M^2}$. In this case, the hypothesis testing problem is between the two probability distributions $\P_{\alv} = (p_i(\alv), p_2(\alv), \dots, p_M(\alv))$ with  
\begin{equation}
    p_i(\alv) = \begin{cases}
        \frac{1}{M} + \frac{1}{M^2} & ;\; i \in \text{odd} \\
        \frac{1}{M} - \frac{1}{M^2} & ;\; i \in \text{even}
    \end{cases}
\end{equation}
and $\P_{\rm fixed} = (\mu_1, \mu_2, \dots, \mu_M)$ which is the same uniform distribution with $\mu_i = 1/M$ for all. In this example the one-norm distance is $\|\P_{\alv} - \P_{\rm fixed}\|_1 = 1/M$. Following Eq. (\ref{correctanswer_probability_loose_bound}), the probability of correctly identifying the true hypothesis is upper-bounded as follows:
\begin{equation}
    \text{Pr}[``{\text{right decision between } \mathcal{H}_0 \, \text{and} \, \mathcal{H}_1} "] \leq \frac{1}{2} + \frac{N}{4M}.
\end{equation}
Here, the bound remains exponentially close to $\frac{1}{2}$ as $\frac{N}{4M} \in \mathcal{O}{\left(\frac{1}{2^n}\right)}$ with polynomial number of measurements. 

\end{document}